\documentclass[fontsize =12pt,titlepage]{article}

\usepackage[margin=0.8in]{geometry}
\usepackage[sort,numbers]{natbib}
\usepackage{microtype}
\usepackage{graphicx}
\usepackage{subfigure}
\usepackage{booktabs} 
\usepackage{bbm}
\usepackage{comment}
\usepackage{hyperref}

\usepackage{tabularx}
\usepackage{multirow}
\usepackage[font=small]{caption}
\usepackage{dblfloatfix}

\usepackage{float}
\usepackage{sidecap}
\usepackage{ccaption}
\usepackage[parfill]{parskip}

\usepackage{amsmath}
\usepackage{amssymb}
\usepackage{mathtools}
\usepackage{amsthm}
\newcommand*\diff{\mathop{}\!\mathrm{d}}
  
\usepackage[capitalize,noabbrev]{cleveref}
\usepackage{setspace}
\theoremstyle{plain}
\newtheorem{theorem}{Theorem}
\newtheorem{proposition}{Proposition}

\newtheorem{lemma}{Lemma}
\newtheorem{corollary}{Corollary}
\theoremstyle{definition}
\newtheorem{definition}{Definition}

\theoremstyle{remark}

\usepackage{enumerate}
\usepackage{enumitem}
\usepackage{authblk}
\usepackage{xcolor}
\usepackage{caption}
\definecolor{subsectioncolor}{HTML}{8d3a16}

\usepackage{titlesec}

\titleformat{\subsubsection}
{\normalfont\itshape}
{\thesubsubsection.}{0.5em}{\underline}

\begin{document}
\title{\Large \textbf{Latent computing by biological neural networks: A dynamical systems framework}}

\author[1,2,3,4]{Fatih Dinc}
\author[1,5]{Marta Blanco-Pozo}
\author[6]{David Klindt}
\author[3]{Francisco Acosta}
\author[1,5]{\authorcr Yiqi Jiang}
\author[1]{Sadegh Ebrahimi}
\author[7]{Adam Shai}
\author[2,8,$\dag$]{\authorcr Hidenori Tanaka}
\author[9,10,$\dag$]{Peng Yuan}
\author[1,5,11,$\dag$]{Mark J. Schnitzer}
\author[3,$\dag$]{Nina Miolane}
\affil[1]{CNC Program, Stanford University, Stanford, USA}
\affil[2]{Phi Lab, NTT Research, Sunnyvale, USA}
\affil[3]{Geometric Intelligence Lab, UC Santa Barbara, Santa Barbara, USA}
\affil[4]{Kavli Institute for Theoretical Physics, UC Santa Barbara, Santa Barbara, USA}
\affil[5]{James Clark Center, Stanford University, Stanford, USA}
\affil[6]{Cold Spring Harbor Laboratory, USA}
\affil[7]{Simplex, Berkeley, USA}
\affil[8]{Center for Brain Science, Harvard University, Cambridge, USA}
\affil[9]{State Key Laboratory of Medical Neurobiology, Institute for Translational Brain Research, MOE Frontiers Center for Brain Science, Fudan University, China}
\affil[10]{Research Institute of Intelligent Complex System, Fudan University, China}
\affil[11]{Howard Hughes Medical Institute, USA}

\affil[$\dag$]{Supervisors}

\maketitle

\begin{abstract}
Although individual neurons and neural populations exhibit the phenomenon of representational drift, perceptual and behavioral outputs of many neural circuits can remain stable across time scales over which representational drift is substantial. These observations motivate a dynamical systems framework for neural network activity that focuses on the concept of \emph{latent processing units,} core elements for robust coding and computation embedded in collective neural dynamics. Our theoretical treatment of these latent processing units yields five key attributes of computing through neural network dynamics. First, neural computations that are low-dimensional can nevertheless generate high-dimensional neural dynamics. Second, the manifolds defined by neural dynamical trajectories exhibit an inherent coding redundancy as a direct consequence of the universal computing capabilities of the underlying dynamical system. Third, linear readouts or decoders of neural population activity can suffice to optimally subserve downstream circuits controlling behavioral outputs. Fourth, whereas recordings from thousands of neurons may suffice for near optimal decoding from instantaneous neural activity patterns, experimental access to millions of neurons may be necessary to predict neural ensemble dynamical trajectories across timescales of seconds. Fifth, despite the variable activity of single cells, neural networks can maintain stable representations of the variables computed by the latent processing units, thereby making computations robust to representational drift. Overall, our framework for latent computation provides an analytic description and empirically testable predictions regarding how large systems of neurons perform robust computations via their collective dynamics.
\end{abstract}

\section{Introduction}
Understanding the principles linking neural activities and information representation in the brain is a central goal of systems neuroscience research. Earlier studies identified fixed correspondence between specific neurons to sensory stimuli \cite{o1976place,fuster1971neuron}, suggesting that each behaviorally relevant variable is represented by individual neurons in the brain \cite{barlow1972single,eliasmith2003neural} (\textit{e.g.}, receptive fields in the cat's striate cortex \cite{hubel1959receptive}). However, recent large scale and longitudinal recordings (\cite{stevenson2011advances,urai2022large,langdon2023unifying,gao2015simplicity,kim2022fluorescence,syeda2024facemap,chen2024brain,stringer2019high,kim2016long}) revealed that information can be stably represented through the neuronal population, while individual neurons’ coding properties can change over time \cite{ziv_long-term_2013, deitch_representational_2021, ebrahimi2022emergent, marks2021stimulus, driscoll2017dynamic, rokni2007motor}. The latter phenomenon, called \textit{representational drift}, has challenged the single neuron doctrine in the brain, and called for a novel principle to explain neural coding.

Neural manifolds have since been hypothesized as the population-level substrates of neural coding (neural manifold hypothesis; \cite{mastrogiuseppe2018linking,yuste2015neuron,langdon2023unifying}). In this conjecture, a set of variables, hidden in the collective dynamics of the neuronal activities, are hypothesized to represent stable neural dynamics over time despite representational drift \cite{deitch_representational_2021}. To date, several dimensional reduction techniques have been developed to extract these coding dimensions \cite{schneider2023learnable,williams2018unsupervised,pandarinath2018inferring,depasquale2023centrality,zimnik2024identifying,duncker2021dynamics,versteeg2023expressive,jazayeri2021interpreting,glaser2020recurrent, cunningham2014dimensionality,bjerke2022understanding}, some of them integrating time-dependent information such as CEBRA \cite{schneider2023learnable} and TCA \cite{williams2018unsupervised}. These methods have provided a big step forward, revealing computational insights in the form of attractor structures \cite{nair2023approximate}, identifying inter-area communication channels in animal brains \cite{perich2021inferring,perich2020rethinking}, and uncovering preserved neural motives across animals \cite{safaie2023preserved}. Yet, performing a dimensional reduction of neural data does not provide a causal and physiologically motivated model of its time dynamics. Notwithstanding, it can be the first step towards a comprehensive theory of neural coding  (\textit{e.g.}, to describe how neural activity manifolds respond to novel perturbations).

Building on the neural manifold hypothesis, new promising approaches for making causal predictions under external perturbations rely on inferring dynamical systems from neural recordings (Fig. \ref{fig1}\textbf{A}). These include models with linear latent dynamics \cite{o2022direct,abbaspourazad2024dynamical}, recurrently switching linear dynamical systems (rSLDS) \cite{linderman2017bayesian,vinograd2024causal,hu2025modeling} and recurrent neural networks \cite{valente2022extracting,mastrogiuseppe2018linking,dubreuil2022role,beiran2021shaping,Langdon2025}. However, these approaches have their own shortcomings (See Table \ref{tabs1} for more details). Linear latent models have significantly restricted dynamics (\textit{e.g.}, linear dynamical systems cannot explain limit cycles), whereas for the others, low-dimensional latent dynamics lead to low-dimensional flat neural manifolds, contradicting empirical evidence from recent studies \cite{stringer2019high,manley2024simultaneous}. All in all, these approaches primarily offer phenomenological descriptions of neural recordings (\textit{e.g.}, identifying approximate line attractors within \emph{projected} neural activities \cite{vinograd2024causal,liu2024encoding}), rather than providing a comprehensive theory of neural computation. 

To bridge this fundamental gap, by building upon recent seminal work connecting structure and function in low-rank RNNs \cite{beiran2021shaping,dubreuil2022role,valente2022extracting,mastrogiuseppe2018linking}, we introduce latent computation framework (LCF), a rigorous theory of neural coding. Our framework is built on the assumption that neural coding is carried out in dynamical systems embedded in high-dimensional neural activities, \textit{i.e.}, ``latent processing units (LPUs).'' Unlike real-world objects that follow immutable physical laws, or central processing units (CPUs) that perform pre-defined logic operations, LPUs consist of dynamical systems designed to carry out computational tasks, which evolve according to a set of \emph{learnable} differential equations (parameterized by synaptic connections in a network). We start by illustrating how to extract LPUs from simulated recurrent neural networks (RNNs) performing behavioral tasks. Later, we demonstrate how LPUs explain several key phenomena identified in systems neuroscience literature, a feat not achieved by existing work modeling neural dynamics. Then, we develop a new method for simulating representational drift within task-trained RNNs. We use it to illustrate how LPU dynamics can remain stable even after almost all neurons in a network change their coding properties. Finally, we discuss how the new generation of large-scale interventional tools can allow testing our empirical predictions, eventually verifying or falsifying the putative use of LPUs in biological networks.

\section{Results} \label{sec:results}
\subsection{Latent processing units: From constrained units to universal computation}

LCF starts with three key elements (illustrated in Fig. \ref{fig1}\textbf{B}): 1) a dynamical model of observed neural activities, 2) LPUs, \textit{i.e.}, coding subspaces within the dynamical model that subserve behavior, and 3) falsifiable empirical predictions based on individual neurons' contributions to these subspaces. We define the latent dynamics through a generic encoding-embedding relationship (Fig. \ref{fig1}\textbf{B}; Definition \ref{def:embedding}):
\begin{subequations} \label{eq1}
\begin{align}
\kappa(t) &= \phi(r(t)), \\
\tau \dot r(t) &= -r(t) + \varphi(\kappa(t),u(t)),
\end{align}
\end{subequations}
where $r(t) \in \mathbb R^{N_{\rm rec}}$ represents the activities of $N_{\rm rec}$ recurrently connected neurons at a given time $t$, $\kappa(t) \in \mathbb R^K$ denotes latent variables with $K \ll N_{\rm rec}$, $\phi(.)$ and $\varphi(.)$ are encoding and embedding maps respectively, $u(t) \in \mathbb R^{N_{\rm in}}$ represents external inputs, and $\tau \in \mathbb R$ is the neuronal timescale. The key insight of this formulation is that the embedding map constrains neural activity derivatives, unlike traditional ``decoder'' formulation that map latent variables back to neural activities (as in models such as CEBRA \cite{schneider2023learnable}; \textbf{Methods}). Intuitively, one can consider the encoding-embedding relationship as the joint generator of the network ($\dot r(t)$) and the latent ($\dot \kappa(t)$) dynamics, with the pair of encoding-embedding maps ($\{\varphi,\phi\}$) defining the dynamical system equations (\textbf{Methods}). Therefore, the resulting model can provide causal predictions of time dynamics under outside perturbations (\emph{e.g.}, optogenetics). This relationship captures the low-rank RNN architecture studied in previous work \cite{mastrogiuseppe2018linking,dubreuil2022role,valente2022extracting,beiran2021shaping} (Table \ref{tabs1}), while opening the door for a broader class of biologically motivated RNNs (\textbf{Methods}).

We refer to the dynamical system defined by latent variables, $\kappa(t)$, as an LPU when two conditions are met (\textbf{Methods}; Definition \ref{def2}): i) $\dot \kappa(t)$ does not explicitly depend on $r(t)$, only implicitly through $\kappa(t) = \phi(r(t))$ (\textit{i.e.}, $\kappa(t)$ constitutes a latent dynamical system), and ii) the network parameters can be tuned to approximate any flow map $\dot \kappa(t) = g(\kappa(t),u(t))$. The latter property is known as universal approximation in the RNN literature \cite{schafer2006recurrent}, in which ``recurrent weights" ($W^{\rm rec}$) between artificial neurons can be trained to approximate time series data. With LPUs, we assume that latent dynamical systems are encoded through linear maps and embedded in neural activities nonlinearly. This design choice is motivated by three distinct reasons: i) physiological relevance, \textit{e.g.}, to allow linear readouts of internal states through motor/downstream neurons (see below for more); ii) theoretical sufficiency, \textit{i.e.}, this combination of the encoding-embedding relationship leads to closed latent dynamical systems, whereas nonlinear embedding enables universal computation (\textbf{Methods}); and iii) this asymmetric choice is also at the core of compressed sensing theory \cite{donoho2006compressed}, which has previously been applied to the study of sparse neural population codes \cite{ganguli2010short,ganguli2012compressed}.

For the rest of this work, all our theorems and propositions are proven on a broad class of biologically motivated RNN architectures, which have linear encoding and nonlinear embedding maps and are characterized by their ability to form multiple synapses between artificial neurons (\textbf{Methods}, Definition \ref{def3}). On the other hand, for clarity and demonstration purposes, we perform our experiments on a simple, yet representative, network architecture. Specifically, by combining the linear encoding with one of the most basic embedding maps—a weighted sum of latent variables followed by neuron-wise $\tanh(.)$ nonlinearity, we arrive at the leaky firing-rate RNN architecture \cite{masse2019circuit,dinc2023cornn}:
\begin{equation} \label{eq_lfrnn}
\tau \dot r(t) = -r(t) + \tanh\left( W^{\rm rec} r(t) + W^{\rm in} u(t) + b \right),
\end{equation}
where $W^{\rm rec} \in \mathbb R^{N_{\rm rec} \times N_{\rm rec}}$ are recurrent weights as alluded to above, $W^{\rm in} \in \mathbb R^{N_{\rm rec} \times N_{\rm in}}$ are input weights, $b \in \mathbb R^{N_{\rm rec}}$ are biases, and the rest are defined as before. This simple architecture (which is distinct from what has been studied in the previous low-rank RNN literature \cite{mastrogiuseppe2018linking,dubreuil2022role,beiran2021shaping,valente2022extracting}; See \textbf{Methods}) satisfies the universal approximation property (Theorem \ref{thm1}), which has two crucial implications: it enables the network to subserve complex behavior through linear readouts (Theorem \ref{thm2}), and allows latent dynamical systems to operate at different timescales compared to neural activities (Proposition \ref{prop1}; \textbf{Supplementary Note}). The latter is particularly relevant given the separation between neuronal ($O(ms)$) and behavioral ($O(s)$) timescales \cite{zheng2024unbearable}. We discuss all these concepts in detail below. 

\subsection{Extracting latent processing units from task-trained RNNs}

Now, with a simple example, we demonstrate how LCF can provide causal predictions and unique insights in the presence of outside interventions. To start with, in order to extract the LPUs in RNNs whose time evolution follows Eq. \eqref{eq_lfrnn}, we enforce a low-rank constraint on the weight matrix: 
\begin{equation}\label{eq_lowrank}
    W^{\rm rec} = \frac{1}{N_{\rm rec}} \sum_{p=1}^K m^{(p)} n^{(p)T}, 
\end{equation}
where $m^{(p)} \in \mathbb R^{N_{\rm rec}}$ and $n^{(p)} \in \mathbb R^{N_{\rm rec}}$ are referred to as ``embedding'' and ``encoding'' weights, respectively. In practice, even when the underlying weight matrix is not low-rank by design, it is often possible to approximate it with a low-rank counterpart for RNNs performing behavioral tasks \cite{valente2022extracting,dubreuil2022role}. (This approach is also how we primarily train low-rank RNNs in this work; \textbf{Methods}.) Then, the effectively low-rank nature of the recurrent weights allows defining latent variables as: 
\begin{equation} \label{eq_kappa}
    \kappa_p(t) = \frac{1}{N_{\rm rec}} n^{(p)T} r(t), \quad \text{for} \quad p=1,\ldots K.
\end{equation}
Since this relationship is a simple weighted sum, we refer to it as ``linear encoding.'' As alluded to above, combined with Theorem \ref{thm1}, the choice of linear encoding ensures that $\kappa(t)$ constitutes a LPU in this architecture. (Here, the latent variable definition in Eq. \eqref{eq_kappa} may seem similar to the existing low-rank RNN literature \cite{valente2022extracting,beiran2021shaping,dubreuil2022role}, but it differs fundamentally. See \emph{``Existing models in encoding-embedding framework''} in \textbf{Methods} and Table \ref{tabs1} for a detailed comparison).

To demonstrate how LPUs perform computations, we trained a rank-2 RNN on a 2-bit flip-flop task (Fig. \ref{fig1}\textbf{C-F}). This task requires the RNN to maintain four distinct internal states representing possible 2-bit combinations (00, 01, 10, 11). RNNs learn to transition between these states in response to the incoming $\pm 1$ pulses through two inputs (Fig. \ref{fig1}\textbf{C}). In a representative RNN, individual neurons exhibited state-dependent persistent activity (Fig. \ref{fig1}\textbf{D}), with four attractive fixed points forming in the two-dimensional LPU (Fig. \ref{fig1}\textbf{E}). While these attractors could also be identified from neural activities using established methods such as energy minimization \cite{sussillo2013opening}, LPUs uniquely allow analyzing network behavior beyond equilibrium states, \textit{e.g.}, latent state trajectories away from attractor points (Fig. \ref{fig1}\textbf{E}). Moreover, it is possible to quantify how latent dynamics would transform under synaptic modifications or external inputs. We demonstrated this capability by introducing an input to the first channel, which induced a bifurcation in the system dynamics (Fig. \ref{fig1}\textbf{F})—a qualitative transformation resulting in the disappearance of two of the four stable fixed-point attractors. 

\subsection{Designing large-scale RNNs with distinct latent and neural timescales}

Biological networks are capable of separating neural ($\tau \sim O(ms)$) and kinematic ($\sim O(s)$) time-scales orders of magnitude apart \cite{zheng2024unbearable}. For instance, biological networks may need to collect evidence for extended durations of time before committing to a final action, which would require the underlying LPU to operate at longer time scales than $\tau$. Fortunately, Proposition \ref{prop1} (\textbf{Methods}) guarantees that LPUs can operate in arbitrarily long time-scales. Yet, this theoretical result holds only in the limit of $N_{\rm rec} \to \infty$, which requires a deeper (empirical) look. 

Now, we study the limitations of LPUs supported by finite number of units (\textit{e.g.}, biological neurons) that operate at ultra-fast time-scales. To this end, we set out to simulate large-scale RNNs capable of solving flip-flop tasks. First, we attempted the traditional strategy of training RNNs using backpropagation through time (BPTT) (Fig. \ref{figs1}\textbf{A}). But, even with networks of 5,000 neurons, BPTT required substantially reduced learning rates to converge (Fig. \ref{figs1}\textbf{B}), making it impractical for training larger networks to store flip-flop states. Nonetheless, studying RNNs that we could successfully train (even though they had only $100$ neurons per network) revealed a key insight: the encoding ($n$) and embedding ($m$) weights for individual neurons were distributed as if sampled from a distribution with zero mean and a positive correlation (Fig. \ref{figs1}\textbf{C}). This observation, combined with a theoretical mean-field approach, allowed designing arbitrarily large RNNs with one-dimensional LPUs to store flip-flop states: We sampled the pair $\{n_i,m_i\}$ for the neuron $i$ from a data-informed probability distribution (Fig. \ref{fig1}\textbf{G}; \textbf{Methods}). 

Using the procedure outlined above, we tested the empirical limits of Proposition \ref{prop1} by designing LPUs operating at various time-scales ($O(ms)$ to $O(s)$) using ultra-fast neurons (with $\tau \sim 1ms$). As anticipated, LPUs with time-scales close to $\tau$ could be reliably constructed with as little as a thousand neurons (Fig. \ref{figs1}\textbf{D-E}). However, the same number of neurons failed to support LPUs operating at second-long time-scales (Fig. \ref{fig1}\textbf{H}). For this case, approximately one million neurons were needed (Fig. \ref{fig1}\textbf{I}). This finding raises up an important practical question for neuroscience research when recording the whole population is not an option, \textit{e.g.}, as in recordings from mammalian brains. Can LPUs be reconstructed using partial observations, \textit{i.e.}, without having recorded every neuron? Unfortunately, studying this question extensively is already a challenging endeavor by itself \cite{ayed2019learning,das2020systematic,qian2024partial}. Instead, in \textbf{Supplementary Note}, we present an analytical framework around a simple (and relevant) toy model of bistable dynamics. There, we study the linearized dynamics of one-dimensional LPUs around $\kappa \sim 0$ and examine how reconstruction errors of the growth rates (which explain the dynamics in the linearized region) scale with the number of recorded neurons ($N_{\rm obs}$). In this model, reconstruction errors decrease following $\propto \tau_{\rm latent} N_{\rm obs}^{-0.5}$ for $N_{\rm obs}\ll N_{\rm rec}$, where $\tau_{\rm latent}$ refers to the timescale of the LPU. For $N_{\rm obs}\sim N_{\rm rec}$, increasing the number of recorded neurons accelerated error reduction beyond a power law (See Figure \ref{figa2}). Overall, combined with prior research demonstrating that few thousands of neurons may be sufficient to optimally extract the linear information encoded in neural populations \cite{hazon2022noise,rumyantsev2020fundamental}, this simple model clearly demonstrates the comparative advantage and necessity of large-scale neural recordings for inferring neural activity dynamics over second-long timescales.

We could have explored the preceding questions using a leaky current low-rank RNN, as studied in prior work \cite{dubreuil2022role,beiran2021shaping,valente2022extracting,mastrogiuseppe2018linking}. Beyond this point, however, the two models yield fundamentally different predictions (Table \ref{tabs1}). In the remainder of this work, we demonstrate how the biologically motivated RNN architectures that constitute the LPUs uniquely explain key empirical observations in systems neuroscience.

\subsection{Identifiable linear readouts can subserve complex behavior}

The choice of the linear encoding for LPUs provides a viable solution to the common identifiability issues shared across the latent variable models \cite{zhou2020learning, naselaris2011encoding, haxby2001distributed, pillow2011model, lange2023bayesian, dabney2020distributional, park2014encoding}. Notably, in a general case with nonlinear encoding, if $\kappa(t)$ were to constitute a (generalized) latent unit satisfying Eq. \eqref{eq1}, so would $T(\kappa(t))$ for an invertible (potentially nonlinear) transformation $T$. In such a case, $\kappa(t)$ would be considered non-identifiable, \textit{i.e.}, there are infinitely many non-trivial transformations that render several choices of $\kappa(t)$ equivalent for Eq. \eqref{eq1}. Fortunately, linearity of the encoding significantly constraints the group of transformations. Specifically, if $\kappa(t)$ constitutes an LPU, so would any $\tilde \kappa(t) = T \kappa(t)$, where $T \in \mathbb R^{K\times K}$ is an arbitrary invertible matrix (\textbf{Methods}). In contrast, a strictly nonlinear transformation $\tilde \kappa = T(\kappa)$ would violate the linearity of the encoding, rendering $\tilde \kappa$ no longer fit for an LPU. Hence, LPUs are linearly identifiable by definition \cite{hyvarinen2024identifiability}, which has been the target to achieve in previous literature \cite{schneider2023learnable}. 

Here, we take the identifiability considerations one step further and study a key question: In order to achieve optimal behavioral readouts $\psi(\kappa)$, which linearly transformed version of $\kappa(t)$, if any, would need to be tracked within a biological network? Fortunately, if the linear information (\textit{i.e.}, what may be readily accessible to downstream neurons through simple synaptic projections \cite{moreno2014information}) was sufficient for subserving complex behavior, there is no reason for networks to keep explicit track of latent variables at all. Linearity of the encoding ensures that a linear readout from the LPUs, $\psi(\kappa)$, can be equally represented with a counterpart defined on neural activities, $\tilde \psi(r)$. Then, information can be extracted from the neurons without explicitly identifying the former (Fig. \ref{fig2}\textbf{A}, Lemma \ref{lm2}; \textbf{Methods}). Though linear readouts have been widely discussed in neuroscience and provide a biologically plausible method for brain regions to efficiently communicate \cite{rigotti2013importance,kira2023distributed,ebrahimi2022emergent,semedo2019cortical}, their optimality has been an open question to date. Below, we demonstrate how they can become optimal in the sense of communicating LPU outputs.

Since the class of nonlinear decoders subsumes their linear counterparts, enforcing linearity on the readout could potentially limit its accuracy. Notwithstanding, prior work has shown that linear readouts in basic RNNs can provably approximate nonlinear functions of time by serving as ``universal approximators'' \cite{schafer2006recurrent}. The resolution of these seemingly paradoxical statements follows from the independency of encoding-embedding maps in our framework (which subsumes the basic RNN architecture \cite{schafer2006recurrent}). Intuitively, the \emph{trainable} embedding of LPUs into neural activities allows complex (nonlinear) computation to unroll over time, whereas having a \emph{linear} (also trainable) encoding allows easy access to latent variables. We formalize this intuition in Theorem \ref{thm2}, where we prove that linear readouts from neural activities can approximate any \emph{differentiable (nonlinear) functions of LPUs} in a broad class of networks. This flexibility, however, requires an increase in the latent dimensionality (\textbf{Methods}). Specifically, if a $K$-dimensional LPU suffices for a nonlinear readout, a $K+1$-dimensional LPU may be necessary for a linear readout to extract the same information. Hence, as long as there are available resources (\textit{i.e.}, $N_{\rm rec} \gg K$), LPUs endow linear readouts with a theoretical sufficiency guarantee.

Testing the sufficiency of linear readouts is challenging, as the ability to decode specific information from a brain region at a given time does not necessarily imply causal communication. Intermediate computational results may still correlate, often nonlinearly, with the variable we wish to decode. Therefore, to rigorously test the sufficiency of linear decoding, we focused on a basic and globally relevant information: trial identity in a Go-NoGo task. Specifically, we reanalyzed previously published large-scale neural recordings of layer 2/3 pyramidal neurons across up to eight neocortical regions (\cite{ebrahimi2022emergent}; Figs. \ref{fig2}\textbf{B-D} and \ref{figs2}). Meanwhile, mice performed visual discrimination tasks in which each trial begins with the presentation of a Go or NoGo stimulus, followed by a brief delay period after which the mouse must lock in its decision by licking a spout or refraining from licking (Fig. \ref{fig2}\textbf{B}). During this task, up to eight neocortical regions are recorded with a mesoscope, capturing on average $3595 \pm 989$ neurons while mice perform $1296 \pm 379$ correct trials within $5 \pm 1$ imaging sessions (all quantities in mean $\pm$ s.d.; Fig. \ref{fig2}\textbf{C}). The correct trial identity correlates with both the cue identity (during stimulus presentation) and the response (during the response period); therefore, we reasoned that this information is sufficiently general to be actively communicated across neocortical regions \cite{ebrahimi2022emergent}. 

To test the sufficiency of linear information readout on trial identity, we evaluated the performance of linear (logistic regression and linear discriminant) and nonlinear (random forest classifier and quadratic discriminant) decoders trained on neural activities from individual brain regions to predict the trial identity (Figs. \ref{fig2}\textbf{D} and \ref{figs2}). The decoders were regularized using dimensional regularization with partial least squares, following the procedures outlined in \cite{ebrahimi2022emergent}. Notably, across all brain regions and trial segments (stimulus, delay, response, post-response), decoder performance saturated with as few as 10 (often even far less) regularization dimensions, beyond which additional dimensions did not enhance performance (Fig. \ref{figs2}). Since the logistic regression and random forest classifiers were allowed to optimize their parameters with cross-validation (\textbf{Methods}), these results suggest that decoding accuracies can remain optimal despite low-dimensional bottlenecks. Moreover, under a very strict bottleneck with only three dimensions, linear decoders slightly outperformed nonlinear ones, likely due to overfitting in more complex architectures (Fig. \ref{fig2}\textbf{D}), though increasing the dimensionality of the bottleneck resulted in nearly equal performance for both decoder types (Fig. \ref{figs2}). 

These results align with previous findings \cite{rumyantsev2020fundamental}, where similar outcomes were observed in decoders trained on primary visual cortical populations. Specifically, \cite{rumyantsev2020fundamental} found that the class covariance matrices were approximately equal, making linear discriminants as effective as their quadratic counterparts. Our findings extend this principle across diverse cortical regions, suggesting that for information regularly communicated between brain regions, \emph{optimal} linear readouts (as potential communication subspaces) may be feasible (Theorem \ref{thm2}).

\subsection{Nonlinear embedding can lead to unbounded scaling of linear dimensionality}

The combination of encoding and embedding must provide nonlinearity to the network; otherwise, the LPU reduces to a simple linear dynamical system (\textbf{Methods}; Table \ref{tabs2}). Our analysis has demonstrated that linear encoding endows LPUs with a crucial, physiologically motivated property: the ability to extract information using linear readouts without requiring latent variable identification. This constrains the nonlinearity to originate from the embedding map (Eq. \eqref{eq1}; Fig. \ref{fig2}\textbf{E}). Below, we demonstrate that this constraint aligns with empirical observations - specifically, the high linear dimensionality observed in large-scale neural recordings could naturally arise from a low-dimensional LPU through nonlinear embedding \cite{manley2024simultaneous,stringer2019high}.

Theoretically, it is possible to design a scenario in which even a one-dimensional latent system can inflate linear dimensionality proportional to the neuron count (Proposition \ref{prop2}; \textbf{Methods}). However, the inflated linear dimensionality is more than just theoretical curiosity and can manifest practically in several ways. Specifically, a study by \cite{stringer2019high} (recently replicated by \cite{manley2024simultaneous} with up to a million neurons; also see \cite{pospisil2024revisiting}) found that many cortical areas exhibit a power law drop-off in the spectrum of neural activity. Intuitively, this means that the $i^{\text{th}}$ principal component (PC) explains $i^{-\alpha}$ variance in the data, with $\alpha$ being the slope of the drop-off. Critically, this $\alpha$ was found to be close to one \cite{stringer2019high}, suggesting a slow drop-off and, consequently, the high dimensionality of the observed neural activities. One explanation is that the neural code is simply high-dimensional, containing coding variables that can only be extracted through observing more and more neurons. An alternative explanation, consistent with LPUs (and briefly alluded to in \cite[Figure 4]{stringer2019high}), is that fixed latent dimensionality can lead to unbounded scaling due to nonlinear embedding maps. Below, we perform three distinct simulations that provide three possible paths to explain high-dimensional neural activities with low-dimensional LPUs. 

Firstly, to show that even static embedding alone can lead to increased linear dimensionality (which is the case in our theoretical example in Proposition \ref{prop2}; \textbf{Methods}), we performed a static \emph{noiseless} embedding of a $K$-dimensional unit sphere following the formula $r = \sin(\sum_{i=1}^K m^{(i)}\kappa_i)$, where $m^{(i)} \in \mathbb R^N$ are embedding weights drawn from a zero mean normal distribution (s.d. varying across experiments), whereas $\kappa \in \mathbb R^K$ are sampled uniformly from a $K$-dimensional unit sphere. A static embedding only constrains the neural activity $r$, not its time-derivative $\dot r$, and in that sense is less redundant. Here, we confirmed that with as little as $50$ latent dimensions, the embedding has shown nearly unbounded scaling up to $10,000$ neurons (Fig. \ref{fig2}\textbf{F}). This occurs as early as $5-10$ latent dimensions with larger embedding weights or a different nonlinearity (Fig. \ref{figs3}\textbf{A}).

A recent biological finding through recording up to a million neurons is that around half the variance was found to be in dimensions uncorrelated with animal's behavior \cite{manley2024simultaneous}. We observed that the same scenario can be reproduced with the embedding of a $10$-dimensional sphere (Fig. \ref{fig2}\textbf{G}). Here, only the neural activities in the first $10$ PCs have non-negligible correlations with the latent sphere. In contrast, approximately half of the variance is in higher PCs that have little-to-no correlations with the LPU (Fig. \ref{fig2}\textbf{G}). We confirmed that this result could not be explained due to the non-smoothness of the embedding manifold, as the configurations used in this experiment preserved the distances in the embedded unit sphere (Fig. \ref{figs3}\textbf{B}). Thus, while the observations in \cite{manley2024simultaneous} can be due to the high-dimensional nature of the neural code, they can also stem from nonlinear embedding of low-dimensional coding variables into curved activity manifolds.

As a second experiment, we studied networks with dynamical embedding maps. Specifically, we considered noisy low-rank leaky firing-rate RNNs trained to perform the sequence sorting tasks (Fig. \ref{fig2}\textbf{H-I}). In this task, the RNNs are presented with a set of numbers, which they have to sort and output in the sorted order. With these RNNs, we found again an increased scaling of the linear dimensionality as a function of neurons (Fig. \ref{fig2}\textbf{H}) but not as a function of the rank, \textit{i.e.}, latent dimensionality (Fig. \ref{fig2}\textbf{I}). Moreover, these observations were robust with respect to the changing sequence length that needed to be sorted (Fig. \ref{figs3}\textbf{C-E}). Overall, this experiment provided evidence that the inherent structure of the LPU, not necessarily its rank, may set the rules of the dimensionality scaling, in line with the previous predictions about task dimensionality \citep{gao2015simplicity,gao2017theory}.

Finally, previous work has shown that full-rank RNNs with randomly sampled weights can show self-sustained time dynamics \cite{sompolinsky1988chaos}. We found that their low-rank counterparts, after appropriate weight adjustments (\textbf{Methods}), show similar behavior, in which nearby trajectories that start with only small distances apart can grow over time (Fig. \ref{figs3}\textbf{F}). For these networks, the scaling between the number of neurons and the linear dimensionality became nearly linear in a log-log plot (Fig. \ref{figs3}\textbf{G}). Though saturation is quite possible when more neurons are used to embed the latent variables, this simulation shows an example of a locally linear increase in linear dimensions of neural activities despite embedding a lower dimensional (randomly connected) dynamical system. In real networks, it is not far-fetched to think that a few ranks of the connectivity matrix may remain random, perhaps due to them not being purposed for a task yet. Thus, increased linear dimensionality can also stem from nonlinear embedding of low-dimensional \emph{random/unpurposed} LPUs.

\subsection{The geometry of the neural code: Principal, tuning, and coding dimensions}

The geometry of neural activity manifolds fundamentally determines what can be extracted from neural recordings about coding subspaces. At one extreme, neural activities may lie on hyperplanes that match the latent dimensionality. This simplifying assumption has been made by most existing models, including rSLDS \cite{linderman2017bayesian,vinograd2024causal,nair2023approximate,mountoufaris2024line,hu2025modeling}, leaky current low-rank RNNs \cite{dubreuil2022role,valente2022extracting,beiran2021shaping,mastrogiuseppe2018linking}, and the latest work on latent circuits \cite{Langdon2025}. These works define a linear constraint from latent variables to neural activities ($r(t) = M \kappa(t)$ for some $M \in \mathbb R^{N_{\rm rec} \times K}$). In the resulting flat geometries, three key types of dimensions align: principal dimensions (those with highest variance), tuning dimensions (those including neurons tuned to latent variables), and coding dimensions (those encoding LPUs). 

However, biological reality is more complex. The brain exhibits high redundancy \cite{ebrahimi2022emergent}, and not all neurons tuned to a particular variable are causally linked to it \cite{feinberg1978efference}. A more realistic scenario, as illustrated in Fig. \ref{fig2}\textbf{E-I}, is that neural activities reside on ``curved'' neural activity manifolds—those with higher dimensionality than the underlying LPU. While mathematically challenging (as principal, tuning, and coding dimensions no longer align), models capable of incorporating these flexible geometries are more general and may better capture biological computation. Fortunately, LPUs provide a rigorous framework for understanding and studying the relationship between these three experimentally relevant neural activity dimensions. In networks with linear encoding weights $n^{(p)}$ for $p=1,\ldots, K$, coding dimensions are defined by the span of these encoding weights, as deviations in neural activities orthogonal to them decay exponentially with no little-to-no effects on the LPU operation (Theorem \ref{thm4}, see below). The nonlinear embedding introduces curvature that inflates principal dimensions, whereas encoding-embedding relationship in Eq. \eqref{eq1} decouples coding dimensions from those tuned to latent variables.

Now, we illustrate the independency of tuning and coding in LPUs by conducting simulations with RNNs (Fig. \ref{fig3}). For RNNs defined in Eq. \eqref{eq_lfrnn}, we can extract the tuning curves of neural activities as follows (\textbf{Methods}): 
\begin{equation} \label{eq_tuning}
r_i(\kappa) = \tanh\left(\sum_{p=1}^K m^{(p)}_i \kappa_p\right), 
\end{equation}
where the $i$th component of the embedding weights, $m^{(p)}_i$, determines the tuning properties of the $i$th neuron with respect to the latent variables $\kappa \in \mathbb{R}^K$. To omit complicating our illustration with linear identifiability discussions, without loss of generality (\textbf{Methods}), we assume that $\kappa$ are aligned with behaviorally relevant variables (\textit{e.g.}, the flip-flop outputs in Fig. \ref{fig1}). 

The expression in Eq. \eqref{eq_tuning} clearly demonstrates that the encoding weights ($n^{(p)}$) do not influence the tuning curves; rather, tuning properties—whether neurons are sparse or mixed-selective—are solely determined by the structure of the embedding weights and the nonlinearity. For instance, consider a well-trained RNN with $N_{\rm rec}$ neurons. A ``redundant'' neuron $j$ can be added to this network such that $m^{(p)}_j \neq 0$ but $n^{(p)}_j = 0$ for $p = 1, \ldots, K$. This redundant neuron would reflect latent variable values in its activities, \textit{i.e.}, be tuned to them. However, zero encoding weight means zero causal contribution to the LPU (Fig. \ref{fig3}\textbf{A}), \textit{i.e.}, the network dynamics would remain invariant under the addition of (and any interventions through) this redundant neuron.

In practice, one brain region may use redundant neurons to convey information about its internal computations to another. However, theoretically, such neurons are unnecessary for performing a specific task, as they do not directly contribute to the computation. To test this hypothesis, we reanalyzed the encoding, embedding, and input weights from the RNNs trained on the 3-bit flip-flop task (Fig. \ref{figs1}\textbf{C}). Across neurons, these three weights exhibited strong correlations (Pearson’s $r = 0.87$, $0.90$, and $0.97$). We observed the existence of neurons with near zero encoding and embedding weights, \textit{i.e.}, a subset was neither involved in computation nor showed tuning to internal states. However, these RNNs contained few, if any, redundant neurons that were tuned to internal variables without contributing to the LPU (Fig. \ref{figs1}\textbf{C}). Thus, as expected, when the goal was to solve a specific task and nothing else, artificial neurons were either coding or had no contribution, but were not redundant.

To demonstrate how studying only the tuning curves might mask the redundant neurons in biological networks, we next incorporated redundant neurons into a representative RNN. Specifically, we modified the RNN from Fig. \ref{fig1}\textbf{D-F} by adding 100 redundant neurons (\textbf{Methods}). These neurons were given new embedding weights, $m^{(p)}$, while their encoding and input weights were set to zero. In this construction, it is theoretically possible to prove that redundant neurons have no effect on the latent computation, as there are no connections between the redundant neurons or from the redundant neurons back to the original population (Fig. \ref{fig3}\textbf{B}). Yet, by Eq. \eqref{eq_tuning}, these neurons are guaranteed to be tuned to $\kappa(t)$, the same way the original coding population does.

To introduce additional complexity to our empirical analysis and show how connectomics may mask functional connections in LPUs, we sampled a random weight matrix, with roughly twice the standard deviation of the structured component, and added it to the low-rank component (Fig. \ref{fig3}\textbf{B}). This led to substantial connections within and between coding and redundant neurons, \textit{i.e.}, redundant neurons now influenced the activities of coding neurons. Nonetheless, while we observed expected state changes when coding neurons were subjected to external perturbations (Fig. \ref{fig3}\textbf{C}, \emph{left}), perturbing redundant neurons did not impact the internal state. Furthermore, once the perturbation on the redundant population was lifted, the network quickly reverted to its original memory state (Fig. \ref{fig3}\textbf{C}, \emph{right}). Hence, even though the two populations were densely connected to each other, the lack of (structured) encoding weights for the redundant neurons made them ineffective for perturbing the LPU dynamics. 

Further increasing the strength of the random components could entirely mask LPU operation, eventually reducing the RNN to a randomly connected network. Yet, even in this simple illustration, the LPU functional connections were obscured by random connections, but still operated faithfully (Fig. \ref{fig3}). The extent to which connectomics constrains network function remains an open question in the field \cite{ozdil2024centralized,lappalainen2024connectome}. LPUs offer a fresh perspective on this issue, highlighting the intricate interplay between structured and random components. A detailed investigation of this phenomenon is left for future work.

\subsection{Robustness of latent processing units to changing neural tuning curves}
A major advantage of LPUs is related to their causal modeling of neural and latent dynamics under external perturbations, which involves changes in the synaptic weights between neurons. Beyond focusing on behavioral readouts in the presence of drift \cite{deitch_representational_2021,driscoll2022representational}, LPUs allow studying the robustness of latent dynamics \emph{and} linear readouts jointly. In our framework, the decoupled representation of encoding and embedding provides a simple path for studying how changes in tuning curves affect network operations. Specifically, the tuning relationship described in Eq. \eqref{eq_tuning} enables a novel approach to simulating representational drift in RNNs: modify a particular weight vector ($m^{(d)} \to m^{(d)} + \Delta m$) while maintaining a fixed set of encoding weights ($n^{(p)}$ for $p = 1, \ldots, K$). Theoretical analysis reveals that LPU dynamics can remain approximately invariant to drift when either of the two conditions are met (Theorem \ref{thm5}; \textbf{Methods}): when small changes $\Delta m$ are either purely random or when they are confined to an $N_{\rm rec} - K$ subspace orthogonal to the encoding subspace (Fig. \ref{fig4}\textbf{A}). Below, we demonstrate these cases by simulating representational drift in RNNs performing $K$-bit flip-flop tasks, which maintain $2^K$ attractive fixed-points in their LPUs. 

To start with, we illustrate a drift simulation in Fig. \ref{fig4}\textbf{B}, using the example RNN from Fig. \ref{fig1} performing a 2-bit flip-flop task. After aligning the latent variables so that each corresponds to a specific output, we can examine the embedding weights to reveal the tuning properties of individual neurons. Notably, only a few neurons exhibit significant tuning to the flip-flop states (black dots; Fig. \ref{fig4}\textbf{B}), with each tuned to a specific input. After a mild drift is applied to the embedding weights (blue dots; Fig. \ref{fig4}\textbf{B}), most neurons altered their tuning properties, with some even developing mixed-selective tuning, as shown in the tuning curves of neuron 59 before and after the drift (Fig. \ref{fig4}\textbf{B}). 

Next, using this mechanism, we simulated representational drift in large-scale RNNs in three distinct ways (\textbf{Methods}): (i) with fully random $\Delta m$ sampled from a Gaussian distribution with zero mean and $g_{\rm drift}$ s.d., (ii) with $\Delta m$ confined to the subspace orthogonal to the encoding, and (iii) with $\Delta m$ tangential to the encoding subspace (Figs. \ref{fig4}\textbf{C-J}). Latent dynamics are theoretically predicted to remain stable under mild drifts of the first two types, but not the third. To test this, we simulated representational drift by varying $g_{\rm drift}$ from $10^{-3}$ to $10$. Given that the extreme values for the embedding weights are approximately $\sim \pm 1$ (Fig. \ref{figs1}\textbf{C}), we expect that at around $g_{\rm drift} = 1$, the tuning properties of individual neurons will show significant changes. We confirmed this prediction by considering the fraction of neurons retaining their coding properties as a function of $g_{\rm drift}$ (Fig. \ref{fig4}\textbf{C}). Consequently, a network that remains invariant to drift at this level can be considered robust. Since the drifts were applied to the embedding weights, a behavioral readout (dependent on invariant encoding weights) defined before drift would be expected to remain accurate if the LPU dynamics remained unaffected by drift. Thus, we assessed robustness by computing state estimation accuracies derived from latent variables aligned with network outputs (\textbf{Methods}). Specifically, we assigned network states based on the positive and negative values of $\kappa$ and compared the internal state to the target state (\textbf{Methods}). The drift tolerance threshold ($g_{\rm half}$) was defined as the transition point from a sigmoidal fit on the accuracy vs. $g_{\rm drift}$ curve. 

In our experiments, RNNs with as few as a few hundred neurons exhibited robustness to drifts orthogonal to encoding subspaces, while robustness to fully random drift emerged only when networks had a few thousand neurons (Fig. \ref{fig4}\textbf{D}). Unsurprisingly, even networks with up to 300,000 neurons were not robust to drifts within encoding subspaces (Fig. \ref{fig4}\textbf{D}). Inspired by this observation, we aimed to develop a mechanistic understanding of robustness (or lack thereof) to drift in encoding orthogonal and tangential directions. We revisited the example RNN from Fig. \ref{fig1}, which was trained to perform a 2-bit flip-flop task and possessed four attractive fixed points corresponding to $2^2 = 4$ internal flip-flop states. We applied a strong encoding orthogonal ($g_{\rm drift} = 0.8$, Fig. \ref{fig4}\textbf{E-F}) and a mild encoding tangential ($g_{\rm drift} = 0.1$, Fig. \ref{fig4}\textbf{G-H}) drifts to this network. In line with Fig. \ref{fig4}\textbf{C}, the strong orthogonal drift had a more pronounced effect on the tuning properties of individual neurons (Fig. \ref{fig4}\textbf{E, G}), yet the mild tangential drift had a significant qualitative impact on the LPU dynamics. Specifically, under strong encoding orthogonal drift, although the attractive fixed points shifted closer to the origin, the LPU retained all four internal states (Fig. \ref{fig4}\textbf{F}). In contrast, even a mild encoding tangential drift resulted in a complete restructuring of the LPU, \textit{i.e.}, the LPU underwent bifurcations that eliminated all four fixed-point attractors (Fig. \ref{fig4}\textbf{H}). These findings align with Theorem \ref{thm5} (\textbf{Methods}), which suggests that changes in embedding weights orthogonal (but not tangential) to the encoding subspace have minimal first-order effects.

Finally, we considered the effects of further increasing the neuron counts on the three distinct types of drift (Fig. \ref{fig4}\textbf{I-J}). As expected, the encoding tangential drift (but not others) significantly affected the latent code, regardless of the number of neurons (Figs. \ref{fig4}\textbf{I}). Interestingly, partial observation from RNNs with up to 300,000 neurons still revealed qualitative robustness to drift, though the quantitative extent of drift tolerance was masked (Fig. \ref{fig4}\textbf{J}).  Thus, these simulations confirmed that even when observing only 100 out of thousands of neurons, a linear readout from these subpopulations could still be robust to drift.

\section{Discussion} \label{sec:discussion}
\subsection{Combining past and present: from tuning curves to population codes}

For over two decades, researchers have explored population codes, where collective neural activities jointly encode computationally and behaviorally relevant variables \cite{abbott1999effect}. While subsequent theoretical work examined the linear information extractable from correlated neural populations \cite{moreno2014information}, these models primarily focused on the readouts from collective dynamics rather than developing a comprehensive theory from fundamental principles \cite{sussillo2014neural}. The field has advanced significantly in this decade, focusing on identifying the complete set of coding latent variables extractable from neural populations \cite{depasquale2023centrality,zimnik2024identifying,schneider2023learnable}. However, since these models reconstruct neural activities without enforcing dynamics, they provide an incomplete, under-determined picture of computation through collective dynamics (\textbf{Methods}). To bridge this gap, recent work utilizing dynamical models has fitted latent dynamical systems to reproduce neural recordings, leading to several scientific breakthroughs \cite{liu2024encoding,valente2022extracting,vinograd2024causal,o2022direct,abbaspourazad2024dynamical}. Yet, existing approaches have fallen short in providing a rigorous and physiologically motivated theoretical account of the neural processing units we seek to identify from large-scale neural recordings.

The LPUs we introduce in this work build upon and significantly extend a recent line of seminal research in computational neuroscience. Previous work has shown that low-rank decomposition, similar to Eq. \eqref{eq_lowrank}, leads to latent dynamical systems in a particular (constrained; \textbf{Methods}) RNN architecture \cite{beiran2021shaping,dubreuil2022role,mastrogiuseppe2018linking,valente2022extracting}. These studies established important foundations by relating low-rank connectivity between individual neurons with computations in lower-dimensional subspaces and proving a universal approximation theorem analogous to our Theorem \ref{thm1}. They also demonstrated that deviations orthogonal to the encoding subspace decay exponentially (as in our Theorem \ref{thm4}), hence the low-dimensional coding subspace. Yet, this line of work focuses on a single RNN architecture, in which latent dynamical systems are obtained through a constraint that enforces a linear map from latent variables to neural activities \cite{beiran2021shaping,dubreuil2022role,mastrogiuseppe2018linking,valente2022extracting}. On the other hand, low-rank structures can be abundantly found in the physical world \cite{thibeault2024low}. With this insight, LPUs generalize the low-rank RNNs by focusing on a broader class, with a flagship architecture in Eq. \eqref{eq_lfrnn} that differs from the RNN architecture studied in this literature. 

Our framework introduces a principled class of models derived from biological constraints—specifically, those exhibiting linear encoding and nonlinear dynamical embedding properties (\textbf{Methods}). These models demonstrate several crucial capabilities: they support universal computations in their LPUs (Fig. \ref{fig1}), allow linear readouts to become optimal for communication between brain regions (Fig. \ref{fig2}; top), account for seemingly high-dimensional curved neural manifolds (Fig. \ref{fig2}; bottom), and maintain robustness both to targeted perturbations (\textit{e.g.}, stimulation experiments in Fig. \ref{fig3}) and to variations in synaptic connections (\textit{e.g.}, representational drift in Fig. \ref{fig4}). While previous studies have proposed separate phenomenological models to explain some individual aspects of these observations \cite{driscoll2022representational,stringer2019high,bounds2024network}, LPUs provide a rigorous theoretical framework that derives these (and more) properties from a first principle: computation in large populations of neurons emerges from low-dimensional universal LPUs, whose time evolution does not explicitly depend on any individual neuronal activities. In summary, our findings in Theorems \ref{thm2}, \ref{thm5} and Propositions \ref{prop1} and \ref{prop2} do not have a counterpart in the existing low-rank RNN literature (though the original proof on full-rank RNNs \cite{schafer2006recurrent} establishes a similar connection as in Theorem \ref{thm2}; \textbf{Methods}), whereas our Theorem \ref{thm4} enables coding subspaces to be curved as opposed to the hyperplanes assumed by existing work \cite{beiran2021shaping,dubreuil2022role,mastrogiuseppe2018linking,valente2022extracting}. With these theorems, LPUs operationalize the \emph{broadly} observed low-rank connections (\cite{thibeault2024low}) as the link between structure and function in networks, and thereby constitute a formal framework to explain the ``neural manifold hypothesis.''

\subsection{Direct empirical tests of latent processing units}

The latent computation framework makes several experimentally testable predictions that can be validated through a two-step experimental approach. The main testable prediction concerns the dimensionality of effective perturbations. While neural populations often exhibit high-dimensional activity patterns when analyzed with methods like PCA (scaling with the number of neurons, $N_{\rm rec}$), our theory predicts that effective perturbations—those that meaningfully influence behavior—should be confined to a much lower-dimensional space (equal to $K$, Fig. \ref{fig3}). This prediction becomes particularly relevant with recent advances in single-cell manipulation techniques. First, using optogenetic manipulations and targeted interventions, we can test whether effective perturbations of neural activity are confined to a low-dimensional subspace. Second, by examining how perturbations decay along orthogonal dimensions, we can determine whether the high linear dimensionality observed in neural recordings arises from nonlinear embedding of low-dimensional latent dynamics. While these experiments would provide direct evidence for our theoretical framework, robustly constraining the rank of the perturbation will require new generation single-cell intervention tools that allow fine-grained  optogenetic access to thousands of individual cells.

Beyond these interventional tests, our theory makes specific predictions about neural manifold geometry. Even in seemingly simple tasks, we predict that neural activities will reside in curved manifolds rather than linear hyperplanes. Unlike existing approaches such as recurrent switching linear dynamical systems (rSLDS) \cite{linderman2017bayesian,vinograd2024causal,nair2023approximate,mountoufaris2024line,hu2025modeling} or leaky current low-rank RNNs \cite{dubreuil2022role,valente2022extracting,beiran2021shaping,mastrogiuseppe2018linking}, which implicitly assume minimal curvature, our theory explicitly predicts non-zero curvature in the neural manifold (Fig. \ref{fig2}). Preliminary evidence exists in spatial navigation, where grid cells form toroidal manifolds \cite{gardner2022toroidal,hermansen2024uncovering}. While previous work has characterized linear dimensionality in passive viewing tasks \cite{stringer2019high}, manifold curvature remains  underexplored, partly due to challenges in estimating manifold dimensionality under noisy conditions (for instance, recent works mitigated this issue by assuming the existence of a particular latent structure \cite{bjerke2022understanding,acosta2023quantifying,klindt2023topological}). That being said, recent tools for analyzing nonlinear geometric features \cite{fortunato2024nonlinear} provide supporting evidence for our predictions and enable more complex tests.

\subsection{Insights for experimental systems neuroscience}

In connecting our framework to classical tuning curve analyses \cite{o1976place,hubel1959receptive}, we uncover a critical distinction: while it is often easier to determine neurons' tuning properties (\textit{i.e.}, their embedding weights $m^{(p)}$), these properties may not reflect their functional impact on behavior. Statistical analyses based on tuning properties identifies neurons with large embedding weights ($M$), but these neurons may reside outside the coding subspace that directly influences the LPU (Fig. \ref{fig3}). Instead, for effective perturbation of the latent state, we must target neural activities within the encoding subspace defined by $n^{(p)}$ (Theorem \ref{thm4}). This prediction is supported by recent findings \cite{bounds2024network}, where perturbations based on tuning curves failed to disrupt behavior, while targeting neurons based on their network contributions did (see also \cite{posani2024rarely}).

Our theoretical framework also provides specific predictions about hub neurons, which are characterized by extremely high connectivity, with local excitatory inputs and long-range projections \cite{bocchio2020hippocampal}. These neurons are uniquely positioned to store and communicate LPU outputs, with their anatomical organization suggesting a natural implementation of the encoding-embedding separation: local inputs could provide access to the encoding subspace, while distant projections enable broadcasting of encoded information. Our theory predicts that for hub neurons to track latent variables, they likely operate in a linear regime of their activation functions (see \cite{lafosse2024cellular}) —a direct consequence of the linear encoding assumption (Fig. \ref{fig1}). This prediction is particularly noteworthy and somewhat counterintuitive because neurons typically operate nonlinearly, yet maintaining linear encoding would require hub neurons to remain within a more restricted, linear range.

The scale-dependence of neural computation provides another key insight. While RNNs can solve many behavioral tasks with few neurons \cite{dubreuil2022role}, biological systems employ millions of neurons for seemingly simple decisions \cite{ebrahimi2022emergent}. This apparent paradox is resolved by our framework's prediction that larger populations are essential for two aspects: maintaining precise control over extended temporal sequences and ensuring robustness to representational drift. Specifically, while networks become increasingly robust to random drift as they grow (Fig. \ref{fig4}), smaller networks require strict orthogonality between drift dimensions and encoding weights—a constraint whose biological plausibility remains uncertain \cite{driscoll2017dynamic}. Reconciling the process with which the drift occurs (random or structured; if latter, encoding orthogonal or tangential) is crucial for a mechanistic explanation of how neural circuits maintain functional stability despite continuous cellular changes. 

Another interesting future direction involves extending our analysis to spiking recurrent neural networks (sRNNs). Prior work introduced a training paradigm in which latent factors—interpreted as time series data—serve as computationally central elements, providing structured targets for robust and flexible training of spiking networks \cite{depasquale2023centrality}. However, it remains unclear whether a self-sufficient and universal latent dynamical system can emerge within such spiking architectures, which is the precursor of the LPUs we introduced here. Moreover, future work will be needed to determine whether LPUs in sRNNs, if they exist, can achieve similar levels of generality and explanation power of empirical phenomena observed in this work, while adhering to the constraints of discrete spike-based communication. Nevertheless, incorporating LPUs into biologically realistic spiking networks could bridge the gap between theoretical models and neural circuit implementations, offering a promising avenue for further biological investigation.

Finally, LPUs, if implemented by biological networks, have direct implications for the development of brain-machine interfaces (BMIs) \cite{safaie2023preserved,degenhart2020stabilization}. As we discussed in this work, while biological systems require large numbers of neurons to sustain complex dynamics, manipulating the information stored in these systems could be achieved through relatively low-dimensional control signals. Specifically, it may be possible to interfere with low-dimensional LPUs by choosing the right group of neurons and in a direction aligned with their encoding weights. Moreover, the linearity of encoding further suggests that local field potential recordings—arising from the linear summation of neural activity—may be as effective as single-unit recordings in decoding information from the brain \cite{jackson2016decoding}, especially if the encoding weights have spatial smoothness properties. This insight implies that high-performance BMIs may require, not necessarily increased spatial resolution, but theoretical advances that enable extracting drift-robust encoding dimensions. Overall, LPU-based algorithms, designed to track the encoding subspaces and compensate for occasional drifts, could significantly improve BMI stability and performance.

To conclude, we have introduced a rigorous \emph{dynamical} theory of neural codes carried out collectively by large neural populations. While some concepts we studied were partly described in experimental and computational neuroscience, our work provides the first rigorous theory that generalizes existing approaches and connects experimental paradigms. The key insight—the decoupled encoding-embedding relationship in Eq. \eqref{eq1}—allows us to explain diverse phenomena under a single theoretical umbrella while providing empirically testable predictions that could conclusively evaluate the neural manifold hypothesis \cite{yuste2015neuron}.

\section*{Acknowledgements}
We thank Vasily Kruzhilin, Dr. Itamar Landau, Mert Yuksekgonul, Dr. Kanaka Rajan, Dr. Sarah Kushner, Mathilde Papillon, and Dr. Ali Cetin for insightful discussions, and Xavier Gonzalez, Dr. Adrian Valente, and Dr. Chong Chen for their helpful feedback on highlighting the relevant work. FD expresses gratitude for the valuable mentorship he received at PHI Lab during his internship at NTT Research. MJS gratefully acknowledges funding from the Simons Collaboration on the Global Brain and the Vannevar Bush Faculty Fellowship Program of the U.S. Department of Defense. FD receives funding from Stanford University's Mind, Brain, Computation and Technology program, which is supported by the Stanford Wu Tsai Neuroscience Institute. This research was supported in part by grant NSF PHY-2309135 and the Gordon and Betty Moore Foundation Grant No. 2919.02 to the Kavli Institute for Theoretical Physics (KITP).

\clearpage
\newpage

\setcounter{equation}{0}
\setcounter{theorem}{0}
\setcounter{lemma}{0}
\setcounter{definition}{0}
\setcounter{corollary}{0}

\renewcommand\theequation{M\arabic{equation}}

\section*{\LARGE Methods}

\section*{A theory of the latent processing units}

What distinguishes recurrent neural networks (RNNs) from their feed-forward counterparts is the ability to model temporally varying computation. Studying temporal dynamics requires understanding dynamical systems, which is our starting point as we present the mathematics behind our latent processing units (LPUs). This exposition yields an \emph{interpretable} and \emph{rigorous} definition of LPUs, and generalizes existing work on a restricted class of low-rank RNNs (\cite{mastrogiuseppe2018linking,beiran2021shaping,dubreuil2022role,valente2022extracting}) and recurrently switching linear dynamical systems (rSLDS) \cite{linderman2017bayesian,vinograd2024causal,nair2023approximate} to a broad class of artificial and biological networks. 

Latent computation framework builds on a significant assumption: Neural computation results from the time evolution of a high-dimensional dynamical system, which follows a generic set of equations:
\begin{equation}
    \dot r(t) = F(r(t),u(t);\tilde W),
\end{equation}
where $r(t) \in \mathbb R^{N_{\rm rec}}$ are the state variables of the dynamical system; $u(t) \in \mathbb R^{N_{\rm in}}$ are inputs provided by the user to control the system; $\tilde W$ (of arbitrary dimensions) are the parameters of the system and $F: \mathbb R^{N_{\rm rec}+N_{\rm in}} \to \mathbb R^{N_{\rm rec}}$ is the vector field that defines the time dynamics. Here, the state variables $r(t) \in \mathbb R^{N_{\rm rec}}$ can be thought of as representing the activity of $N_{\rm rec}$ neurons. The parameters of the dynamical system, $\tilde W$, would then be the combinations of synaptic connections between the $N_{\rm rec}$ neurons and the weighting of the outside inputs. Throughout this work, we use $W^{\rm rec}$ to refer to the former, and $W^{\rm in}$ to refer to the latter. Lastly, $F(.)$ would then define the time dynamics of neural activities in the biological network.

\subsection*{\color{subsectioncolor}A dynamical encoding-embedding framework of neural computation}

In many RNNs, individual units, \emph{e.g.}, neurons, have pre-defined input-output relationships for transforming electrical currents into membrane potentials and subsequent action potentials \cite{kandel2000principles}. Even though these individual units are constrained, their populations can provide complex network outputs and model dynamics of experimentally observed variables \cite{schafer2006recurrent}. In this work, to formalize the ability of a broad class of RNNs' to model low-dimensional dynamical systems, we start by proposing the concepts of encoding and embedding in dynamical systems:
\begin{definition}[Encoding and embedding in dynamical systems]\label{def:embedding}
Let $ r \in \mathbb{R}^{N_{\rm rec}}$ be the state variables of a high-dimensional dynamical system $\dot  r = F(r,u)$, with $F(.)$ its defining vector field, and input variables $u(t) \in \mathbb R^{N_{\rm in}}$. Let $\kappa \in \mathbb{R}^K$ be state variables that follow another, lower-dimensional, dynamical system $\dot \kappa = g(\kappa,u)$, with $g(.)$ its defining vector field, and $K \ll N_{\rm rec} + N_{\rm in}$. Then, the dynamical system defined by $\kappa(t)$ is said to be an encoding of the dynamical system $r(t)$, if two conditions are satisfied:
\begin{enumerate}
    \item There exist a differentiable encoding map $\phi: \mathbb{R}^{N_{\rm rec}} \to \mathbb{R}^K$ and an embedding map $\varphi: \mathbb{R}^{K+N_{\rm in}} \to \mathbb R^{N_{\rm rec}}$ s. t.:
    \begin{equation} \label{eq:encoding}
\begin{split}
      \kappa(t) &= \phi(r(t)), \\
        \tau \dot r(t) &= -r(t) + \varphi(\kappa(t),u(t)),
\end{split}
\end{equation}
    \item The encoding and embedding maps satisfy the following property such that $\kappa(t)$ constitutes a dynamical system:
\begin{equation} \label{eq:embedding_condition}
    \dot \kappa(t) := \nabla \phi(r(t)) \; \dot r(t) =  \underbrace{\nabla \phi(r(t))}_{\text{depends on } r(t)} \underbrace{\left[ - r(t) + \varphi(\kappa(t),u(t))  \right]}_{\text{depends on } r(t)} = \underbrace{g(\kappa(t),u(t))}_{r(t) \text{ independent}},
\end{equation}
\end{enumerate}
where $\tau$ is neuronal decay times, $[\nabla \phi(r)]_ {ij} = \partial \phi_i(r) / \partial r_j$ is the Jacobian of $\phi(r)$, and $g(\kappa(t),u(t))$ does not explicitly depend on $r(t)$.
\end{definition}

This definition states that for a set of variables, $\kappa(t)$, to be considered as embedded in a high-dimensional dynamical system, $r(t)$, their time evolutions should both be self-consistent (with the exception of outside input variables, which are assumed to be pre-defined) and the former dynamical system should be recoverable from the latter. Although the encoding is a map, $\phi$, between state variables, we emphasize that the embedding $\varphi$ only constrains the time derivative $\dot r(t)$, allowing redundant representation, \textit{i.e.}, more than one $r(t)$ can correspond to the same $\kappa(t)$. In contrast, a static embedding function of the form $r(t) = h(\kappa(t))$ would uniquely define an $r(t)$ corresponding to a particular $\kappa(t)$. Hence, we will refer to $\varphi(\cdot)$ as a \textit{dynamical} embedding.

There are several trivial pairs of dynamical systems that can satisfy these encoding-embedding relationships. For example, consider $\kappa(t) = \phi(r(t)) = C$  with some constant $C \in \mathbb R^K$. Notably, this would satisfy the constraints with $\dot \kappa(t) = 0$ and leads to a leaky integrator for the transformed input, $A \dot r(t) = -r(t) + \varphi(u(t))$. However, this trivial latent dynamical system cannot solve complex computational tasks, barring some that require trivial integrations. Therefore, our definition of ``a latent processing unit," \textit{i.e.}, the putative substrate of neural computation, cannot solely depend on the ability of the latent dynamical system to satisfy the encoding and embedding relationship. Yet, as we discuss below, the conditions that $\kappa(t)$ is self-consistent and accessible from $r(t)$ constitute important steps toward rigorously defining LPUs.

For the rest of this section, we show that the dynamical nature of the embedding function leads seamlessly to dynamical system definitions for both $\kappa(t)$ and $r(t)$. This represents a significant departure from the static reconstruction of observed neural activities, \textit{e.g.}, performed by traditional encoder-decoder models \citep{eliasmith2005unified,eliasmith2003neural,boerlin2013predictive}. The type of encoding and embedding maps, linear or nonlinear, divides the networks into four groups with distinct geometrical properties. The geometry then significantly impacts the statistical properties and complexity of the neural computation that can be performed by the dynamical systems (discussed below and summarized in Table \ref{tabs2}). For simplicity and since we are interested in computations that can be intrinsically achieved by the dynamical systems, we assume that there are no inputs (and biases) unless otherwise specified. We now study some (non-exhaustive) properties of dynamical systems belonging to the each of the four groups mentioned above. 

\subsubsection*{Linear encoding and linear embedding} 

Assuming that both the encoding and the embedding are linear leads to the following relationships:
\begin{subequations}
    \begin{align}
        \kappa(t) &= N r(t), \\
       A \dot r(t) &= - r(t) + M \kappa(t), 
    \end{align}
\end{subequations}
where $N \in \mathbb R^{K\times N_{\rm rec}}$, $A \in \mathbb R^{N_{\rm rec}\times N_{\rm rec}}$, and $M \in \mathbb R^{N_{\rm rec}\times K}$ are weight matrices. Unless otherwise stated, we always assume that these matrices are maximally ranked, \textit{e.g.}, $M$ is rank $K$. We refer to them as encoding weights, neuronal decay times (if $M = 0$, $r(t)$ decay to zero following the eigenvalues of $A$), and embedding weights, respectively. Here, in the absence of latent excitations, \emph{i.e.}, $\kappa = 0$, neural activities decay to zero following neuronal decay times stored in the diagonal $A$ matrix. Combining both equations, we obtain the following relationship for the neural activities:
\begin{equation}
    A \dot r(t) = - r(t) + M N r(t) \implies  \dot r(t) = - \underbrace{A^{-1}[I-MN]}_{\Gamma}r(t),
\end{equation}
where we lumped the combination of weight matrices into an effective decay matrix, $\Gamma$. This means that neural activities would follow a linear equation, with the solution $r(t) = \sum_{i=1}^N \alpha_i e_i \exp(-\lambda_i t)$ for some overlap $\alpha_i \in \mathbb R$ depending on the initial conditions, and eigenvalues ($\lambda_i \in \mathbb C$) and eigenvectors ($e_i \in \mathbb C^N$) of $\Gamma$. Consequently, the latent variables would also follow a linear equation, $\kappa(t) = \sum_{i=1}^N \alpha_i \tilde e_i \exp(-\lambda_i t)$, where we defined $\tilde e_i = N e_i$. However, this equation implies that $\kappa(t)$ can only decay, oscillate, or blow up. Hence, for $\kappa(t)$ to support a diverse set of computations, both encoding and embedding cannot be linear. 

\subsubsection*{Nonlinear encoding and linear embedding} 

Next, we consider the case where encoding is nonlinear and embedding is linear, with the following equations:
\begin{subequations}
    \begin{align}
        \kappa(t) &= \phi(r(t)), \\
      A  \dot r(t) &= - r(t) + M \kappa(t), 
    \end{align}
\end{subequations}
for some nonlinear encoding function, $\phi(.)$, and $A$ and $M$ defined as before. This type of encoding enforces some conditions on the nonlinearity, since the latent variables follow the equations:
\begin{equation}
    \dot \kappa(t) = \nabla \phi(r(t)) \dot r(t) = -  \nabla \phi(r(t)) A^{-1} r(t) +  \nabla \phi(r(t)) A^{-1} M \kappa(t).
\end{equation}
Currently, we do not know whether there are encoding and embedding functions that can produce closed-form solutions for $\kappa(t)$. However, for the sake of studying the properties of such a system, we consider the embedding equation closely. Specifically, we define $B = A^{-1}$, $C = A^{-1}M$, and multiply both sides with the matrix $e^{Bt}$:
\begin{equation}
    e^{Bt}\frac{\diff r(t)}{\diff t} = -e^{Bt}Br(t) + e^{Bt}C\kappa(t) \implies \frac{\diff(e^{Bt} r(t))}{\diff t} = e^{Bt}C\kappa(t).
\end{equation}
With this identity, the solution becomes:
\begin{equation}
    r(t) = \int_{-\infty}^t  e^{-B(t-t')}C \kappa(t') d t'.
\end{equation}
In other words, $r(t)$ becomes a weighted sum of latent variables, with some memory kernel $e^{B(t-t')}C$. When $|t-t'| \gg \lambda_B^{-1}$, where $\lambda_B$ is the smallest eigenvalue of $B$, the exponential term is nearly zero, and thus the non-negligible contribution to the integral comes from a region around $|t-t'| \sim \lambda_B^{-1}$. 

An important insight from biology is that several orders of magnitude separate the synaptic and behavioral time scales. Specifically, in biological networks, the synaptic time scales (related to the eigenvalues of $B$) are $O(ms)$, whereas the behavioral time scales for computation are $O(s)$. In mathematical terms, this separation of scales can be stated by assuming that $\kappa(t')$ would be approximately constant for $|t-t'|\sim \lambda_B^{-1}$ such that $\kappa(t') \approx \kappa(t)$. Then, the embedding state variables, $r(t)$, become a linear transformation of the latent variables via:
\begin{equation} \label{eq:linear_embedding}
    r(t) \approx M \kappa(t), \quad \text{where} \quad  \left[\int_{-\infty}^{t} e^{-B(t-t')} \diff t'\right] C = \left[\int_{0}^{\infty} e^{-Bt'} \diff t'\right] C = B^{-1}C = M.
\end{equation}
Thus, in the case of a full-rank linear embedding, the dimension of the linear subspace spanned by the state variables $r(t)$ is at most equal to the dimension of the linear subspace spanned by the latent variables $\kappa(t)$. Yet, since nonlinearity is provided by the (nonlinear) encoding by definition, this class of networks can support complex computation other than linear dynamics. However, whether such constructions, where $\kappa(t)$ can both be a nonlinear function of $r(t)$ but still have the same linear dimensionality as $r(t)$, are feasible remains to be explored.

\subsubsection*{Nonlinear encoding and nonlinear embedding}

A typical autoencoder model involves reconstructing neural activities using low-dimensional bottlenecks. Changing the output from neural activities to time derivative, to be consistent with the definition of a dynamical embedding, could enable defining a model with nonlinear encoding and nonlinear embedding. Since this relationship captures the linear encoding and nonlinear embedding case, which we show can lead to universal approximators, the resulting RNN architectures can allow complex computations. However, it is not exactly clear under which conditions the latent variables follow closed-form dynamical system equations. Therefore, we leave it to future work to study these models. 

\subsubsection*{Linear encoding and nonlinear embedding} 

This form of linear encoding and nonlinear embedding relationship leads to well-known recurrent neural network architectures, as well as many unknown versions. Specifically, as long as the neuronal decay times are homogeneous, \textit{i.e.}, $A = \tau I$, with linear encoding, defined as $\kappa(t) = N r(t)$ following the above definitions, the property in Eq. (\ref{eq:embedding_condition}) is satisfied regardless of the embedding map since:
\begin{equation}
     \dot \kappa(t) = \nabla \phi(r(t)) \; \dot r(t) =  N \dot r(t) = -\frac{1}{\tau}\kappa(t) + \frac{1}{\tau} N\varphi(\kappa(t),u(t)) =g(\kappa(t),u(t)).
\end{equation}
As we show below, even the choice of a simple linear + nonlinear form for the embedding map can endow the network with a universal approximation property. As we discuss below, this form of encoding-embedding relationship can also explain several existing approaches to dynamical modeling of neural activities.

\subsection*{\color{subsectioncolor}Existing models in encoding-embedding framework}

As an illustration of our framework, we first study the static autoencoder models, which are defined with a direct map from latent variables to neural activities, in contrast with our dynamical embedding mapping latent variables to the time derivative of the neural activities. Then, we reproduce the prior results on low-rank RNNs \cite{mastrogiuseppe2018linking,dubreuil2022role,beiran2021shaping,valente2022extracting}, and later a deterministic limit of recurrently switching linear dynamical systems \cite{linderman2017bayesian,vinograd2024causal,nair2023approximate}.

\subsubsection*{Static autoencoder models}

A reasonable question one might ask is why define the embedding in a dynamical manner as we have done in Definition \ref{def:embedding}. To answer this question, we now consider a more traditional autoencoder model, assuming zero input for simplicity. In this model, the embedding map is static and therefore becomes a ``decoder'' of the neural activities (not of their time dynamics) from the latent variables:
\begin{equation} 
\begin{split}
      \kappa(t) &= \phi(r(t)), \\
        r(t) &= \tilde \varphi(\kappa(t)).
\end{split}
\end{equation}
As we noted in Eq. \eqref{eq:embedding_condition}, the time derivative of the latent variable depends on $\dot r$ following:
\begin{equation}
    \dot \kappa(t) = \nabla \phi(r(t)) \dot r(t).
\end{equation}
In Definition \ref{def:embedding}, since the dynamical embedding constrains $\dot r(t)$, fitting the model automatically provided the time evolution for both $\kappa(t)$ and $r(t)$. But, if we assume a static embedding, then we simply find:
\begin{equation}
    \dot r(t) = \nabla \tilde \varphi(\kappa(t)) \dot \kappa(t) \implies \dot r(t) = \nabla \tilde \varphi(\kappa(t)) \nabla \phi(r(t)) \dot r(t).
\end{equation}
This is a tautology, following from the implicit assumption that $r(t) = \tilde \varphi(\kappa(t)) = \tilde \varphi(\phi(r(t)))$. In other words, such a static model construction provides no information about how latent variables or, equivalently, how neural activities evolve in time. Then, as long as $\dot \kappa(t)$ follows a self-sufficient set of equations ($\dot \kappa(t) = g(\kappa(t))$ for some $g$), it is trivial to prove that for any pair of static maps $\tilde \varphi$ and $\phi$, one can always find an equivalent $\varphi$ satisfying Definition \ref{def:embedding} following:
\begin{equation}
    \dot r(t) = \nabla \tilde \varphi(\kappa(t)) \dot \kappa(t) \implies \dot r(t) = -r(t) + \underbrace{ \tilde \varphi(\kappa(t)) + \nabla \tilde \varphi(\kappa(t)) g(\kappa(t))}_{\varphi(\kappa(t))}.
\end{equation}
Thus, though it is possible to fit a post hoc model on a statically learned latent variable system such that $\dot \kappa(t) = g(\kappa(t))$, the resulting fit would be subsumed by the family of the dynamical embedding maps we introduced in Definition \ref{def:embedding}. Another relevant work enforces linear evolution in the latent variables and trains an autoencoder model for the encoding and (static) embedding maps \cite{abbaspourazad2024dynamical}, which is limited in its modeling capabilities due to the linearity assumption. However, the encoding-embedding framework, which enables the joint identification of latent variables and their temporal dynamics, contains these (and potentially more) solutions by definition, is therefore a more general approach to modeling latent variables with self-sufficient time dynamics. Below, we show two such examples that effectively enforce a linear map between latent variables and neural activities, $r(t) = \tilde \varphi(\kappa(t)) = C \kappa(t) + b$ for some linear parameters $C$ and $b$.

\subsubsection*{Leaky current RNNs}
Next, we reproduce the previous work on low-rank leaky current RNNs within the encoding-embedding framework:
\begin{equation} \label{eq:lcrnn}
    \tau \dot r (t)= -r(t) + W^{\rm rec}\tanh(r(t)) + W^{\rm in} u(t) +b,
\end{equation}
where $\tau \in \mathbb R$ is the neuronal time scales, $u(t) \in \mathbb{R}^{N_{\rm in}}$ a vector of inputs, $W^{\rm rec} \in \mathbb R^{N_{\rm rec} \times N_{\rm rec}}$, $W^{\rm in} \in \mathbb R^{N_{\rm rec} \times N_{\rm in}}$, and $b \in \mathbb R^{N_{\rm rec}}$ are weights and bias parameters. 

Previous work has established that this architecture, when subjected to a low-rank constraint such that $W^{\rm rec} = M N$ with $M \in \mathbb R^{N_{\rm rec}\times K}$ and $N \in \mathbb R^{K\times N_{\rm rec}}$ defined as before, can support a self-sustained latent dynamical system if neural activities are constrained in low-dimensional subspaces (reasons for which will be clear below) \cite{beiran2021shaping,valente2022extracting}. For completeness, following \cite{valente2022extracting}, we briefly summarize the mathematical steps. Recent work has introduced a set of latent variables, $\kappa(t) \in \mathbb R^K$, following \cite{beiran2021shaping}:
\begin{equation} \label{eq_s17}
    r(t) = M \kappa(t) +  W^{\rm in} v(t) + r_\perp(t) + b,
\end{equation}
where $v(t) \in \mathbb R^{N_{\rm in}}$ is a function of inputs, \textit{i.e.}, $\dot \tau v(t) = -v(t) + u(t)$, and $r_\perp(t)$ is orthogonal to the linear subspace,  spanned jointly by the column vectors of $M$ and $W^{\rm in}$. Inputting this into Eq. (\ref{eq:lcrnn}), Ref. \cite{beiran2021shaping} arrives at the set of equations:
\begin{subequations} \label{eq:latent_equations_lcrnn}
    \begin{align}
        \tau  \dot \kappa(t) &= - \kappa(t) + N \tanh(r(t)), \\
        \tau \dot v(t) &= - v(t) + u(t) , \\
        \tau \dot r_\perp (t) &= -r_\perp (t).
    \end{align}
\end{subequations}
The last equation implies that any perturbation in the orthogonal direction would die out exponentially. Though, any finite deviation, $r_\perp (t)$, means that the first equation is not self-sufficient since $r(t)$ depends on $r_\perp(t)$ in addition to $\kappa(t)$ and $v(t)$. The second equation simply suggests that a transformed version of the inputs enters the latent dynamical system. 

If we enforce the condition that $r_\perp (t) = 0$ and assuming that $M$ and $W^{\rm in}$ are full-rank, Eq. \eqref{eq_s17} implies that the linear dimensionality of the neural activities is equal to the dimensionality of the latent dynamical system plus the number of inputs. Under these assumptions, the first equation can be rewritten to be self-contained as:
\begin{equation} \label{eq:lcrnn_latent_circuit}
    \tau  \dot \kappa(t) = - \kappa(t) + N \tanh\left( M \kappa(t) +  W^{\rm in} v(t) + b \right).
\end{equation}
It should be noted that previous work has established this latent dynamical system, defined in \eqref{eq:lcrnn_latent_circuit}, as \emph{universal approximators} of any $K$ dimensional dynamical system in the limit $N_{\rm rec} \to \infty$ \cite{beiran2021shaping}.

Now, if one assumes (for simplicity) that the column vectors of $M$ and $W^{\rm in}$ are all orthogonal to each other, Eq. (\eqref{eq_s17}) can be inverted to define the encoding map as \cite{valente2022extracting}:
\begin{equation}
    \kappa_p(t) = \frac{1}{||m^{(p)}||_2^2} m^{(p)T} \left[x(t) -b \right],
\end{equation}
where $m^{(p)}$ is the $p$th column vector of $M$. Notably, this is a linear function. Similarly, using Eq. \eqref{eq:latent_equations_lcrnn}, we can define the embedding map as:
\begin{equation}
    \tau \dot x(t) = - x(t) + \underbrace{W^{\rm rec} \tanh\left(  M \kappa(t) +  W^{\rm in} v(t) + b \right) + W^{\rm in} u(t) + b}_{\varphi(\kappa(t),u(t))},
\end{equation}
where $\varphi(\kappa(t),u(t))$ is a nonlinear function, hence this is a nonlinear embedding. Therefore, our encoding and embedding framework can explain the latent dynamical system in regularly studied low-rank RNN architectures.

\subsubsection*{Recurrently switching linear dynamical systems}

We next discuss a variant of the latent variable models known as recurrently switching linear dynamical systems or rSLDS \cite{linderman2017bayesian,vinograd2024causal,nair2023approximate}. These models start with defining the latent states $\{Z_1,Z_2,\ldots,Z_L\}$ such that for a given time, the network is in one of the given states $Z[s] := Z_i$ for all $s$ and some $i=1,\ldots,L$. (Here, we use the bracket notation $\cdot[s]$ to refer to discrete variables, where $s$ corresponds to the discretized time.) Then, depending on the state, the latent variables follow a discrete linear dynamical system:
\begin{equation} \label{eq:rslds_or}
    \kappa[s+\Delta s] = A_{Z[s]} \kappa[s] + b_{Z[s]} + \sigma_\kappa,
\end{equation}
where $A_{Z[s]} \in \mathbb R^{K\times K}$ and $b_{Z[s]} \in \mathbb R^K$ are state dependent parameters and $\sigma_\kappa \in \mathbb R^K$ is some noise term. Similar to leaky current RNNs, in this model, the discretized neural activities are assumed to follow a linear relationship:
\begin{equation} \label{eq_s23}
    r[s] = M \kappa[s] + d,
\end{equation}
for some $M \in \mathbb R^{N_{\rm rec} \times K}$ and $d \in \mathbb R^{N_{\rm rec}}$. The traditional rSLDS model also considers a transition probability $p(Z[s+\Delta s],Z[s] = k, \kappa_t)$,
which is often defined via a linear term plus softmax function on $\kappa$. In the most general case, it is not possible to assign individual states to particular $\kappa$, since the states follow a stochastic process. Our encoding-embedding framework above, however, focuses on deterministic networks. Fortunately, in practice, it is often possible to assign approximate states to distinct latent variable combinations \cite{vinograd2024causal,nair2023approximate}. Inspired by this, we update Eq. \eqref{eq:rslds_or} as:
\begin{equation}
    \kappa[s+\Delta s] = \sum_{i=1}^L \left[ A_{Z_i} \kappa[s] + b_{Z_i} + \sigma_\kappa \right] \mathbbm 1[Z[s] =  Z_i](\kappa[s]),
\end{equation}
where $\mathbbm 1[Z[s] =  Z_i](\kappa[s])$, a function of $\kappa$, is one for a set of $\kappa$ variables that correspond to the state $Z_i$, zero otherwise. Then, ignoring the noise terms, we can turn the rSLDS equations into continuous form as:
\begin{subequations}
    \begin{align}
        r(t) &= M \kappa(t) + d, \\
        \dot \kappa(t) &= \sum_{i=1}^L \left[ \tilde A_{Z_i} \kappa(t) + \tilde b_{Z_i}\right] \mathbbm 1[Z(t) =  Z_i](\kappa(t)),
    \end{align}
\end{subequations}
where we redefined the parameters to account for changes during the transition from discrete to continuous representation. As a first note, the latent dynamical system is indeed self-consistent. Second, similar to before, we can enforce low-dimensionality on the neural activities and define a \emph{linear} encoding map as:
\begin{equation}
    \kappa_p(t) = \frac{1}{||m^{(p)}||_2^2} m^{(p)T} \left[r(t) -d \right],
\end{equation}
where we once again assumed orthogonality between column vectors of $M$. Then, we have:
\begin{equation}
     \dot r(t) = M \dot \kappa(t) +  d =  \sum_{i=1}^L \left[ M \tilde A_{Z_i} \kappa(t) + M \tilde b_{Z_i} \right] \mathbbm 1[Z(t) =  Z_i](\kappa(t)) + d +\left[ M\kappa(t) +d \underbrace{- M \kappa(t) - d}_{-r(t)}\right],
\end{equation}
which leads to the \emph{nonlinear} embedding relationship:
\begin{equation}
    \dot r(t) = -r(t) + M\kappa(t) + 2d + \sum_{i=1}^L \left[ M \tilde A_{Z_i} \kappa(t) + M \tilde b_{Z_i} \right] \mathbbm 1[Z(t) =  Z_i](\kappa(t)).
\end{equation}
This function is piecewise linear (hence nonlinear) in $\kappa(t)$, where the regions of linearity are marked by the internal states $Z_i$. Consequently, we have reformulated the deterministic limit of rSLDS as an instance of our encoding-embedding framework with linear encoding and nonlinear embedding.

\subsubsection*{Motivations for moving beyond the traditional models}

Above, we discussed that static autoencoder models lead to under-determined systems for latent variables, \textit{i.e.}, no particular time evolution equation ($\dot \kappa = g(\kappa(t))$ for some $g$) is constrained by the data during fitting. As a step forward, leaky-current RNNs and rSLDS constitute two common dynamical models utilized in the field. These models can support self-sufficient latent dynamical systems in their neural activities, hence satisfying the encoding-embedding relationship we introduced in this work. However, both methods have the same shortcomings: 
\begin{enumerate}
    \item Following Eqs. \eqref{eq_s17} and \eqref{eq_s23}, two models assume that neural activities should lie on a flat manifold with the same dimensionality as the LPUs ($K + N_{\rm in})$. However, empirical evidence suggests that neural activities are instead high-dimensional \cite{stringer2019high,manley2024simultaneous}. Thus, these two models either model only a portion (first few principal components) of the data \cite{valente2022extracting}, or simply focus on the reconstruction accuracy of the latent space \cite{nair2023approximate,vinograd2024causal}. New models that no longer enforce low-dimensionality constraints may prove beneficial for modeling neural and the underlying latent dynamics.
    \item The linear relationship from the LPUs to the neural activities prevents these models from explaining complex tuning properties of individual neurons, as all neurons will monotonically increase or decrease their firing rates with respect to $\kappa$ depending on the specific linear weight, $M$. We discuss this further below when we introduce the tuning curves under LPU. 
\end{enumerate}

\subsection*{\color{subsectioncolor}New class of RNNs in encoding-embedding framework}

Having studied the encoding-embedding maps across several models, we are now ready to define LPUs:
\begin{definition}[Latent processing units] \label{def2}
    We define low-dimensional dynamical systems, spanned by $\kappa(t) \in \mathbb R^K$, as LPUs if they satisfy two conditions:
\begin{enumerate}[label=(\roman*)]
    \item \textbf{Accessible embedding of self-sufficient latent dynamics:} The low-dimensional dynamical systems, $\kappa(t)$, are embedded in high-dimensional neural activities, $r(t) \in \mathbb R^{N_{\rm rec}}$, following Definition \ref{def:embedding} above. With this construction, $\kappa(t)$ are both accessible from $r(t)$ (without depending on inputs or other terms) through $\kappa(t) = \phi(r(t))$ and support a self-sufficient latent dynamical system since $\dot \kappa(t) = g(\kappa(t),u(t))$ does not explicitly depend on $r(t)$.
    \item \textbf{Universal approximation property:} The encoding-embedding relationship has trainable parameters, which endow the low-dimensional latent dynamical system with the ability to become universal approximators of $K$ dimensional dynamical systems (with $N_{\rm in}$ inputs) in the limit $N_{\rm rec} \to \infty$.
\end{enumerate}
\end{definition}
For the purpose of this work and for all theorems below, we will also make the additional assumption that LPUs are linearly encoded. Though future work may consider LPUs without linear encoding, linear encoding has several desirable properties that will be apparent below.

With this definition, it is straightforward to show that the leaky current RNNs, introduced above, satisfy these conditions under the $r_{\perp}(t) = 0$ assumption, \textit{i.e.}, in which neural activities are restricted to lie in low-dimensional subspaces \cite{mastrogiuseppe2018linking}. However, as noted earlier, experimental work suggests a high linear dimensionality for neural activities \cite{stringer2019high,manley2024simultaneous}, rendering these approaches incompatible for reconstructing neural activities. Fortunately, the encoding-embedding framework we introduced earlier allows us to extract LPUs from a broader class of artificial and biological neural networks. Interestingly, as we show below, following one of the simplest encoding and embedding functions can lead to an RNN architecture that has been severely understudied compared to its counterparts \cite{masse2019circuit,dinc2023cornn}. Consequently, its latent dynamics and capabilities remained mystery to date. Here, we first derive the time evolution equations of the LPUs for this class of RNNs and then prove a universal approximation theorem for this architecture. Then, we generalize the LPUs to a broader class of RNNs, which we refer to as ``Linearly-encoded latent networks." Afterwards, we connect the LPUs, which explain the neural population coding dynamics, to the tuning curves of individual neurons, a traditional view of neural coding. Finally, we provide a brief discussion of the well-known identifiability issues with the latent variables, and why they are not of concern for studies with LPUs.

\subsubsection*{Latent processing units in leaky firing rate RNNs}
To define the LPUs in leaky firing rate RNNs, we assume linear encoding and nonlinear embedding functions. For the embedding, we assume one of the simplest forms: a linear function followed by nonlinearity, which leads to the equations:
\begin{subequations}
\begin{align}
    \kappa(t) &= N r(t), \\
    \tau \dot r(t) &= -r(t) + f\left( M \kappa(t) + W^{\rm in} u(t) + b \right),
\end{align}
\end{subequations}
where $M$, $N$, $W^{\rm in}$, and $b$ are defined as before, $\tau$ is the homogeneous neuronal decay time, and $f(.)$ is a pre-defined nonlinearity (taken as $f(.) = \tanh(.)$ throughout this work). Defining $W^{\rm rec} = MN$ as the weight matrix and some straightforward algebra reveal both the RNN equations and the latent dynamical system:
\begin{subequations} \label{eq:lfrnns}
    \begin{align}
        \tau \dot r(t) &= -r(t) + f\left( W^{\rm rec} r(t) + W^{\rm in} u(t) + b \right), \\
        \tau \dot \kappa(t) &= - \kappa(t) + N f\left( M \kappa(t) + W^{\rm in} u(t) + b \right).
    \end{align}
\end{subequations}
Moreover, as we show in the next theorem by naturally extending the existing proofs \cite{schafer2006recurrent,hornik1989multilayer}, leaky firing rate RNNs are also universal approximators of $K$-dimensional dynamical systems:
\begin{theorem}[Universal approximation theorem]  \label{thm1}
    Let a recurrent neural network follow the time dynamics:
    \begin{equation}
        \tau \dot r(t) = -r(t) + f\left( W^{\rm rec} r(t) + W^{\rm in} u(t) + b \right),
    \end{equation}
    where $r(t) \in \mathbb R^{N_{\rm rec}}$ are neural activities of $N_{\rm rec}$ neurons, $u(t) \in \mathbb R^{N_{\rm in}}$ inputs, $\tau \in \mathbb R$ is a neuronal time scale, $f(.)$ is some non-polynomial nonlinearity, $W^{\rm rec} = MN$ with $N \in \mathbb R^{K\times N_{\rm rec}}$ and $M \in \mathbb R^{N_{\rm rec} \times K}$, $W^{\rm in} \in \mathbb R^{N_{\rm rec} \times N_{\rm in}}$ input weights, and $b \in \mathbb R^{N_{\rm rec}}$ biases. Then, the latent dynamical systems, whose variables are defined by the linear projection, $\kappa(t) = N r(t) \in \mathbb R^K$, are universal approximators of $K$-dimensional dynamical systems with $N_{\rm in}$ pre-defined inputs, $\dot \kappa(t) = g(\kappa(t),u(t))$ for an arbitrary flow map $g(.)$, in the limit $N_{\rm rec}\to \infty$, i.e., when the number of rows in $M$ and $W^{\rm in}$ goes to infinity.
\end{theorem}
\begin{proof}
    The dynamical system equations followed by $\kappa(t)$ are given as:
    \begin{equation} \label{eq:lfrnn_latent_circuit}
        \tau \dot \kappa(t) = - \kappa(t) + N f\left( M \kappa(t) + W^{\rm in}u(t) + b \right).
    \end{equation}
    Here, the term $\tau \dot \kappa(t) + \kappa(t)$, is equal to a feed forward network with a single hidden layer that takes $\kappa$ as input. Such a network can approximate arbitrary functions (barring some regularity conditions \cite{leshno1993multilayer}) $\mathbb R^{K+N_{\rm in}} \to \mathbb R^K$ as long as $f(.)$ is a non-polynomial nonlinearity from the standard universal approximation theorem \cite{leshno1993multilayer}. Setting the function approximated by this network to $\tau g(\kappa,u) + \kappa$, we get:
    \begin{equation}
         N f\left( M \kappa + W^{\rm in} u+ b \right)\approx \tau g(\kappa,u) + \kappa,
    \end{equation}
    and thus: $\dot \kappa(t) \approx g(\kappa(t),u(t))$. 
    
    As an alternative proof, after a the transformation of the input $u \to v$ following $\tau \dot v(t) = -v(t) + u(t)$, the latent dynamical systems in Eq. (\ref{eq:lfrnn_latent_circuit}) are the same as the ones in Eq. (\ref{eq:lcrnn_latent_circuit}), and therefore are universal approximators following \cite{beiran2021shaping}.
\end{proof}

The leaky firing rate RNNs in Eq. (\ref{eq:lfrnns}) satisfy the defining conditions of encoding and embedding maps and their latent dynamical systems are universal approximators: thus, they are said to contain LPUs. Notably, unlike leaky current RNNs studied heavily in the field \cite{mastrogiuseppe2018linking,dubreuil2022role,valente2022extracting,beiran2021shaping}, this alternative architecture does not require any restriction on neural dynamics in order for the latent dynamical system to be self-contained but has the same universally expressive power.

\subsubsection*{Multi-synapse latent networks (MSLNs)}

For the leaky firing rate RNNs, we assumed a simple nonlinear embedding that consisted of a linear term followed by a neuron-wise nonlinearity. Here, we discuss how to incorporate a broader class of nonlinear embeddings and still retain the encoding-embedding relationship in a new class of RNNs that we call multi-synapse latent networks. To do so, we first define a current shared between neurons (inline with \cite{perich2021inferring}), using a weight matrix $W^{\rm rec} \in \mathbb R^{N_{\rm rec} \times N_{\rm rec}}$, as:
\begin{equation}
    z(t) = W^{\rm rec} r(t).
\end{equation}
Intuitively, this means that neural activities induce currents in other neurons through synaptic connections, which are stored inside the weight matrix $W^{\rm rec}$. Enforcing low-rank constraint on this weight matrix, $W^{\rm rec} = MN$ with $M$ and $N$ defined as before, leads to a definition for latent variables through linear encoding:
\begin{equation}
    \kappa(t) = N r(t).
\end{equation}
In this analogy, one can consider the firing rates in Eq. \eqref{eq_lfrnn} as:
\begin{equation}
    \tau r(t) = - r(t) + \tanh(z(t)),
\end{equation}
where $\tanh(z(t))$ softly thresholds the currents, transforming them into continuous activities, whereas $\tau \dot r(t) = -r + \cdot$ performs a low-pass filtering in time. This equation describes an artificial synapse, with a single ``current channel'' $z(t)$. Yet, such simple constructions may be insufficient, for instance, to model complicated (and nonlinear) dendritic integration processes (\cite{gidon2020dendritic}), as these processes are likely more complicated than can be captured by simple linear sums of neural activities followed by soft-thresholding discussed above. To introduce a broader class of RNNs that can model complicated nonlinear embedding maps, and thereby dendritic processes, we introduce multi-synapse latent networks:
\begin{definition}[Multi-synapse latent networks] \label{def3}
We broadly define multi-synapse latent networks as those following the dynamical system equations:
    \begin{equation}
    \tau \dot r_i(t) = -r_i(t) + f_i(u(t),[W^{(1)}r(t)+b^{(1)}]_i,[W^{(2)}r(t)+b^{(2)}]_i,\ldots),
\end{equation}
where $\tau$ is the neuronal time scale, $u\in \mathbb R^{N_{\rm in}}$ are input variables, $\{W^{(j)}\}$ with $W^{(j)} \in \mathbb R^{N_{\rm rec} \times N_{\rm rec}}$ are collection of weights and $\{b^{(j)}\}$ the corresponding biases, $r(t) \in \mathbb R^{N_{\rm rec}}$ are the neural activities. In this architecture, any two neurons are allowed to create multi-synapses with each other. The corresponding currents are denoted as $\{z^{(1)}, z^{(2)}, \ldots\}$. A set of nonlinear transfer functions, $f(.) =\{ f_1(\cdot),\ldots,f_{N_{\rm rec}}(\cdot)\}$, takes currents into neurons as inputs and transforms the full set of currents into the neural activities. 
\end{definition}

In MSLNs, similar to before, the currents ($z^{(j)}$) can be used to define a set of latent variables. Specifically, if we assume $W^{(j)} = M^{(j)} N^{(j)}$ with $M^{(j)} \in \mathbb R^{N_{\rm rec} \times K_j}$ and $N^{(j)} \in \mathbb R^{K_j \times N_{\rm rec} }$, we can define a subset of latent variables via $\kappa^{(j)}(t) = N^{(j)}r(t)$ with $\kappa^{(j)}(t) \in \mathbb R^{K_j}$. Then, the full set of latent variables is given by the group $\kappa = \{ \kappa^{(1)}, \ldots \}$ and constitutes a $K = \sum_{j} K_j$ dimensional dynamical system, which is guaranteed by the linearity of the encoding as discussed above. As long as the transfer functions, $f(\cdot)$, are chosen to satisfy the universal approximation property (see leaky firing rate RNNs for a simple example), the variables, $\kappa(t) \in \mathbb R^K$, constitute a LPU.

For simplicity of notation, we ignore the biases moving forward. For this general class of RNNs, the encoding is once again linear, but the embedding takes the form:
\begin{equation} \label{eq:neural_general}
    \tau \dot r_i(t) = -r_i(t) + f_i(u(t),[M^{(1)}\kappa^{(1)}(t)]_i,[M^{(2)}\kappa^{(2)}(t)]_i,\ldots),
\end{equation}
leading to the latent dynamical system:
\begin{equation} \label{eq:latent_general}
    \tau \dot \kappa^{(j)}_p(t) = -\kappa^{(j)}_p(t) + \sum_{k=1}^{N_{\rm rec}} N^{(j)}_{pk} f_k(u(t),[M^{(1)}\kappa^{(1)}(t)]_k,[M^{(2)}\kappa^{(2)}(t)]_k,\ldots),
\end{equation}
where $\kappa^{(j)}_p \in \mathbb R^{K_j}$ is the $p$th latent variable defined through the weights in the $j$th current channel. This is a closed dynamical system, respecting the encoding and embedding framework that we introduced above. However, whether this becomes a universal approximator depends on the properties of $f(.)$, as noted above.

The full extent of LPUs that can be extracted with our encoding- embedding framework is left for future work. For example, the leaky firing rate RNN architecture is an example of this class, for which there is only one current per neuron and $f(.)$ is a predefined non-polynomial nonlinearity. Another example would be to define $f(.)$ as a multilayer perceptron. This approach could be helpful for studying dendritic computations that use multiple currents through distinct synapses \cite{trautmann2021dendritic,gidon2020dendritic}. For this work, we mainly focus on the leaky firing rate RNNs, though we will use the general architecture in several of our proofs to ensure general applicability. In the proofs moving forward, for simplicity of notation, we concatenate all latent variables into a single vector, $\kappa(t) = \{\kappa^{(1)},\ldots\}$ such that:
\begin{equation}
   \tilde f_i\left(u(t),\{ m^{(p)}_i \kappa_{p}\}\right):=\bar f_i(u(t),[M^{(1)}\kappa^{(1)}(t)]_i,[M^{(2)}\kappa^{(2)}(t)]_i,\ldots),
\end{equation}
where $\kappa \in \mathbb R^{K}$ with $K =\sum_j K_j$, $\tilde f_i$ is the redefined collection of nonlinear transfer functions, and $m^{(p)} \in \mathbb R^N$ is the embedding vector of $\kappa_p$, corresponding to one of the column vectors in the collection $\{M^{(i)}\}$. Moreover, similar to the embedding weights, we define $\frac{1}{N}n^{(p)}$ as the encoding weight of $\kappa_p(t)$, which corresponds to one of the row vectors in the collection $\{N^{(i)}\}$ above. 

\subsection*{\color{subsectioncolor}From latent processing units to neural tuning curves}

\subsubsection*{Latent time scales}

Neural tuning curves describe how individual neurons respond to varying stimuli. These curves have historically served as a tool for characterizing preferential activation by specific sensory inputs. Neurons in different brain regions exhibit distinct tuning properties, such as orientation tuning in V1 \cite{hubel1959receptive}, frequency tuning in the auditory cortex \cite{fischer2011owl}, spatial tuning in the hippocampus \cite{mao2021spatial,o1976place}, and gravity tuning in the thalamus \cite{laurens2016gravity}. Although the LPU framework focuses on population-level coding dynamics \cite{yuste2015neuron}, it can also account for neural tuning curves commonly studied in the literature. 

Here, we provide the mathematical details behind the connection between LPUs and neural tuning curves. First, we state a proposition that formalizes the notion that latent variables do not have to adhere to the same time scale ($\tau_\kappa\neq \tau$) as the neural activities:
\begin{proposition}[Arbitrary latent time scale]\label{prop1}
    Let $r(t) \in \mathbb R^{N_{\rm rec}}$ be the neural activities embedding a low-dimensional LPU with the variables $\kappa(t) \in \mathbb R^K$. 
    LPUs can operate at any arbitrary time scale $\tau_\kappa$ such that:
    \begin{equation}
        \tau_\kappa \dot \kappa(t) = - \kappa(t) + g(\kappa(t),u(t)),
    \end{equation}
    where $g: \mathbb R^{K+N_{\rm in}} \to \mathbb R^K$ is a continuous function defining the flow map for the desired dynamics.
\end{proposition}
\begin{proof}
    The proof starts with the time dynamics of the LPU in Eq. \eqref{eq:latent_general}. Here, we can define the collection of functions $h(\kappa) = (h_1(\kappa,u), \ldots, h_K(\kappa,u))$ as:
    \begin{equation}
        h_r(\kappa(t),u(t)) = \frac{1}{N_{\rm rec}}\sum_{i=1}^{N_{\rm rec}} n^{(r)}_i \tilde f_i\left(u(t),\{ m^{(p)}_i \kappa_{p}(t)\}\right),
    \end{equation} 
    with which Eq. \eqref{eq:latent_general} transforms into:
    \begin{equation}
        \tau \kappa(t) = - \kappa(t) + h(\kappa(t),u(t)).
    \end{equation}
    Since $\kappa(t)$ constitute a LPU, $h(\kappa,u)$ is, by the definition of a LPU, a universal approximator for the continus functions $\mathbb R^{K+N_{\rm in}} \to \mathbb R^K$. Then, we can use it to approximate $h(\kappa,u) \approx (1-\frac{\tau}{\tau_\kappa})  \kappa + \frac{\tau}{\tau_\kappa} g(\kappa,u)$ for some arbitrary function $g(\kappa,u)$ that drives the computation in $\kappa$ on the time scale $\tau_\kappa$. Then, the latent equation turns into:
    \begin{equation}
        \tau_\kappa \dot \kappa(t) = - \kappa(t) + g(\kappa(t),u(t)). 
    \end{equation}
    This concludes the proof.
\end{proof}
It is worth noting that in the limit $\tau_\kappa/\tau \to \infty$, $h(\kappa)$ becomes an identity map. For finite $\tau_\kappa$, designing a LPU with longer timescales compared to individual neurons ($\tau_\kappa \gg \tau$) requires the ability to model a function that only slightly changes from the identity map $h(\kappa) = \kappa + \frac{\tau}{\tau_\kappa}\left[g(\kappa) - \kappa\right]$. Although the universal approximation is achieved in the limit $N_{\rm rec} \to \infty$, practically achieving such a separation (as opposed to setting, \textit{e.g.}, $\tau = \tau_\kappa$) likely requires large $N_{\rm rec}$. We discuss a related numerical experiment below. 

\subsubsection*{Defining tuning curves with respect to latent variables}
The separation of timescales between latent variables and neural activities suggests that the former can remain approximately invariant for short time intervals such that $\kappa(t') \approx h(\kappa(t'))$ for $t' \in [t,t+O(\tau)]$, but vary rapidly in time scales $\tau_\kappa\gg \tau$. Realistically, one can imagine $\tau \sim O(1-10ms)$ being the time scale of neuronal processes and $\tau_\kappa \sim O(100-1000ms)$ being the time scale of behavior. In this case, despite the dynamical nature of the embedding, it is possible to define approximate tuning curves as if the embedding function was static. To do so, we first define the concept of $\epsilon$-stability:
\begin{definition}[$\epsilon$-stability]
    We refer to a set of variables $\kappa(t) \in \mathbb R^K$ as $\epsilon$-stable for the time window $[t,t+T]$ if the following inequality is satisfied for some small $\epsilon$:
    \begin{equation}
    \forall i \; \forall t',t'' \in [t,t+T] \quad |\kappa_i(t') - \kappa_i(t'')| \leq \epsilon.
    \end{equation}
\end{definition}
Here, $\epsilon$ is a user-defined (small) parameter which can, for example, be set to the noise level. Then, if we assume that $\kappa(t)$ are $\epsilon$-stable for the time interval $t \in [t_1,t_1+T]$ with $T \gg \tau$, in the absence of any input, the embedding map becomes:
\begin{equation}
    \tau \dot r_i(t) = -r_i(t) + \tilde f_i\left(\{ m^{(p)}_i \kappa_{p}\}\right) + O(\epsilon),
\end{equation}
where $\kappa_p := \frac{1}{T} \int_{t_1}^{t_1+T}  \kappa(t) \diff t$ is fixed. Ignoring the terms in $O(\epsilon)$, the above equation can be solved to yield:
\begin{equation}
    r(t) \approx \left(r(t_1) - \tilde f_i\left(\{ m^{(p)}_i \kappa_{p}\}\right)\right) e^{-t/\tau} + \tilde f_i\left(\{ m^{(p)}_i \kappa_{p}\}\right),
\end{equation}
for $t \in [t_1,t_1+T]$. Therefore, after the transient response dies out, the neural activities (on average) follow the (internal) tuning curves:
\begin{equation} \label{eq:tuning}
    r_i(\kappa) = \tilde f_i\left(\{ m^{(p)}_i \kappa_{p}\}\right).
\end{equation}
Notably, tuning curves are often defined with respect to external experimentally relevant variables, not internal latent variables. We discuss below how we reconcile these two concepts through utilizing an inherent symmetry of the LPUs. Moreover, it should be noted that the noisy nature of neural activities, \textit{e.g.}, addition of random noise to the embedding, would mean that this relationship likely holds on average, \textit{e.g.}, when averaged across multiple trials, and after the latent variables were stable for a while $\sim O(\tau)$, which is how tuning curves are often computed in practice.

\subsubsection*{Defining external tuning curves through identifiable latent processing units}

For $\tilde \kappa(t)$ to be considered a LPU, its dynamics should not depend explicitly on $r(t)$. Although latent dynamical systems with linear encoding automatically satisfy this criterion (see above), there are infinitely many transformations that can lead to equivalent LPUs with different latent variables. Here, by equivalence, we mean two conditions:
\begin{itemize}
    \item \textbf{Dynamical invariance:} After the transformation, $\kappa = S(\tilde \kappa)$, the new variables should constitute a LPU while the neural dynamics ($\dot r = F(r,u)$) remain invariant.
    \item \textbf{Output invariance:} A behavioral output from the LPU achieved through a ``decoder" from neural activities remains invariant.
\end{itemize}
Here, as a side note, the term decoder explicitly refers to the maps from the neural activities or the linearly encoded latent variables to the behavioral outputs. Now, as an example, consider the class of transformations, $\kappa = S \tilde \kappa$, where $S \in \mathbb R^{K\times K}$ is an invertible matrix. Such transformations respect both conditions, in which $\kappa$ constitutes an equally valid LPU with its equivalent linear and/or nonlinear behavioral decoders. (To see why, multiply every $\kappa(t)$ from the left by $S^{-1} S$ in Eq. (\ref{eq:latent_general}), and likewise for any decoders defined on $\kappa(t)$). Since $S$ is any invertible matrix, the LPU is said to have the ``$GL_K(\mathbb R)$ symmetry,'' where $GL_K(\mathbb R)$ stands for general linear group of degree $K$ over real numbers $\mathbb R$. This symmetry makes the LPUs non-identifiable. Then, how can we define the tuning curves? Which basis of $\kappa(t)$ is the ``real'' basis that the brain needs to keep track of?

The latent non-identifiability problem has been well known in independent component analysis (ICA), with established conditions for linear identifiability \citep{hyvarinen2000independent}. In the case of a nonlinear mixing function (called a decoder or a static embedding in our terminology) \citep{hyvarinen1999nonlinear}, recent work has provided nonlinear identifiability conditions \citep[reviewed in][]{hyvarinen2024identifiability}. Some nonlinear ICA methods have also been proposed for neural data, e.g., based on VAEs \citep{zhou2020learning} or contrastive learning \citep{schneider2023learnable}. In this work, we address this non-identifiability issue in two distinct ways: 

\begin{enumerate}
    \item For external observers who can keep track of external inputs and behaviorally relevant outputs, the $GL_K(\mathbb R)$ symmetry may be broken to align the latent space with experimental variables. We show an example of this when we discuss the LPUs subserving K-bit flip flop tasks (see also Fig. \ref{fig1}). This direction has some interesting implications. Specifically, since the nonlinearity is unknown but fixed, \textit{e.g.}, by the biophysical neuronal processes, the tuning properties depend strictly on the embedding vectors, $m^{(p)}$. Notably, if some $\kappa$ are aligned to external, behaviorally relevant variables/cues, having only one non-zero $m_i^{(p)}$ would correspond to sparse coding. In contrast, having multiple non-zero $m_i^{(p)}$ would lead to mixed-selective neurons \cite{rigotti2013importance}, providing a theoretical framework for future studies of these coding mechanisms.
    
    \item When discussing the universal decoding theorem in the following section, we argue that explicit identification of LPUs is not necessary for the networks to implement them or extract necessary information from them. Thus, though identification of latent variables may be experimentally beneficial to make sense of their neural coding patterns, networks implementing them do not need to identify them, and thereby their expressivity does not suffer from this non-identifiability problem. In other words, the brain does not need to keep track of tuning curves as opposed to researchers who do so for experimental convenience.
\end{enumerate}

Finally, it is worth noting that if $m_i^{(p)} = 0$ for all $p=1,\ldots, K$, \textit{i.e.}, the particular neuron $i$ is not coding for any of the latent variables, then the $GL_K(\mathbb R)$ symmetry cannot turn neuron $i$ into a ``coding'' neuron, since $S 0 = 0$ for any matrix $S \in GL_K(\mathbb R)$. 

\subsection*{\color{subsectioncolor}Decoding universal outputs from the latent processing units}

So far, we have studied the universality of computations that can be performed by the latent dynamical systems that we call LPUs. However, the results of these computations should be accessible to downstream units, \textit{e.g.}, motor neurons. In this section, we lay out a theoretical foundation and prove that linearly encoded LPUs allow extracting time-dependent (universal) readouts using linear decoders from neural activities. First, with the condition of increased dimensionality, we prove that LPUs allow linear readout of any arbitrary function of their state variables. Next, we illustrate that these readouts can be performed directly on neural activities, mitigating the identifiability issues of latent variables. Finally, we bring both results together into a universal decoding theorem (Theorem \ref{thm2}), which suggests that LPUs can subserve arbitrarily complex network output and/or animal behavior.

\subsubsection*{Sufficiency of linear readouts from extended latent processing units}

For a $K$-dimensional LPU with $B$ outputs defined as differentiable function of its latent variables, it is possible to design an extended LPU with $K+B$ dimensions that can provide the same $B$ outputs through a linear readout, as we formalize with a lemma and its corollary:
\begin{lemma}[Sufficiency of linear readouts from LPUs] \label{lm1}
    Let $r(t) \in \mathbb R^{N_{\rm rec}}$ be high-dimensional neural activities that can support LPUs up to $K$ dimensions with linear encoding maps. Let $\bar \kappa(t) \in \mathbb R^{K-1}$ be a linearly encoded LPU (in $r(t)$) following dynamical system equation $\dot{\bar \kappa}(t) = \bar G(\bar \kappa(t),u(t))$. Let the output, $o(t) \in \mathbb R$, of this system be defined as $o(t) = \bar \psi(\bar \kappa(t))$ for some, potentially nonlinear, differentiable function $\bar \psi: \mathbb R^{K-1} \to \mathbb R$. Then, $r(t)$ can also support a LPU, $\kappa(t) \in \mathbb R^K$, which follows a dynamical system equation $\dot  \kappa(t) = G(\kappa(t),u(t))$ such that  $\kappa_i(t) = \bar \kappa_i(t)$ for $i=1,\ldots, K-1$ and the output is $o(t) = \kappa_K(t)$, \textit{i.e.}, a linear readout.
\end{lemma}
\begin{proof}
    Let us consider the time derivative of the output in the original LPU:
    \begin{equation}
        \dot \kappa_K(t) = \dot o(t) = \nabla \bar \psi (\bar \kappa(t)) \dot{\bar \kappa}(t) =  \nabla \bar \psi (\bar \kappa(t)) \bar G(\bar \kappa(t),u(t)) := G_K(\kappa(t),u(t)),
    \end{equation}
    which is a function of $\kappa(t)$ and $u(t)$ only, respecting a self-contained dynamical system equation. Moreover, we can also define $\dot \kappa_i(t) = \bar G_i(\kappa_1(t),\ldots,\kappa_{K-1}(t),u(t)) = G_i(\kappa(t),u(t))$ for $i = 1,\ldots, K-1$, also self-contained. Thus, jointly we can define a new LPU with the following set of equations:
    \begin{equation}
        \dot \kappa(t) = G(\kappa(t),u(t)).
    \end{equation}
    In other words, there is a new LPU that can be supported by the neural activities, which stores not only the same latent variables as before in its first $K-1$ entries, but also sustains the output in its $K$th entry. 
\end{proof}

Then, one can extend the process described in Lemma \ref{lm1} to a scenario with $K$-dimensional LPU and $B$ outputs by induction. As a caveat, it is worth noting that Lemma \ref{lm1} does not guarantee that the neural tuning properties with respect to the original latent variables will remain unchanged. Similarly, the weight matrix that supports the computation before and after the extension does not necessarily remain invariant. Instead, Lemma \ref{lm1} suggests that the same group of neurons can be repurposed, perhaps with a different set of parameters, to support the extended LPU as long as $N_{\rm rec} \gg K$. Our focus in this work is on the invariance of latent dynamics, though future work may be interested in studying conditions that allow defining LPUs as truly mathematical constructs with no changes in any real observables, \textit{e.g.}, when also the neural activities remain invariant under these extensions/transformations. That being said, Lemma \ref{lm1} has a desirable corollary for the biological relevance of neural computation with LPUs:

\begin{corollary}[Linear decoding from neural activities] \label{cor1}
    An output defined on the linearly encoded LPUs as $o(t) = \kappa_K(t)$ can be linearly decoded from neural activities, $r(t)$.
\end{corollary}
\begin{proof}
    Since $\kappa_K(t)$ is linearly encoded, one can define the output as:
    \begin{equation}
        o(t) = \frac{1}{N_{\rm rec}} n^{(K)T}r(t),
    \end{equation}
    where $\frac{1}{N} n^{(K)T}$ is defined as the $K$th row vector of $N \in \mathbb R^{N_{\rm rec} \times K}$. This concludes the proof.
\end{proof}
The extension of the LPU to a slightly higher dimension may seem like a restriction. However, since we assume $K \ll N_{\rm rec}$, increases in the form of $B \sim O(K)$ do not change the scaling relationships, and may even be desirable for robust operation. It is worth highlighting that Lemma \ref{lm2} does not guarantee that an output from the LPU will be linearly accessible; it only means that they could be. A slightly larger (yet equivalent) dynamical system can be designed that contains the same latent variables as a subsystem, but also allows linear readout. By using slightly more resources, biological and/or artificial networks can, but not necessarily have to, allow linear readouts from their LPUs.

\subsubsection*{Linear readouts from LPUs without explicit identification} 

Corollary \ref{cor1} suggests that the output, if linearly encoded into a latent variable, can be accessed with a linear readout from neural activities. However, as we discussed in the previous section, latent variables have $GL_K(\mathbb R)$ symmetry, meaning that assigning the behavioral output to a particular variable may break the symmetry in a biologically irrelevant manner. In fact, networks may not keep track of their latent variables explicitly, but rather use the circuits implicitly through matrix multiplication as in Definition \ref{def3}. 

To keep generality, instead of enforcing $\kappa_K(t) = o(t)$, we define the output as a linear readout from LPUs as $\tilde \psi(\kappa(t)):=\tilde W_{\rm out} \kappa + \tilde b_{\rm out}$ for some parameters $W_{\rm out} \in \mathbb R^{B \times K}$ and $b_{\rm out}\in \mathbb R^B$. This definition respects the $GL_K(\mathbb R)$ symmetry, since the transformation on the latent variables would still lead to another linear readout. Fortunately, as long as $\tilde \psi(\kappa(t))$ can be computed directly from $r(t)$, the behavioral readout can be achieved without having to identify the state variables $\kappa(t)$, making this symmetry irrelevant for the output. We formalize this observation with the following Lemma:
\begin{lemma}[Equivalence of decoding from neural and latent activities] \label{lm2}
    Let $r(t) \in \mathbb{R}^{N_{\rm rec}}$ represent neural activities and $\kappa(t) \in \mathbb{R}^K$ be the corresponding latent variables obtained via a linear encoding, where $\kappa(t) = N r(t)$ for some rank-K matrix $N \in \mathbb{R}^{K \times N_{\rm rec}}$. Let $P_N$ be the projection operator onto the row space of $N$, and $\hat y \in \mathbb{R}^B$ denote the behavior readouts from the network, where $N_{\rm rec} \gg K \geq B$. We define a linear decoder on neural activities as any function $\psi$ such that: $\hat{y} := \psi(r) = W_{\rm out} r + b_{\rm out}$, where $W_{\rm out} \in \mathbb{R}^{B \times N_{\rm rec}}$ and $b_{\rm out} \in \mathbb{R}^B$ are trainable parameters. Then, for any latent decoder $\tilde{\psi}: \mathbb{R}^K \to \mathbb{R}^B$, there is at least one equivalent linear decoder $\psi(r)$ such that
\begin{equation}
    \hat{y} := \psi(r) = \tilde{\psi}(\phi(r)).
\end{equation}
     Conversely, for any linear decoder $\psi(r)$ satisfying the relationship $W_{\rm out}(1-P_N) = 0$ (referred to as ``linear bottleneck''), and only for those that do, there is at least one equivalent linear latent decoder $\tilde \psi(\kappa)$.
\end{lemma}

\begin{proof} 
To start the proof, we construct a general linear latent decoder, $\tilde \psi(r):\mathbb R^K \to \mathbb R^B$, from the latent variables as:
\begin{equation}
  \hat y:= \tilde \psi(\kappa) =\tilde W_{\rm out} \kappa + \tilde b_{\rm out},
\end{equation}
where $\tilde W_{\rm out} \in \mathbb R^{B\times K}$ is a set of decoder weights and $\tilde b_{\rm out}\in \mathbb R^B$ are biases. As a first step, by simply replacing $\kappa = Nr$, we prove the first statement of the theorem:
\begin{equation}
    \psi(r) = \tilde{\psi}(\phi(r)) = \underbrace{\tilde W_{\rm out} N r}_{W_{\rm out}} + \underbrace{\tilde b_{\rm out}}_{b_{\rm out}}.
\end{equation}
Thus, for any latent decoder, there is an equivalent decoder from neural activities such that $W_{\rm out} = \tilde W_{\rm out} N$ and $b_{\rm out}=\tilde b_{\rm out}$. To prove the converse direction, we start by writing:
\begin{equation}
\begin{split}
    \hat y = \psi(r) &= W_{\rm out} r + b_{\rm out} = W_{\rm out} [P_Nr + (1-P_N)r] + b_{\rm out} = W_{\rm out} P_Nr + \underbrace{W_{\rm out}(1-P_N)}_{=0}r + b_{\rm out}, \\
    &= W_{\rm out} P_Nr + b_{\rm out} =  W_{\rm out} N^T (N N^T)^{-1} \underbrace{N r}_{\kappa} + b_{\rm out}, \\
    &= \underbrace{W_{\rm out} N^T (N N^T)^{-1}}_{\tilde W_{\rm out}} \kappa + \underbrace{b_{\rm out}}_{\tilde b_{\rm out}}.
\end{split}
\end{equation}
Here, we used the fact that $P_N = N^T (N N^T)^{-1} N$. Thus, as long as $W_{\rm out}(1-P_N) = 0$, one can always construct a linear latent decoder, making this a sufficient condition, which ensures that the decoders defined on the neural activities lie on the subspace defined by the encoding weights. 

Finally, we now prove that this condition is also necessary. Suppose, by contradiction, that the condition does not hold. Then the linear readout can be written as:
\begin{equation}
\hat y = \underbrace{W_{\rm out} N^T (N N^T)^{-1} \kappa + b_{\rm out}}{\tilde \psi(\kappa)} + \bar W_{\rm out} r,
\end{equation}
where $\bar W_{\rm out} = W_{\rm out}(1-P_N)$.
If $\hat y$ were truly a function of $\kappa$ alone, then each value of $\kappa$ should map to a unique value of $\hat y$. We can disprove this by construction. Consider two vectors $r_{\pm} = r_\parallel \pm r_\perp$, where: $r_\parallel = P_N r$ and $r_\perp = (1-P_N)r$. Both vectors map to the same $\kappa$:
\begin{equation}
\kappa = N r_{\pm} = N r_\parallel
\end{equation}
However, they produce different outputs:
\begin{equation}
\hat y_{\pm} = \tilde\psi(\kappa) + \bar W_{\rm out} r_\pm = \tilde\psi(\kappa) + \bar W_{\rm out} (r_\parallel \pm r_\perp)
\end{equation}
Since $\bar W_{\rm out} r_\parallel = W_{\rm out} (1-P_N) P_N r = 0$, this simplifies to:
\begin{equation}
\hat y_{\pm} = \tilde\psi(\kappa) \pm \bar W_{\rm out} r_\perp = \tilde\psi(\kappa) \pm \bar W_{\rm out} r_+,
\end{equation}
where $r_+ \in \mathbb R^{N_{\rm rec}}$ by definition above. Therefore, unless $\bar W_{\rm out} = 0$, we have $\hat y_{+} \neq \hat y_{-}$ despite both corresponding to the same $\kappa$. This contradicts our assumption that $\hat y$ is a function of $\kappa$ alone, proving the necessity of the condition and concluding the proof.
\end{proof}

One might wonder whether the linear bottleneck condition, \textit{i.e.}, $W_{\rm out}(1-P_N) = 0$, is too restricting. Fortunately, the readout too can be incorporated into the LPU dynamics. Specifically, if $W_{\rm out}(1-P_N)\neq 0$ and for simplicity assume $B=1$, we can define $n^{(K+1)} = W_{\rm out}(1-P_N)$ and subsequently $\kappa_{K+1}(t) = n^{(K+1)T} r(t)$. Then, by inspection, $\dot \kappa_{K+1}(t)$ depends only on $\{\kappa_1,\ldots,\kappa_K\}$ since $\dot r(t)$ followed a (nonlinear) embedding equation of the $K$-dimensional LPU by assumption. Thus, similar to our calculations in Lemma \ref{lm1}, it is possible to extend the LPU to a $K+1$-dimensional dynamical system that allows a linear readout.

\subsubsection*{Universal decoding theorem}

Now, we return to our discussion of identifiability of LPUs, this time focusing on the encoding weights. As discussed above, LPUs respect a $GL_K(\mathbb R)$ symmetry. Specifically, we can transform latent variables with an invertible matrix $S \in \mathbb R^{K \times K}$ ($\kappa \to S \kappa$) without changing the LPU operation. Under this transformation, the encoding weights would transform following $N \to S N$. Fortunately, this does not change the encoding subspace (defined by $P_N$), only the basis vectors we choose within it since:
\begin{equation}
    P_N = N^T (N N^T)^{-1} N \to  N^T S^{T} (S N N^T S^T)^{-1} S N = N^T S^T S^{-T} (NN^T)^{-1} S^{-1}S N = N^T (N N^T)^{-1} N = P_N,
\end{equation}
where we used the fact that both matrices $NN^T \in \mathbb R^{K\times K}$ and $S$ are full-rank. The latter is true by definition, whereas the former is true as long as $N$ is rank $K$ (as assumed throughout this work). Thus, the symmetry transformation does not change the encoding subspace. Yet, we defined the latent variables following $\kappa_p(t) = \frac{1}{N_{\rm rec}} n^{(p)T} r(t)$, where $\frac{1}{N_{\rm rec}} n^{(p)}$ is the $p$th column vector of $N$. Therefore, the symmetry transformation ($S$) does change the coordinate system that defines the variables $\kappa(t)$. Then, the question becomes: Should a network keep explicit track of a particular coordinate system for $\kappa(t)$, or just the encoding subspace itself? 

Lemma \ref{lm2} suggests that as long as the behavioral readouts are achieved with linear decoding from LPUs, there is no reason for the downstream units (\textit{e.g.}, a motor neuron) to explicitly keep track of latent variables in a particular coordinate system. On the other hand, this also suggests a trade-off. Specifically, the linear decoder, $\psi(r)$, is constrained to use only the neural activity dimensions living on the encoding subspace. However, an optimal decoder, \textit{i.e.}, one that most accurately predicts the desired outputs, could use redundant information in the full neural activity, potentially leading to higher accuracy. As we show in the next section, this bottleneck condition is desirable to achieve robustness to representational drift, \textit{i.e.}, changes in the tuning properties of individual neurons. For now, we conclude this section by bringing Lemmas \ref{lm1} and \ref{lm2} together to prove a theorem on the identifiability of LPUs and the universality of their outputs:
\begin{theorem}[Universal linear decoding without LPU identification]  \label{thm2}
     Let $r(t) \in \mathbb{R}^{N_{\rm rec}}$ represent neural activities that can linearly encode a $K + B$-dimensional LPU. In the limit $N_{\rm rec} \to \infty$, a particular set of differentiable $B$-dimensional outputs of a universally approximated $K$-dimensional dynamical systems can be linearly decoded from $r(t)$, without explicitly identifying a set of weights ($N$) that encode the latent variables.
\end{theorem}
\begin{proof}
    Theorem \ref{thm1} ensures that, in the limit $N_{\rm rec} \to \infty$, the first $K$ dimensions of the LPU can approximate the latent dynamical system of interest. Applying Lemma \ref{lm1} repetitively ($B$ many times), one can design an extended LPU in which a particular set of $B$ outputs from the target dynamical system can be readout linearly from the remaining $B$ latent variables. Hence, it is possible to design a $K+B$-dimensional LPU that models the $K$-dimensional target dynamical system and provides $B$ linear outputs from the latent variables. Finally, Lemma \ref{lm2} states that any linear readout from the LPU can be achieved directly from the neural activities, which does not require explicit identification of $N$. This concludes the proof.
\end{proof}

The true power of Theorem \ref{thm2} lies in its agnosticism to individual neural processes. As long as the overall nonlinearities are expressive enough, \textit{e.g.}, as in the case for leaky firing rate RNNs above, (extended) LPUs can store desired outputs of the computation in a linearly accessible manner. 

\subsubsection*{Connections to the existing universal approximation theorems}

The mathematical foundation of our work builds upon the seminal universal approximation theorem for feedforward neural networks \cite{hornik1989multilayer}. This theorem later inspired Schaefer and Zimmermann's pivotal proof \cite{schafer2006recurrent}, which demonstrated that discrete recurrent neural networks (RNNs) could universally approximate open dynamical systems. Their proof utilized two distinct neural populations, one implementing the latent dynamics and another handling the readout function. Subsequently, \cite{beiran2021shaping} advanced the field by extending these concepts to continuous leaky current RNN architectures, focusing primarily on modeling latent dynamics while introducing an alternative proof methodology. Although, as our analysis above demonstrates, the classical proof approach utilizing multilayer networks remains applicable to these modern RNN architectures.

Our Theorems \ref{thm1} and \ref{thm2} advance this framework in two significant ways. First, they establish the universal approximation properties for leaky firing rate RNNs, effectively extending the discrete-time results in \cite{schafer2006recurrent} to their continuous counterparts, with a careful consideration of increased latent dimensionality and identifiability issues. Second, most notably, by assuming that the (nonlinear) readout, $\bar \psi(\kappa)$, is differentiable, we have shown that an optimal linear readout can be designed using the same group of neurons as the ones supporting the latent dynamical system, extending the original proof in \cite{schafer2006recurrent}.

\subsection*{\color{subsectioncolor}Redundancy and neural activity manifolds subserved by latent processing units}
In this section, we study the embedding properties of the LPUs onto high-dimensional neural activities. We start by proving that nonlinear embedding function leads to curved neural manifolds, in which  $K$-dimensional LPUs can lead to unbounded/unsaturated scaling in the linear dimensionality of the neural activities with increased number of $N_{\rm rec}$ neurons. Then, we discuss how LPUs naturally give rise to redundant representations, and why this redundancy is necessary for universal computation. Finally, we prove that any noise added to the neural manifold in the direction that is orthogonal to the encoding subspace (defined via $P_N = N^T (NN^T)^{-1} N$) decays exponentially.

\subsubsection*{Curved neural activity manifolds with unbounded linear dimensionality}

We begin by defining linear dimensionality:
\begin{definition}[Linear dimensionality]
    The linear dimensionality of a set of neural activities ($\{r_i[s]\}$ for $s = 1,\ldots, T$ and $i = 1,\ldots, N_{\rm rec}$) is defined as smallest dimension $d$ such that there exists a subspace of dimension $d$ that contains all the points. 
\end{definition}
Since $r[s] \in \mathbb R^{N_{\rm rec}}$ and $\kappa[s] \in \mathbb R^K$, the maximum linear dimensionality is $N_{\rm rec}$. Throughout this section, we assume that all latent subspace is occupied during circuit operation, hence linear dimensionality of the LPU (also referred to as ``latent dimensionality'') is assumed to be $K$. Since noise and transient responses can artificially inflate the linear dimensionality, in this section, we restrict our analysis to latent variables that are noiseless,  $\epsilon$-invariant for time scales much larger than $\tau$, and have negligibly fast transient responses such that the embedding relationship becomes effectively static, as discussed above. Consequently, the results in this section allow us to understand the traditional framework of autoencoders.

Next, we define flat and curved manifolds. While more quantitative definitions are possible (see, \textit{e.g.}, \cite{acosta2023quantifying}), our focus here is on a binary distinction between the two classes. Therefore, we adopt a simplified definition:
\begin{definition}[Flat and Curved Manifolds] 
A neural activity manifold (defined on $r(t)$) is considered flat if the linear dimensionality and the latent dimensionality of the network are equal. If the linear dimensionality exceeds the latent dimensionality, the neural activities are said to lie on a curved manifold. 
\end{definition}

The curvature of neural manifolds plays a central role in our ability to estimate the dimensionality of neural activities. As we have seen above (see Eq. (\ref{eq:lcrnn})), the neural activities in regularly studied low-rank RNN structures reside in flat manifolds \cite{valente2022extracting,dubreuil2022role,beiran2021shaping}. In contrast, as we illustrate in Fig. \ref{fig2}, the networks we defined in Eq. (\ref{def3}) lead to curved manifolds and therefore possess distinct geometrical properties. 

Now, inspired by previous work on counting the piecewise-linear regions of deep neural networks \cite{montufar2014number,pascanu2013number}, we show that LPUs even in simple RNNs can lead to unbounded linear dimensionality:
\begin{proposition}[Unbounded linear dimensionality] \label{prop2}
   Let $r(t) \in \mathbb R^{N_{\rm rec}}$ be the neural activities that embed a low-dimensional LPU $\kappa(t) \in \mathbb R^K$. Let $\kappa$ be $\epsilon$-stable for the range of values $\kappa \in [0,1]$ for some time $T \gg \tau$, leading to approximately static embedding. For a nonlinear embedding function $f_i(.)$ satisfying the condition $f_i(x) = 0$ for some $x \leq 0$ and $f_i(x) > 0$ for $x > 0$, it is possible to design an embedding such that the linear dimensionality in $r(t)$ is equal to the maximum, \textit{i.e.}, $N_{\rm rec}$, for a fixed latent dimensionality $K$. 
\end{proposition}
\begin{proof}
    For simplicity, we focus on the case with $K=1$. However, the proof generalizes to higher dimensional LPUs, \textit{e.g.}, by separately considering sets of data points sampled alongside a coordinate $\kappa_d$ such that $\kappa_i = 0$ for all $i\neq d$. From now on, we refer to a one-dimensional latent variable as $\kappa \in \mathbb R$. We prove this proposition by \emph{constructing} a simple positive example. Consider the tuning curve for a single latent variable (transforming Eq. (\ref{eq:tuning})):
    \begin{equation}
        r_i(\kappa) = f_i(m_i \kappa + b_i)
    \end{equation}
    where the $\epsilon$-stability of $\kappa$ allowed us to define the static embedding relationship. Then, by setting all embedding weights to one, $m_i = 1$, and defining a neuron-specific bias $b_i = \frac{i-1}{N_{\rm rec}}$ for $i = 1,\ldots,N_{\rm rec}$, we obtain:
    \begin{equation}
        r_i(\kappa) = f_i\left(\kappa-\frac{i-1}{N_{\rm rec}}\right), \quad \text{for} \quad \kappa \in [0,1].
    \end{equation}
     Now, consider the collection of points $\{r^{(j)}\}$ that result from the embedding $\kappa \in \{0,\frac{1}{N_{\rm rec}},\frac{2}{N_{\rm rec}}, \ldots, \frac{j-1}{N_{\rm rec}}, \ldots  1 \}$ for $j = 1,\ldots, N_{\rm rec} + 1$. Then, we can show that the subspace defined by the points $\{r^{(j)}\}$ is $N_{\rm rec}$ dimensional. For $j=1, ..., N_{\rm rec} +1$ and any given $i=1, ..., N_{\rm rec}$, the following is true:
    \begin{equation}
        r_i^{(j)} = 
        \begin{cases}
            0 \quad \text{if}\quad i \geq j, \\
            f_i\left( \frac{j-i}{N_{\rm rec}}\right) \quad \text{o.w.}
        \end{cases}
    \end{equation}
    Then, noting that $r^{(1)}=0$, we can define the vectors spanning the subspace using $\{r^{(j)}\}$ for $j\geq 2$. Arranging these vectors as rows yields a $N_{\rm rec} \times N_{\rm rec}$ matrix where each row $r^{(j)}$ contains non-zero elements only in its first $j-1$ positions, corresponding to the neuron indices for which the embedding function is non-zero. This structure guarantees that the matrix is lower triangular, with non-zero diagonal elements since $f(x) >0$ for $x>0$. Since a lower triangular matrix with non-zero diagonal entries has full rank, these $N_{\rm rec}$ vectors are linearly independent. Hence, they span an $N_{\rm rec}$ dimensional subspace. As there are $N_{\rm rec}$ neurons, $N_{\rm rec}$ is the maximum linear dimensionality that can be achieved, concluding the proof.
\end{proof}

Proposition \ref{prop2} demonstrates that even simple networks can exhibit unbounded linear dimensionality, despite having only a single latent variable:

\begin{corollary}[Unbounded Linear Dimensionality with a Single Latent Variable] \label{cor2}
Let $\kappa \in \mathbb R$ be statically embedded into neural activities $r \in \mathbb R^{N_{\rm rec}}$ through a combination of linear and nonlinear operations, using activation functions such as the rectified linear unit (ReLU) or step nonlinearity. With an appropriate choice of bias parameters, embedding $\kappa$ values confined to the interval $[0,1]$ can result in neural activities that achieve the full linear dimensionality of $N_{\rm rec}$. 
\end{corollary}
In this section, we have provided all our proofs in static embedding limit. However, this is only a limit under restrictive assumptions, in which $\kappa$ are assumed to evolve slowly. Since the dynamical embedding framework also captures this scenario, one can consider these results as possible lower bounds for what can be achieved in dynamical networks. Therefore, an unbounded linear dimensionality observed in real neural recordings (see \cite{stringer2019high,manley2024simultaneous}) does not necessarily provide information regarding the latent dimensionality, \textit{i.e.}, the true dimensionality of the neural code.

\subsubsection*{Theoretical analysis of representational redundancy and ineffective perturbation dimensions}

For this section, we start from the time-evolution of neural activities, described in Eqs. (\ref{eq:neural_general}) and (\ref{eq:latent_general}), by dividing the neural activities into two:
\begin{equation}
    r(t) = P_Nr + (1-P_N)r(t) = r_\parallel (t) + r_\perp(t),
\end{equation}
where we define $P_N$ as the projection to the encoding subspace spanned by $N \in \mathbb R^{K \times N_{\rm rec}}$ (assumed to be full rank, \textit{i.e.}, rank $K$), whereas $r_\parallel(t)$ and $r_\perp (t)$ stand for parallel and orthogonal neural activity components, respectively. We first note that the latent activities depend only on $r_\parallel(t)$ since
\begin{equation}
    \kappa(t) = N r(t) = (NN^T) (N N^T)^{-1} Nr(t)  = N \underbrace{[N^T (N N^T)^{-1} N]}_{P_N} r(t) = N P_N r(t) = N r_\parallel(t).
\end{equation}
Thus, for the purpose of this section, we refer to the embedding map as $\varphi(r_\parallel(t),u(t))$, where $r_\parallel(t)$ contains the dependence on $\kappa(t)$. Then, without loss of generality, the RNN equations in Eq. (\ref{eq:neural_general}) can be written as:
\begin{equation} \label{eq:neural_simplified}
    \tau \dot r(t) = -r(t) + \varphi(r_\parallel(t),u(t)).
\end{equation}
We use this form in our proofs for the rest of this section. We start with a lemma:
\begin{lemma} \label{lm3}
     Let $r^{(1/2)}(t) \in \mathbb{R}^{N_{\rm rec}}$ represent the neural activities that support a LPU $\kappa \in \mathbb{R}^K$, linearly encoded by $N \in \mathbb{R}^{K \times N_{\rm rec}}$. Let $P_N(r^{(1)}-r^{(2)}) = 0$, where $P_N$ is the projection operator to the row subspace of $N$. Then, the difference, $\Delta r(t) = r^{(1)}(t)-r^{(2)}(t)$, decays exponentially over time and does not affect the LPU operation.
\end{lemma}
\begin{proof}
     By definition, the two trajectories only differ in their components orthogonal to the encoding subspace, \textit{i.e.}, we can define $r_\parallel = P_N r^{(1)} = P_N r^{(2)}$. Then, following Eq. (\ref{eq:neural_simplified}), we can write the time evolution of $r_\perp^{(1/2)} = (1-P_N)r^{(1/2)}$ as:
    \begin{align}
        \tau \dot r_\perp^{(1/2)} = -r_\perp^{(1/2)} + [I-P_N] \varphi(r_\parallel(t),u(t)),
    \end{align}
    where $r_\parallel(t)$ follows a set of equations ($\tau \dot r_\parallel = -r_\parallel + P_N \varphi(r_\parallel(t),u(t))$) independent of $r_\perp$. Then, for simplicity and since it is $r_\perp$-independent, we refer to $[I-P_N] \varphi(r_\parallel(t),u(t))$ as $h(t)$. Let us now study the deviations ($\Delta r_\perp$) between two trajectories $r_\perp^{(1)}$ and $r_\perp^{(2)}$:
    \begin{subequations}
        \begin{align}
            \tau \dot r_\perp^{(1)} &= -r_\perp^{(1)}  + h(t), \\
            \tau \dot r_\perp^{(2)} &= -r_\perp^{(2)}  + h(t), 
        \end{align}
    \end{subequations}
    where we note that $\Delta r(t) = r^{(1)}(t)-r^{(2)}(t) = r_\perp^{(1)} - r_\perp^{(2)}$. Thus, subtracting both equations leads to:
    \begin{equation}
        \tau \dot \Delta r (t) = - \Delta r(t) \implies \Delta r (t) = \Delta r (0) \exp(-t/\tau).
    \end{equation}
    Since the LPU equations above depend on $r_\parallel(t)$, which is the same for both $r^{(1/2)}$, $\Delta r(t)$ does not affect the LPU operation, concluding the proof. 
\end{proof}
Lemma \ref{lm3} suggests an interesting observation. Although neural activity manifolds are curved, deviations tangential to the $N_{\rm rec}-K$ dimensional subspace defined by $1-P_N$ decay exponentially back to the curved neural activity manifold. A similar observation has been made for the existing low-rank RNN architectures, whose neural activities are constrained to $K$-dimensional flat subspaces \cite{valente2022extracting}. However, unlike the prior architectures, the curvature of the neural manifolds endows the general class of linearly-encoded latent networks from definition \ref{def3} with the ability to also support redundant information:
\begin{lemma}[Neural redundancy]\label{lm4}
    Let $r(t) \in \mathbb R^{N_{\rm rec}}$ be neural activities supporting a LPU $\kappa \in \mathbb R^K$ linearly encoded with a rank-K matrix $N \in \mathbb R^{K \times N_{\rm rec}}$. Let $\kappa$ be $\epsilon$-stable for some time $T \gg \tau$, leading to approximately static embedding. Then, neural activities projected onto a subspace orthogonal to the encoding subspace are still tuned to the LPU, \textit{i.e.}, internal coding variables.
\end{lemma}
\begin{proof}
    We can, once again, rewrite Eq. (\ref{eq:neural_simplified}) by dividing $r(t)$ into two components:
    \begin{subequations}
    \begin{align}
        \tau \dot r_\parallel = -r_\parallel + P_N \varphi(r_\parallel(t),u(t)),\\
        \tau \dot r_\perp = -r_\perp + [I-P_N] \varphi(r_\parallel(t),u(t)).
    \end{align}
    \end{subequations}
    The parallel component follows a self-sufficient dynamical system, whereas the orthogonal component is coupled. Notably, the LPU is fully accessible via the parallel component. Specifically, by multiplying both sides of the first equation with $N$ from the left, we end up with Eq. (\ref{eq:latent_general}). Thus, the second equation is redundant as far as the LPU is concerned. On the other hand, writing the latent dependence explicitly, we arrive at the following equation:
    \begin{equation}
        \tau \dot r_\perp = -r_\perp + [I-P_N] \varphi(\kappa(t),u(t)).
    \end{equation}
   Assuming $\epsilon$-stability of $\kappa(t)$ and no input similar to above, this equation leads to tuning curves with respect to the latent variables:
    \begin{equation}
        r_{\perp}(\kappa) = [I-P_N] \varphi(\kappa(t)).
    \end{equation}
    In general, this is a non-zero function of $\kappa$, concluding the proof.
\end{proof}
Lemma \ref{lm4} suggests that even though the neural activity components orthogonal to the encoding subspace do not effect the LPU operation, they contain redundant information on the latent variables due to nonlinear embedding maps. Therefore, if we enforce linear encoding (due to several interpretability, accessibility, and bio-plausibility reasons as described above), the embedding has to be nonlinear to achieve universal approximation (Table \ref{tabs2}) and thereby lead to unavoidable redundancy.

We can illustrate the concept of unavoidable redundancy with a simple example of a single latent variable with a (nonlinear) tuning curve $r_i(\kappa) = f(m_i\kappa)$, with the encoding and embedding weights $n\in\mathbb R^{N_{\rm rec}}$ and $m \in \mathbb R^{N_{\rm rec}}$, respectively. Without loss of generality (\textit{i.e.}, by rescaling $m$ appropriately such that $W^{\rm rec} = mn^T$ remains invariant), assume that $||n||_2^2 = 1$ such that the one-dimensional encoding subspace is defined by $P_N = nn^T$. Now, assume that for a particular neuron, $n_i = 0$ but $m_i \neq 0$. In this case, the neuron does not contribute to the LPU, yet the LPU inform the activity of that neuron. Then, we arrive at 
\begin{equation}
    (r_{\perp})_i(\kappa) = f(m_i\kappa).
\end{equation}
This particular neuron, despite not contributing to the LPU, has a tuning curve, similar to a neuron that does. Such a case could be particularly useful, for example, when different brain regions communicate their computational results with each other. In such cases, the neural activities in a downstream brain region can represent computation performed by a latent variable native to a different region. However, as we formalize in the next theorem, perturbation in these dimensions do not affect the LPU:
\begin{theorem}[Representational redundancy] \label{thm4}
    Let $r(t) \in \mathbb R^{N_{\rm rec}}$ represent neural activities supporting a LPU $\kappa \in \mathbb R^K$, linearly encoded by a rank-$K$ matrix $N \in \mathbb R^{K \times N_{\rm rec}}$. Perturbations leading to changes in the neural activities orthogonal to the encoding subspace, even through neurons that are tuned to the latent variables, do not affect the operation of the LPU and vanish exponentially after the perturbation offset.
\end{theorem}
\begin{proof}
    This theorem  combines Lemmas \ref{lm3} and \ref{lm4} into a unified statement. Specifically, as shown in Lemma \ref{lm4}, neural activities in directions perpendicular to the encoding subspace can be tuned to latent variables. However, by Lemma \ref{lm3}, the LPU $\kappa(t)$ depends only on the parallel component $r_\parallel(t)$, \textit{i.e.}, the projections to the encoding subspace defined by $P_N$. Since deviations orthogonal to the encoding subspace, $r_\perp(t)$, do not affect $r_\parallel(t)$ by Lemma \ref{lm3}, these perturbations have no impact on the LPU. Finally, Lemma \ref{lm3} proves that these deviations decay exponentially, concluding the proof.
\end{proof}
Essentially, Theorem \ref{thm4} formalizes the notion that, for the class of multi-synapse latent networks, neurons contribute to the latent code through their encoding weights, $N$, whereas their tuning properties are fundamentally determined by both the embedding weights, $M$, and the nonlinearity. In any case, the nonlinear embedding, which is essential for universal computation under linear encoding maps, inevitably gives rise to redundant tuning curves. Therefore, designing perturbation experiments (as illustrated in Fig. \ref{fig1}\textbf{A}) based on neural tuning properties would be a suboptimal experimental strategy, as these properties may not reliably indicate the neuron's true computational role in information encoding.

It is worth noting, as a final caveat, that interventions should likely be considered as being performed on the full LPU (the subspace spanned by $P_N$), rather than on individual latent variables (a specific $n^{(r)}$). Specifically, due to the network effects, it is often plausible that changing the activity of a particular latent variable $\kappa_i(t)$ may be best achieved by interfering with the encoding direction of another set of (or combination thereof) latent variables $\{\kappa_j(t)\}$ with $i \not \in j$. 

\subsection*{\color{subsectioncolor}Theoretical analysis of representational drift in RNNs} \label{app:representational_drift}

In this section, we study the robustness of LPUs to representational drift. We start by defining tuning curves with respect to the LPUs, which may appear concrete and interpretable, but suffer from a non-identifiability issue that is well known in other contexts \cite{hyvarinen2024identifiability}. Thus, to define the tuning curves, we first align the latent variables with the network outputs, breaking the internal symmetry of LPUs in a behaviorally informed manner. Then, after introducing few mild assumptions, we state our main theoretical results on representational drift and provide their proofs. Finally, we provide the experimental details behind our tests with simulated RNNs.

Representational drift refers to constant changes in the tuning curves of neural responses to specific stimuli and/or behaviorally relevant outside variables \cite{driscoll2022representational}. The tuning curves defined in Eq. (\ref{eq:tuning}) suggest a method for studying these changes in RNNs. Once the $GL_K(\mathbb R)$ symmetry of the latent variables is used to align the distinct latent variables to the behaviorally relevant quantities (see below for how we achieve this for K-bit flip flop tasks), we can vary the tuning curves (simulating representational drift) by changing in an embedding weight, $m^{(d)}$:
\begin{equation}
    m^{(d)} \to m^{(d)} + \Delta m,
\end{equation}
where $d$ refers to the specific embedding weight receiving the update. For a network that is robust to the drift, the following is expected: If the latent variables, $\kappa(t)$, follow a dynamical system that is invariant to these updates, $\Delta m$, then even though the tuning curves of individual neurons (as in Eq. (\ref{eq:latent_general})) undergo changes, the LPUs and the behavioral readouts from these LPUs remain robust to these changes. 

We start our analysis by defining the time-dependent joint Jacobian in Eq. (\ref{eq:tuning}) for the set of nonlinear transfer functions, $f_i(\cdot):\mathbb R^{K} \to \mathbb R^{1} $ with $i = 1,\ldots, N_{\rm rec}$, with respect to their arguments ($m^{(p)}_i \kappa_{p}(t)$ for the $p$th argument with $p=1,\ldots, K$) as:
\begin{equation} \label{eq:jacobian}
    J_{is}(t) = \frac{\partial}{\partial [m_i^{(s)} \kappa_s(t)]} \tilde f_i\left(\{ m^{(p)}_i \kappa_{p}(t)\}\right) = \nabla \tilde f_{is}\left(\{ m^{(p)}_i \kappa_{p}(t)\}\right),
\end{equation}
where we recall $[m_i^{(s)} \kappa_s] \in \mathbb R$ as the current contributed by the latent variable $\kappa_s(t)$ to the $i$th neuron. This Jacobian term can intuitively thought as quantifying the saturation level of individual neurons (when $\kappa_s(t)$ changes) and the local sensitivity of the tuning curves to changes in latent variables (when $m_i^{(s)}$ changes). For an example of the former, consider a simple neuron with the relationship $r(\kappa) = \tanh(\kappa)$. Here, the latent variable can be considered as current into the neuron, which excites the neuron. For inactive neurons, we have $J(\kappa) = 1-\tanh^2(\kappa)\xrightarrow{\kappa\to0} 1$, suggesting that small changes in $\kappa$ linearly effect neural activities. On the other hand, active neurons are almost insensitive to changes in the latent variable since $J(\kappa \to \pm \infty) \to 0$. With this definition, we now make two key assumptions:
\begin{enumerate}
    \item Entries of encoding and embedding weights are sampled following some unknown probability distribution $P$, such that there are finite and well-defined expectations for $E_P[n_i^{(p)} m_i^{(p')}]$ for any $p,p'$ combination.
    \item We either assume that the drift in the embedding vector $m^{(d)}$ is purely random, or for non-random drift, we require a conditional orthogonality between drift and encoding dimensions, \textit{i.e.}, $E[n_i^{(p)} \Delta m_i | J_{id}] = 0$ for $p = 1,\ldots,K$. Here, $J_{id}$ refers to the Jacobian defined in Eq. \eqref{eq:jacobian}.
\end{enumerate}

The first assumption is common in the field \cite{beiran2021shaping}. It suggests that the distributional characteristics of synaptic connections, rather than the precise values of individual synapses, are the primary factors driving computation. This assumption can also be considered a rigorous formulation of the neural manifold hypothesis, in which neural manifolds are supported by statistical properties of synapses. The second assumption is more restrictive, suggesting that either the system is subject to random drift, or, for systems with structured drift, changes in the synaptic weights are assumed to be orthogonal to the $K$-dimensional encoding subspace and by factors independent of the specific local linearization of the network's dynamics, \textit{i.e.}, the conditional orthogonality with respect to $J_{id}$.  Finally, we consider the discretized neural activities:
\begin{equation} \label{eq:discretized_general}
    r[s+1] = (1-\alpha)r[s] + \alpha \tilde f \left(\{ m^{(p)} \kappa_{p}[s] \}_{p=1,..., K}\right), \quad \text{for} \; s=0,\ldots,T,
\end{equation}
where $\alpha = \Delta t/\tau$ is the discretization parameter. As before, we refer to variables of the discretized dynamical system as $\cdot[s]$, where $s$ denotes the discretized time. We are now ready to state our main result:
\begin{theorem}[Robustness of LPUs to representational drift] \label{thm5} 
    Let a recurrent neural network follow a trajectory $\mathcal O:= \{r[0], r[1], \ldots, r[T] \} $, dynamically embedding the LPU $\kappa[s] \in \mathbb R^K$ following the equation \eqref{eq:discretized_general} with a set of embedding weights $m^{(p)}$ for $p=1,\ldots,K$. Let the LPU be linearly encoded following $\kappa_p[s] = \frac{1}{N}n^{(p)T}r[s]$, where we refer to $\frac{1}{N}n^{(p)T}$ as the encoding vector for the $p$th latent variable. Let a particular embedding weight, $m^{(d)}$, undergo a representational drift, $\Delta m[s]$, with which the particular trajectory, $\mathcal{O} $, starting from the same neural activity $r[0]$, is now modified to to $\mathcal{O}_{\rm new} := \{r[0],r_{\rm new}[1], \ldots, r_{\rm new}[T] \}$. In the limit $N_{\rm rec}\to \infty$, the dynamics of the LPU remains invariant despite the drift, up to a correction in $O(\Delta m^2)$, as long as $E[n_i^{(p)} \Delta m_i | J_{id}] = 0$, where $J_{id}$ (see Eq. \eqref{eq:jacobian}) is the Jacobian for the set of nonlinear functions in Eq. \eqref{eq:discretized_general}. 
\end{theorem}
\begin{proof}
    We define the latent variables before and after drift as:
    \begin{equation}
        \kappa_p[s] = \frac{1}{N_{\rm rec}} n^{(p)T}r[s], \quad \tilde \kappa_p[s] = \frac{1}{N_{\rm rec}} n^{(p)T}r_{\rm new}[s],
    \end{equation}
    where we assume the trajectories start form the same initial condition $\kappa[0] = \tilde \kappa[0]$. For the latent dynamics to remain unperturbed, we require that $\kappa[s] := \tilde \kappa[s]$ for $s = 1,\ldots, T$. Defining $\Delta \kappa = \tilde \kappa - \kappa$, we get:
\begin{equation} \label{eq:deltakappa_origin}
\begin{split}
   \tilde \kappa_p[s]:=& \frac{1}{N_{\rm rec}} n^{(p)T} r_{\rm new}[s] =  (1-\alpha)\tilde \kappa_p[s-1] \\
   &+ \frac{\alpha }{N_{\rm rec}} \sum_{i=1}^{N_{\rm rec}} n^{(p)}_i\tilde f_i\left( m^{(1)}_i \tilde \kappa_{1}[s-1], \ldots,m^{(d)}_i \tilde  \kappa_{d}[s-1] + (\Delta m[s-1])_i \tilde  \kappa_d[s-1], \ldots \right), \\
     =& (1-\alpha) \kappa_p[s-1] + (1-\alpha) \Delta \kappa_p[s-1] + \frac{\alpha }{N_{\rm rec}} \sum_{i=1}^{N_{\rm rec}}  n^{(p)}_i f_i\left(\{ m^{(p')}_i \kappa_{p'}[s-1]\}\right) + O(\Delta m[s-1]^2) \\
     &+ \frac{\alpha }{N_{\rm rec}} \sum_{i=1}^{N_{\rm rec}} n^{(p)}_i \sum_{p'=1}^K  J_{ip'} \left[ \Delta \kappa_{p'}[s-1] m_i^{(p')} +  \kappa_d[s-1] (\Delta m[s-1])_i \delta_{p'd} +  \Delta \kappa_d[s-1] (\Delta m[s-1])_i \delta_{p'd}\right],
\end{split}
\end{equation}
where we have performed a Taylor approximation around the inputs to the nonlinear transfer function before the drift and kept the terms in first order of $\Delta m[s-1]$. Then, defining $\delta_{ij}$ as the Kronecker delta function that is one if $i=j$ and zero otherwise, we arrive at:
\begin{equation}
    \begin{split}
          \Delta \kappa_p [s] = (1-\alpha) \Delta \kappa_p[s-1] + \frac{\alpha }{N_{\rm rec}} \sum_{i=1}^{N_{\rm rec}} n^{(p)}_i \sum_{p'=1}^K  J_{ip'} \left[ \Delta \kappa_{p'}[s-1] [m_i^{(p')}  + \Delta m_i[s-1] \delta_{p'd}]+  \kappa_d[s-1] \Delta m_i[s-1] \delta_{p'd} \right].
    \end{split}
\end{equation}
This condition leads to a general perturbation equation, which can be expressed in continuous time (similar to the setup in Eq. (\ref{eq:discretized_general})) as:
\begin{equation}
    \tau \frac{\diff \Delta \kappa_p(t)}{\diff t} = - \Delta \kappa_p(t) + \frac{1}{N} \sum_{i=1}^N n^{(p)}_i \sum_{p'=1}^K  J_{ip'}(t) \left[ \Delta \kappa_{p'}(t) m_i^{(p')} +  \kappa_d(t) (\Delta m(t))_i \delta_{p'd} +  \Delta \kappa_d(t) (\Delta m(t))_i \delta_{p'd}\right].
\end{equation}
We can reorganize this equation in terms of decay and source terms such that:
\begin{equation}
\begin{split}
        \tau \frac{\diff \Delta \kappa_p(t)}{\diff t} =&  \sum_{p'=1}^K \Delta\kappa_{p'}(t) \underbrace{ \left[ -\delta_{pp'} + \frac{1}{N} \sum_{i=1}^N n^{(p)}_i  J_{ip'}(t)  m_i^{(p')}\right]}_{\Gamma_{pp'}(t)} + \left( \kappa_d(t) +  \Delta \kappa_d(t) \right) \underbrace{\frac{1}{N} \sum_{i=1}^N  n^{(p)}_i J_{id}(t)  \Delta m_i(t)}_{S_{dp}(t)}.
\end{split}
\end{equation}
Here, we recognize $\Gamma_{pp'}(t)$ as the time dependent decay matrix of the original LPU before the drift. Thus, this component explains the behavior of the LPU if the initial condition was perturbed in the absence of any drift, $\Delta m$. Therefore, since we assumed the two trajectories start from the same initial conditions, as long as the other terms are zero, the contribution of this term will be zero since $\Delta \kappa(0) = 0$ by assumption. In that regard, the second term constitutes an additional source for the deviations resulting from the drift. Bringing both observations together, we see that as long as $S_{dp}(t) \approx 0$, the LPU dynamics will remain invariable to the drift, \emph{even though the neural trajectories have changed from $\mathcal O \to \mathcal O_{\rm new}$.} For a representational drift satisfying the condition $E[n_i^{(p)} \Delta m_i | J_{id}] = 0$ and in the limit of $N_{\rm rec} \to \infty$, we have:
\begin{equation}
    S_{dp}(t) = \frac{1}{N} \sum_{i=1}^N  n^{(p)}_i J_{id}(t)  \Delta m_i(t) \approx E\left[ n^{(p)}_i J_{id}(t)  \Delta m_i(t) \right] = E\left[ J_{id}(t) \underbrace{E\left[ n^{(p)}_i  \Delta m_i(t) |J_{id} \right]}_{=0} \right]=0,
\end{equation}
concluding the proof.
\end{proof}

Theorem \ref{thm5} suggests that as long as the drift is in an orthogonal direction to the encoding weights, conditioned on $J_{id}(t)$, the latent dynamics will not be affected under small changes in $\Delta m$. Yet, how do we ensure this conditional orthogonality? For a fully random drift, the condition is trivially satisfied if the mean of the drift is zero. For a structured drift scenario, however, this implies that changes in the synaptic weights are driven by factors independent of the specific local linearization of the network's dynamics. Specifically, as illustrated above, the Jacobian $J_{id}(t)$ is related to the operation regime of the neuron dynamics around the nonlinearity (\textit{e.g.}, saturation or linear for $\tanh(\cdot)$ nonlinearity). Then, assuming that the LPU remains unperturbed, Eq. (\ref{eq:jacobian}) suggests that $J_{id}(t)$ are mainly functions of $m^{(p)}$. Thus, as long as $\Delta m$ is small enough, $J_{id}(t)$ following Eq. \eqref{eq:jacobian} is mainly determined by the distribution of $m^{(p)}$ across neurons.  In this case, the condition would be satisfied as long as $E_{i \in \mathcal I}[ n^{(p)}_i  \Delta m_i(t)]= 0$ for the group of neurons, $\mathcal I$, for which $J_{id}(t)$ is approximately equal. Finally, in the event of substantial changes in $\Delta m$, the nonlinearity $f(\cdot)$ often mixes and matches the linear subspaces (mathematically arising as increasing numbers of $n$th order conditions for the higher order terms in the Taylor approximation above), preventing the existence of a perfectly drift-invariant LPU in general, as demonstrated in Fig. \ref{fig4}.

Combining Theorem \ref{thm5} with Lemma \ref{lm3} allows defining a class of linear decoders that can remain robust to structured representational drift:
\begin{corollary}[Class of linear decoders robust to representational drift]
    Let $r(t) \in \mathbb{R}^{N_{\rm rec}}$ represent neural activities and $\kappa(t) \in \mathbb{R}^K$ be the corresponding latent variables obtained via a linear encoding, where $\kappa(t) = N r(t)$ for some rank-K matrix $N \in \mathbb{R}^{K \times N_{\rm rec}}$. Let $P_N$ be the projection operator onto the row space of $N$, and $y \in \mathbb{R}^B$ denote the behavior variables, where $N_{\rm rec} \gg K \geq B$. We define a linear decoder on neural activities as any function $\psi$ such that $\hat{y} := \psi(r) = W_{\rm out} r + b_{\rm out}$, where $W_{\rm out} \in \mathbb{R}^{B \times N_{\rm rec}}$ and $b_{\rm out} \in \mathbb{R}^B$ are trainable parameters. Then, a linear decoder, $\psi(r)$, satisfying the relationship $W_{\rm out}(1-P_N) = 0$ remains robust to the structured drift orthogonal to the encoding subspace in the sense that its output remains unaffected to the first order changes $\Delta m$. 
\end{corollary}
\begin{proof}
    Lemma \ref{lm3} allows defining an equivalent latent decoder $\tilde \psi (\kappa)$ and Theorem \ref{thm5} ensures that $\kappa$ follows an invariant set of dynamical systems. Thus, the output of the linear readout is also invariant. 
\end{proof}
It is worth recalling that the dimensional bottleneck (\textit{i.e.}, $W_{\rm out}(1-P_N)=0$) enforced on this class of linear decoders may prevent them from using redundant information in the full neural activity, leading to suboptimal decoding performances \emph{for a given time point.} Similar trade-offs between robustness and absolute optimality are common and well known in the statistics literature \cite{huber1992robust}.

\section*{Details of the simulation experiments}

\subsection*{\color{subsectioncolor} Latent processing units in RNNs performing K-bit flip flop tasks}

In the main text, we analyzed the LPUs of low-rank leaky firing rate RNNs, which were trained to solve K-bit flip flop tasks. 

\subsubsection*{K-bit flip flop task} 
In these tasks (Fig. \ref{fig1}\textbf{C}), the RNNs receive occasional pulses from K independent input channels and must produce corresponding outputs. Each input-output channel operates independently. In any given time bin, each input channel has a 0.05 probability of emitting a pulse with a value of either +1 or -1; otherwise, the input is 0. The network’s task is to remember the sign of the most recent pulse for each channel and output this sign accordingly. With K channels, there are $2^K$ possible output states. Each trial starts with a random initialization of the neural activities following a normal distribution with mean $0$ and s.d. $1$, and continues for $100$ discrete time points as input pulses are sampled randomly and independently for each of the $K$ channels.

\subsubsection*{Simulating RNNs} 

To simulate the leaky firing rate RNNs, we discretized the time evolution equations as:
\begin{equation}
    r[s+1] = (1-\alpha)r[s] + \alpha \tanh\left( W^{\rm rec} r[s] + W^{\rm in}u[s]\right),
\end{equation}
where $\alpha =0.5$ is the discretization ratio ($\alpha = \Delta t/\tau$ for a discretization time $\Delta t$) and we assumed no biases for simplicity. 

\subsubsection*{Two-step training paradigm} 

To train K-rank RNNs performing the K-bit flip flop tasks, we utilized the two-step a paradigm \cite{valente2022extracting}. First, we trained full-rank RNNs to perform the K-bit flop flop tasks. Second, we constrained low-rank RNNs to reproduce the neural activities of these full-rank RNNs. To train the full-rank RNNs on the K-bit Flip Flop Task, we re-used the publicly available code from \cite{dinc2023cornn}. We employed Adam optimization with a learning rate of $1/N_{hidden}$, and trained for 5000 epochs with a batch size of 50 trials. We used the mean squared error of the output as the loss function. We enforced $W^{\rm rec}_{ii}=0$ for the training of full-rank RNNs by projecting the gradients at each epoch. We also used a learning rate scheduler, which reduced the learning rate if it plateaued for 50 epochs. All weight matrices were initialized using Xavier initialization. To train the low-rank RNNs on the activities of the full-rank RNNs, we used a Pytorch model with a low-rank reconstruction for $W^{\rm rec}=MN$ with $M \in \mathbb R^{N\times K}$ and $N \in \mathbb R^{K\times M}$ (as in \cite{valente2022extracting}), and used the logistic loss as in \cite{dinc2023cornn}. We used Adam optimization with a learning rate of $10^{-3}$ and weight decay of $10^{-7}$, and trained for 10000 epochs with a batch of 5000 trials per network. The resulting networks had weight matrices which were by design rank $K$, the same as the number of inputs. 

\subsubsection*{Aligning latent variables with the output} 

As we discussed above, even when $W^{\rm rec}$ is low-rank, the reconstruction $W^{\rm rec} = MN$ is not unique, leading to equally valid, infinitely many, definitions of latent variables. Fortunately, in this specific experiment, a simple and useful way to break this symmetry is to design the latent variables to align with the outputs of the K-bit flip flop task such that $o_i(t) = \kappa_i(t)$ for $i=1,\ldots, K$.  To achieve this, once the RNNs were fully trained, we first computed the singular value decomposition of the weight matrices and initialized $M$ with the left eigenvectors and $N$ with the singular values times the right eigenvectors. Next, we sampled $4000$ data points as the RNNs were performing the K-bit flip flop task and performed a regression between the network outputs and the linearly encoded $\tilde \kappa(t) = N r(t)$ values. Using the learned regression matrix $S$, we transformed $N \to SN$ and $M \to MS^{-1}$ such that $\kappa(t) = S\tilde \kappa(t) \approx o(t)$. The resulting encoding and embedding weights, which satisfy $W^{\rm rec} = MN$, were used in conjunction with Eq. (\ref{eq:lfrnn_latent_circuit}) to draw and analyze the LPUs.

\subsubsection*{Designing latent processing units} 

We designed one-dimensional LPUs by sampling the encoding and embedding weights for each neuron, $\{n,m\}$, from a zero-mean two-dimensional Gaussian distribution with a positive correlation.  Using a mean-field approach, in the limit $N_{\rm rec} \to \infty$, the LPUs can be reshaped as:
\begin{equation}
    \tau \dot \kappa(t) = -\kappa(t) + \frac{1}{N_{\rm rec}} \sum_{i=1}^{N_{\rm rec}} n_i \tanh(m_i \kappa(t)) \approx - \kappa(t) + \int  n \tanh(m\kappa(t)) \diff p(n,m),
\end{equation}
where $p(n,m)$ is the probability density function for the pair $\{n,m\}$, which follow the probability distribution $P$. 

\subsubsection*{Tuning curve analysis} 

To compute the tuning curves, $r(\kappa)$, with respect to fixed $\kappa$, we consider the steady-state response ($\dot r(t) = 0$) after fixing $\kappa(t) := \kappa$ (\textit{i.e.}, $\epsilon$-stable) such that:
\begin{equation} \label{eq:tuning_curves}
    r(\kappa) = \tanh\left(\sum_{i=1}^K m^{(i)} \kappa_i\right).
\end{equation}
As long as the behavioral timescales that $\kappa(t)$ is responsible for are significantly larger than $\tau$ (\textit{e.g.}, fixed point attractors lead to fixed $\kappa$ values in the K-bit flip flop task), this function defines a tuning curve for the neural activities, $r$, with respect to the latent variables, $\kappa$. For Fig. \ref{fig4}\textbf{C}, each neuron was assigned to one of the $27$ coding groups. We defined coding for a particular flip flop bit by thresholding the corresponding embedding weights, $m^{(i)}$, by $\pm 0.5$. For instance, if $\{m^{(1)},m^{(2)},m^{(3)}\} = \{0.55, 0.1, -0.7\}$, this neuron was assigned to the group $[1,0,-1]$.

\subsubsection*{Additional details on Figs. \ref{fig1} and \ref{figs1}} 

For illustrations in Fig. \ref{fig1}\textbf{C-F}, we trained a representative rank-2 RNN solving the 2-bit flip flop task using the two-step training paradigm described above. The RNN had $100$ neurons, a time scale $\tau = 10ms$, and was discretized with $\Delta t = 5ms$. The neural activities shown in \textbf{D} are collected for $1000$ time points, whereas for \textbf{E-F}, they were collected for $10000$ time points. For Figs. \ref{fig1}\textbf{H-I} and \ref{figa1}, we sampled the embedding and encoding weights from a zero-mean Gaussian distribution that had the covariances $10^{-2}$ and $\approx 209$, respectively. For Fig. \ref{figs1}\textbf{D}, we designed the LPU with a zero-mean Gaussian distribution with the covariances $10^{-1}$ and $\approx 25$, respectively. For Fig. \ref{figs1}\textbf{E}, we designed the LPU with a zero-mean Gaussian distribution with the covariances $1.2$ and $\approx 6.24$, respectively. For all cases, the correlation of $0.7$ was chosen between two variables.

\subsection*{\color{subsectioncolor} Analysis of the neural recordings from the neocortex}
In this work, we re-used our cortex-wide neural imaging datasets from \cite{ebrahimi2022emergent}. Thus, the procedures for obtaining the Ca$^{2+}$ traces, including the subsections below from mouse preparation to cell extractions, are the same. For completeness, we summarize the mouse preparation below.

\subsubsection*{Mouse preparation} 

All animal procedures were approved by the Stanford University Administrative Panel on Laboratory Animal Care. For the imaging studies of layer 2/3 neocortical pyramidal neurons, 4 male and 2 female triple transgenic GCaMP6f-tTA-dCre mice were used. Surgeries were performed on mice aged 10–16 weeks under isoflurane anesthesia in a stereotaxic frame. To minimize inflammation and pain, preoperative carprofen (5 mg/kg) was administered subcutaneously and continued for 3 days post-surgery. We created a cranial window (5-mm diameter) over the right cortical area V1 and the surrounding cortical tissue. We covered the exposed cortical surface with a glass coverslip (\#1 thickness, 64-0700, CS-5R, Warner Instruments) attached to a steel annulus (1 mm thick, 5 mm outer diameter, 4.5 mm inner diameter, 50415K22, McMaster) and secured it with ultraviolet-light curable cyanoacrylate glue (Loctite 4305). For head-fixation during imaging, we cemented a metal head plate to the skull using dental acrylic. In vivo imaging sessions commenced at least 7 days post-surgery.

\subsubsection*{Image acquisition} 

We recorded Ca$^{2+}$ neural activity videos (20 fps; 2048 × 2048 pixels) using a fluorescence microscope with 40–160 $\mu$W/mm–2 illumination. Custom Matlab software controlled visual stimulus presentation, operated the behavioral apparatus via a NI-USB 6008 card, and initiated video capture on the microscope. Post-acquisition, we downsampled the videos to $1024 \times 1024$ pixels and 10 fps. Lateral brain movements were corrected using Turboreg software for image alignment. Using gaussian spatial high-pass filtering ($\sigma = 80 \mu m$), we removed scattered fluorescence and background activity. we then calculated the relative fluorescence changes by computing $\Delta F(t)/F_0$, where $F_0$ is the mean activity of each pixel over the entire session and $\Delta F(t)$ is the mean subtracted activity of each pixel at time $t$. Maximum projection images of each session's  $\Delta F(t)/F_0$ movie were analyzed to quantify lateral spatial displacements. We used \texttt{imregtform} in Matlab to determine optimal transformations between the first session's maximum projection image and subsequent sessions, aligning all videos accordingly. Aligned  $\Delta F(t)/F_0$ videos were concatenated for individual cell extraction and Ca$^{2+}$ activity trace analysis.

\subsubsection*{Cell extraction}
We successively applied principal and independent component analyses (PCA/ICA, \cite{mukamel2009automated}) to extract individual neuron activity from concatenated  $\Delta F(t)/F_0$ movies. Each mouse's preprocessed Ca$^{2+}$ video, approximately 1 TB in size, was divided into 16 tiles, each covering about 1 mm $\times$ 1 mm. We ran PCA/ICA in parallel on 16 computing nodes (1 tile per node) and identified neuron Ca$^{2+}$ activity traces and spatial filters. We then isolated cell somas by thresholding each spatial filter at 4 standard deviations of noise fluctuations and replacing filter weights below this threshold with zeros. The truncated spatial filters were reapplied to the  $\Delta F(t)/F_0$ movie to obtain final Ca$^{2+}$ activity traces.

For 3 of the 16 image tiles per mouse, individual neurons were manually identified based on their morphology and Ca$^{2+}$ transient waveforms. For the remaining 13 tiles, we trained three binary classifiers (Support Vector Machine, Linear Generalized Model, and Neural Network) using manually identified cells as training data. To train the classifiers we used a set of 12 pre-defined features to characterize a neuron’s morphology (eccentricity; diameter; area; orientation; perimeter; and solidity) and Ca$^{2+}$ activity waveform (mean peak amplitude of Ca$^{2+}$ transients; signal-to-noise ratio between Ca$^{2+}$ transients and baseline fluctuations; number of Ca$^{2+}$ transients peaks above 3 s.d. of the baseline; number of Ca$^{2+}$ transients peaks above 1 s.d. of the baseline; the difference of the mean decay and mean rise times of the Ca$^{2+}$ transients, normalized by the sum of these two values; and the FWHM of the average Ca$^{2+}$ transient). Classifiers were then used to identify cells in the 13 remaining tiles through a majority vote approach. Manual inspection verified that all cells identified met the criteria for neuron classification.

\subsubsection*{Descriptive details of the analysis data} 

As described in the main text, we used data from six mice performing a visual discrimination task. In this task (Fig. \ref{fig2}\textbf{B}), the animal was shown horizontal or vertical moving gratings as stimuli for $2$s. After a short delay of $0.5$s, the mice reported the type of the trial by either licking a spout or abstaining (Go-NoGo). For this work, when training our shared decoders, we concatenated the correctly performed trials across multiple imaging days for each animal. Since Ca$^{2+}$ traces collected from the one-photon mesoscope across days can have arbitrary units, we normalized the Ca$^{2+}$ activity traces to have unit power on each day. In total, we designed a massive dataset containing six mice across a total of $30$ imaging days (5, 5, 5, 5, 3, 7) and $7778$ correctly performed trials (1211, 998, 895, 1812, 1041, 1821). The dataset contained a total of $21570$ layer 2/3 pyramidal neurons (5292, 2761, 2236, 4193, 3334, 3754), collected from eight neocortical regions: 5249 neurons from primary visual cortex (V1; 970, 431, 803, 1085, 850, 1110), 1130 neurons from lateral visual cortex (LV; 144, 166, 14, 223, 120, 463), 1797 neurons from medial visual cortex (MV; 409, 123, 250, 335, 445, 235), 3240 neurons from posterior parietal cortex (PPC; 677, 427, 370, 573, 642, 551), 961 neurons from the auditory cortex (A; 106, 347, 0, 157, 18, 333), 7161 neurons from the somatosensory cortex (S; 2297, 1185, 223, 1717, 736, 1003), 324 neurons from the motor cortex (M;  272, 51, 0, 0, 1, 0), and 1708 neurons from the retrosplenial cortex (RSC; 417, 31, 576, 103, 522, 59). For the experiments comparing linear decoders with nonlinear decoders, we enforced that a brain region should at least have $50$ neurons. For the experiments computing the Fisher information of linear decoders, we enforced that a brain region should at least have $10$ neurons.

\subsubsection*{Analysis with the linear and nonlinear decoders} 

For the decoding analysis, we used the partial-least squares regression, principle component analysis, linear/quadratic discriminant analysis functions, random forest classifiers, and the (cross-validated) logistic regression implemented in Python's \texttt{scikit-learn} package \cite{scikit-learn}. We computed the Fisher information $(d')^2$ for discriminating HIT and CR trials by first dividing the trials randomly into three equal sets, following the procedure in \cite{ebrahimi2022emergent}. For the decoding analysis, we used four distinct windows: stimulus window included the mean of the neural activities inside $[0.5,2]$s time bins; the delay window $[2,2.5]$s, the response window $[2.5,5.5]$s; and the trial end window $[5.5,6.5]$s. To regularize these methods, we first trained PLS (with respect to trial identity) using the data from the first set and picked the first $K$ dimensions for dimensional reduction. Then, to compute the Fisher information, we trained a linear discriminant to compute the optimal decision direction, $w$, in the dimensionally reduced space. Then, the Fisher information is computed by projecting the neural activities in the third set via $w$ and computing the squared discriminability index, $(d')^2 = \Delta \mu^2/\sigma^2$. Here, $\Delta \mu$ is the difference in projected mean neural activities between HIT and CR trials and $\sigma^2$ is the pooled standard deviation within classes. Similarly, we computed the test accuracies for linear and nonlinear decoders using the three set division methodology, training the decoders on the second dimensionally reduced set and testing them on the third one.

\subsection*{\color{subsectioncolor} Evaluations of the neural code dimensionality}

In all experiments involving the dimensionality estimation, we defined the linear dimensionality as the number of principal components (PCs) required to explain 99\% of the variance in the neural activities.

\subsubsection*{Low-rank randomly initialized RNNs}

We considered full-rank and low-rank RNNs with sustained neural activities. To obtain chaotic full-rank RNNs, we randomly initialized RNNs with $N$ recurrent units and no external inputs. The entries of the recurrent weight matrix were sampled from a Gaussian distribution, after which the diagonal elements were set to zero:
\begin{equation}\label{eq:chaotic_rnn}
    W^{\rm rec} \sim \mathcal{N}(0,g^2/N), \quad W^{\rm rec}_{ii} = 0 \quad \forall i.
\end{equation}
The dynamics of these full-rank RNNs exhibit a transition to chaos at a critical value of $g = g^{\rm crit}_N$; in the mean-field limit ($N \to \infty$), $g^{\rm crit}_{\infty}=1$ \cite{sompolinsky1988chaos}. 

To obtain the low-rank counterpart, we initialized a full-rank RNN as in Eq.~(\ref{eq:chaotic_rnn}), applied a rank-reduction operation based on the singular value decomposition of $W^{\rm rec}$, and re-scaled the entries to obtain:
\begin{equation}\label{eq:low_rank_chaotic_rnn}
    W^{\rm rec}_K = \frac{\textrm{std}(W^{\rm rec})}{\textrm{std}(W_K)}W_K, \quad W_K = U\Sigma V^T \quad \text{where} \quad \Sigma_{ii} = 0 \quad \forall i > K,
\end{equation}
where $U$ and $V$ contain the left and right singular vectors of $W^{\rm rec}$, respectively. The dynamics of these low-rank RNNs exhibited a divergent trajectories at a critical value $g = g^{\rm crit}_{N,K}$ (Fig. \ref{figs3}\textbf{F-G}). For $g < g^{\rm crit}_{N,K}$ the dynamics decayed to zero, while for $g > g^{\rm crit}_{N,K}$ there was self-sustained activity in the network, marked by the exponential separation of neural trajectories with nearby initial conditions (perturbation magnitude of $10^{-5}$). As shown in Fig.~\ref{figs3}\textbf{F}, we also observed that $g^{\rm crit}_{N,K} \geq g^{\rm crit}_{N}$ for $K \leq N$ for the network parameters tested. On the other hand, it is important to note that these networks do not necessarily sustain chaotic dynamics, as perturbations in dimensions orthogonal to the $K$-dimensional encoding subspace would decay exponentially. Future work should study the chaotic properties of these networks in more details.

\subsubsection*{Sequence sorting task}
For the experiments involving the sequence sorting task, we trained a set of low-rank RNNs on sequence lengths from 2 to 4. The number of neurons per network ranged from 16, 64, 256 or 1024, whereas the ranks were between 1 and 6. We trained using the Adam optimizer with weight decay set to $10^{-4}$ and a learning rate of $10^{-3} \frac{\sqrt{B}}{16}$, where $B=8192$ is the batch size. We trained for 10000 steps with random data draws for each batch and repeated the whole experiment 10 times with different random seeds, leading to 10 distinct networks per condition. We chose $\Delta t = \tau$ and added a random noise sampled from a normal distribution with zero mean and $0.1$ s.d. at each time step.

\subsubsection*{Embedding of the $K$-dimensional unit sphere}
For the experiments involving static embedding maps, we sampled uniform samples from the $K$-dimensional unit sphere by first sampling a vector $\kappa \in \mathbb R^K$, with each element sampled from standard normal distribution. Then, we normalize the points by their norm, leading to samples on the  $K$-dimensional unit sphere. Since the resulting vector has independent units, it has rotational symmetry around the origin, leading to uniformly sampled data points. We used several nonlinearities (tanh, sign, ReLu, and sin) and varying levels of embedding strengths, $\delta m^{(i)}$ to embed the $K$-dimensional unit spheres into the $N_{\rm rec}$ dimensional neural activity spaces. 

\subsection*{\color{subsectioncolor}Effective and ineffective perturbation analyses}

For the perturbation experiments in Figure \ref{fig3}, we used the example RNN with $100$ neurons from Fig. \ref{fig1}, which was trained to perform a 2-bit flip flop task. We generated $100$ new neurons by sampling embedding weights from a discrete probability distribution defined on $\{-1,0,1\}$, with probabilities $\{0.1,0.8,0.1\}$ such that roughly $20$ neurons would be tuning for the flip flop states. The encoding weights for these neurons were zero. We refer to the first and second groups as ``coding" and ``redundant" neurons, respectively. 

We generated the weight matrix subserving the LPU via:
\begin{equation}
    W^{\rm latent} = \sum_{i=1}^{2} m^{(i)} n^{(i)T}.
\end{equation}
The total weight matrix was the summation of this structured component plus a randomly sampled component from a Gaussian distribution with zero mean and s.d. of $0.7/\sqrt{200}$:
\begin{equation}
     W^{\rm rec} =  W^{\rm latent} +  W^{\rm random}.
\end{equation}
The random connections allowed the newly generated neurons to influence the first group, thus preventing a simple interpretation of the perturbation results being due to lack of connectivity between coding and redundant neurons. We used the latent encoding weights to define the readout from the coding neurons. For the redundant neurons, we trained a linear regression on a sample of $1000$ data points. 

For the optogenetics manipulations, we computed two distinct vectors by taking the signs of the embedding weights of the first latent variable for the coding and redundant neurons, respectively. These vectors had zero values for one or the other groups of neurons, and were added as new inputs to the RNN to allow access to distinct groups. For the first experiment in Fig. \ref{fig3}\textbf{C}, we initialized the new network at the state (0,1), and then applied a single pulse (pulse value = $-1$) at [500,505]ms interval to the coding neurons. For the second experiment in Fig. \ref{fig3}\textbf{C}, we initialized the new network at the state (0,1), and then applied multiple pulses (pulse value = $-1$) at [500,600]ms interval to the redundant neurons. Both pulses were designed to switch the state to (0,0).

\subsection*{\color{subsectioncolor}Evaluations of representational drift in task-trained RNNs}

For the drift experiments in the main text, we used several networks trained to perform 3-bit flip flop tasks. Specifically, we used the $50$ networks with $100$ neurons using the two-step training paradigm described above. Since training RNNs with many neurons is computationally expensive, we used an alternative resampling approach to simulate RNNs with more than $100$ neurons. Specifically, each neuron in the RNN is unique defined by $3K$ parameters $\{m_i^{(1)},m_i^{(2)},m_i^{(3)},n_i^{(1)},n_i^{(2)},n_i^{(3)},(W^{\rm in})_{i1},(W^{\rm in})_{i2},(W^{\rm in})_{i3}\}$. We artificially created new neurons by resampling from this set and later normalizing the encoding weight values $n_i^{(p)} \to n^{(p)}_i \times  100 / N_{\rm rec}$.  Since adding new neurons and then normalizing their contribution to the latent variables lead to the same dynamical system equations (See Eq. (\ref{eq:lfrnn_latent_circuit})), this approach for increasing the neuron numbers exactly maintains the original LPUs.

\subsubsection*{Simulating the drift}
We evolved each network for $2000$ data points, after which we performed an instantaneous update in embedding weights, $m^{(p)} := m^{(p)} + \Delta m^{(p)}$ for $p=1,2,3$, and allowed the network to run for another $4000$ time points after the drift. To test the robustness of the networks to different types of representational drift, we picked the changes in $\Delta m^{(p)}$ in three distinct manners. First, we sampled $\Delta m^{(p)}$ from a random Gaussian with zero mean and $g$. This constituted the case of the fully random drift. For the targeted drifts, we projected the random weights to the directions parallel or orthogonal to the encoding weight space, spanned by $n^{(p)}$. To account for changes in power due to projecting, we re-scaled $\Delta m^{(p)}$ to have the same standard deviation as the one before the projection.

\subsubsection*{Evaluating the accuracy} 
We estimated the latent variables by either using the full set of neurons, or by using only the first $100$ neurons. Then, we computed the internal states of the RNN by thresholding the estimated latent variables. These latent states were compared with the outputs. If all three states were correct, the accuracy was one for that time point; otherwise one. For each network, we performed three distinct runs using a randomly chosen initial conditions and calculated mean accuracies across the three runs. When calculating these mean accuracies, we used only the last $2000$ time points, allowing the transient responses right after the drifts to settle.

\subsubsection*{Additional details on Figure \ref{fig4}}

For the illustrative tuning curve example in Fig. \ref{fig4}\textbf{B}, we reused the example RNN from Fig. \ref{fig1}. The applied drift had a strength of $g_{\rm drift} = 0.2$ and was orthogonal to the encoding subspace. We used the same network for the illustrations in \textbf{E-H}, where we evolved the RNN with an input noise drawn from a Gaussian distribution of mean zero and s.d. $10^{-2}$. The noise was added independently at each time point to the current term before the evaluation of the nonlinearity. For the RNNs solving the 3-bit flip flop task, which we reported in Fig. \ref{fig4}\textbf{D, I-J} in batch experiments, we set the standard deviation of the noise to $10^{-1}$. To obtain the fitted curves in these panels, we fit a sigmoid function between $1$ and $1/8$ (chance level) for the accuracy vs the logarithm of the drift strength values. $g_{\rm half}$ values correspond to the transition points of the fitted sigmoid curves, in which the sigmoid terms halve.

\clearpage
\newpage

\section*{Figures and Captions}

\begin{figure}[!t]
    \centering
    \includegraphics[width=\linewidth]{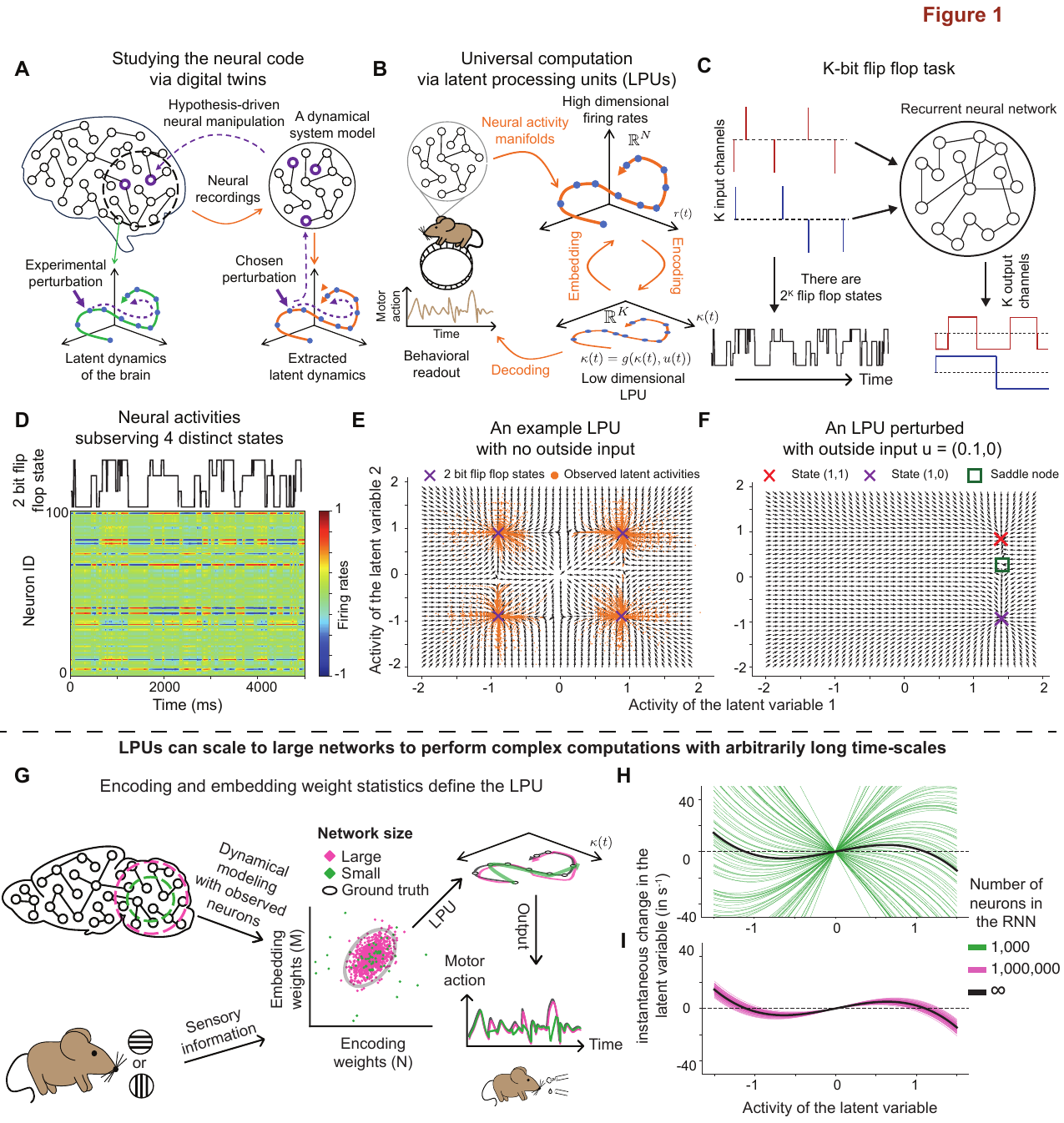}
\end{figure}

\clearpage
\newpage

\captionof{figure}{\textbf{LPUs enable universal computation despite the limitations of individual neurons.}} \label{fig1}

\textbf{(A)} A schematic of a model-driven manipulation experiment designed to test the neural manifold hypothesis. The observed neural dynamics are replicated by an interpretable dynamical model, which suggests potential single-cell manipulation experiments. This model is iteratively refined to predict the brain's responses to novel perturbations. The dimensionality of the computational variables in the model helps determine whether observed neural activities reside on low-dimensional curved manifolds, as predicted by the neural manifold hypothesis. \\ 

\textbf{(B)} An overview of the latent computation framework (LCF) formalizing the connections between animal behavior, observed neural activities, and population-level coding variables. In LCF, latent processing units (LPUs) can be accessed via linear encoding maps from neural activities, subserve arbitrarily complex computations, and jointly support behavior through linear readouts while driving observed neural dynamics via nonlinear embedding maps. \\ 

\textbf{(C)} To demonstrate the utility of LPUs with a simulated example, we trained an RNN to perform a 2-bit flip-flop task. The RNN learns four attractive fixed points, enabling it to maintain four internal states that correspond to the flip-flop combinations. \\ 

\textbf{(D)} Consistent with the traditional view of neural activity, neurons exhibit preferential coding for specific stimuli and/or internal states. Although latent variables can explain tuning properties, the converse direction, \textit{i.e.}, the tuning curve perspective alone, lacks causal predictions linking behavior with neural activities. \\

\textbf{(E, F)} Illustrations of the LPU belonging to the RNN in Panel \textbf{(D)}. \\ 

\textbf{(E)} The LPU without any targeted external input, where four attractive fixed points correspond to the four flip-flop states. Small random perturbations around these points decay back to the original state. \\

\textbf{(F)} The LPU under a directed perturbation aligned with the input weights of the first channel. In this case, the states (0,0) and (0,1) disappear, forcing the network to settle into either state (1,0) or (1,1), depending on the flip-flop state before the perturbation. \\

In both panels, \textbf{(E,F)}, arrows indicate the normalized direction of state updates, orange dots represent the network's current states, crosses mark attractive fixed points, and the rectangle denotes the saddle node that separates the two remaining states following the perturbation. \\

\textbf{(G)} In LCF, the population statistics of the encoding and embedding weights, hence the synaptic connections between neurons, define the LPU. In this view, the collective sum of synaptic connections, not any particular one, is responsible for the circuit operation. The goal of the trained RNN is not just to match motor actions, but also to reproduce the observed neural activities. \\

\textbf{(H, I)} Although individual neurons have decay times of $\tau = 1$ ms, latent variables can be designed to evolve at timescales that are several orders of magnitude slower. The plots demonstrate the latent dynamics obtained by drawing encoding and embedding weights from a predetermined two-dimensional Gaussian distribution that constructs the desired LPU as $N_{\rm rec} \rightarrow \infty$ (black solid lines, \textbf{Methods}). Neural networks with one thousand neurons exhibit significant variability in their LPUs (\textbf{H}), while networks with one million neurons can robustly construct the desired LPU (\textbf{I}). 

\clearpage
\newpage

\begin{figure}
    \centering
    \includegraphics[width=\linewidth]{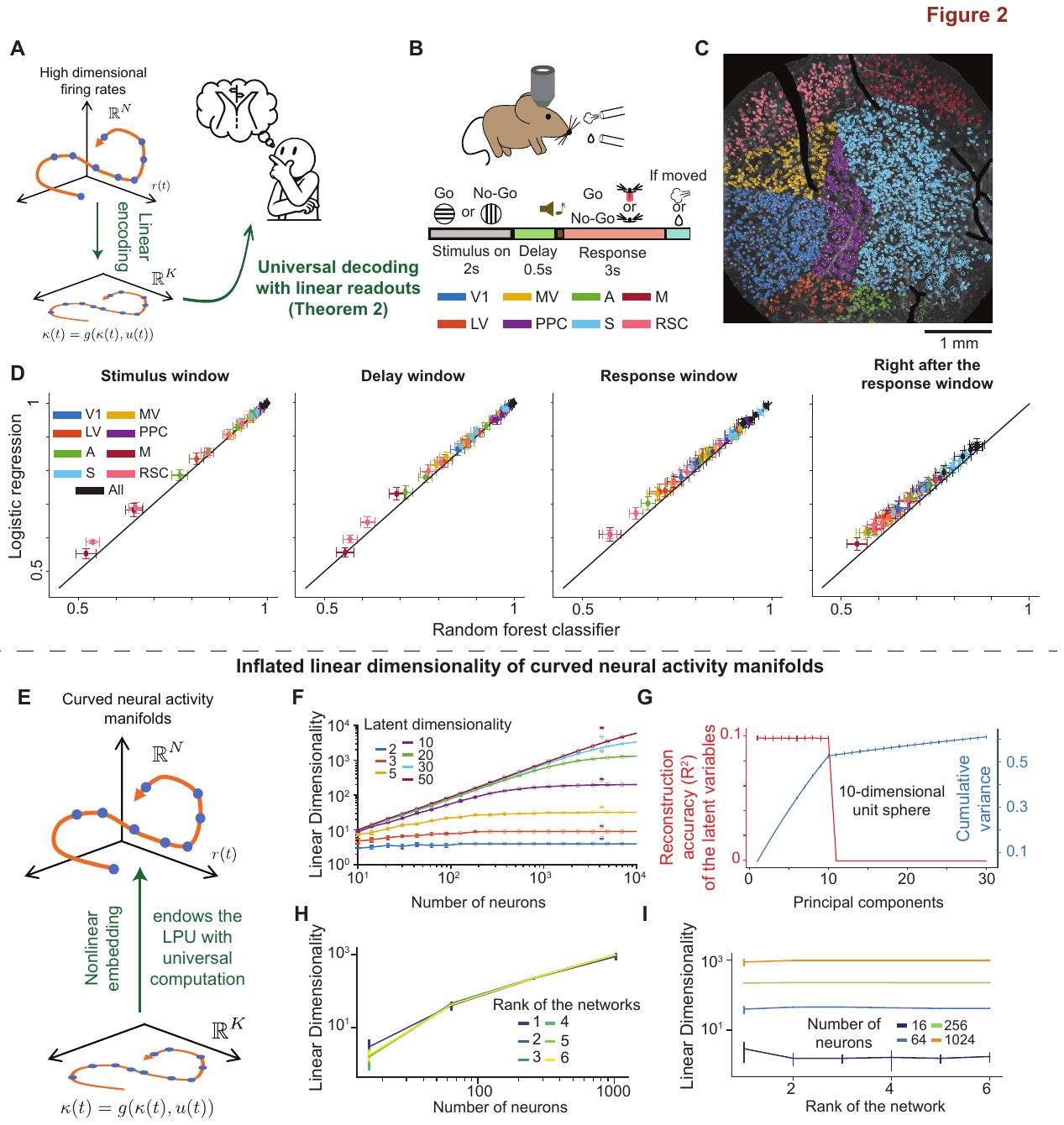}
    
\end{figure}

\clearpage
\newpage

\captionof{figure}{\textbf{Linear encoding and nonlinear embedding maps explain key experimental findings.}}\label{fig2}

\textbf{(A)} The linear encoding of LPUs enables direct decoding of animal behavior from neural activities, bypassing the need for explicit identification of latent variables within neural populations. \\

\textbf{(B, C)} To examine the sufficiency of linear decoders with low-dimensional bottlenecks, we re-analyzed previously published datasets of mice performing a visual discrimination task \cite{ebrahimi2022emergent}.  \\

\textbf{(B)}  A mouse performing a visual Go-NoGo task, where it must lick a spout or remain still based on the visual stimulus presented. Each trial consists of 2 seconds of stimulus presentation, followed by a 0.5-second delay and a 3-second response window.  \\

\textbf{(C)} An example maximum projection of the mouse cortex, where individual neurons are color-coded according to their respective cortical regions. Each dot represents a cell, with a total of 5,292 cells recorded in this session. Brain regions include: V1 (primary visual cortex), LV (lateral visual cortex), MV (medial visual cortex), PPC (posterior parietal cortex), A (auditory cortex), S (somatosensory cortex), M (motor cortex), and RSC (retrosplenial cortex). \\ 

\textbf{(D)} To assess the sufficiency of linear readouts, we performed decoding experiments to classify trial type (correct Go vs. NoGo), applying linear bottlenecks to regularize the dimensionality. These reduced dimensions (3 dimensions here, also see Fig. \ref{figs2}) were obtained using partial least squares (PLS) following the approach in \cite{ebrahimi2022emergent} (see \textbf{Methods}). We trained logistic regression (linear) and random forest (nonlinear) classifiers. In all cases, the linear decoders either matched or outperformed the nonlinear decoders, which tended to overfit. In each panel, dots represent the mean test accuracy for cortical neurons from a specific brain region in a single animal. Error bars indicate standard deviations across 100 training and testing splits. Panels include data from six mice and a total of 30 imaging sessions. \\

\textbf{(E)} Nonlinear embedding enables universal computation at the expense of inflated linear dimensionality. \\

\textbf{(F)} A \emph{static} embedding of a $K$-dimensional sphere via $r = \sin(\kappa)$ leads to orders of magnitude increases in the linear dimensionality of neural activities. Lines: means. Error bars: s.e.m. over five random initializations. \\

\textbf{(G)} A \emph{noiseless} embedding of a 10-dimensional unit sphere results in half the variance being in dimensions that contribute minimally to recovering the LPU. Lines: means. Error bars: s.e.m. over 100 random initializations.\\

\textbf{(H, I)} RNNs performing a sequence sorting task ($T=2$) showed increased linear dimensionality with added neurons \textbf{(H)}, despite having ranks within 1-6. The increase in linear dimensionality cannot be explained by the increase in rank \textbf{(I)}. Lines: means. Error bars: s.d. across 10 runs.

\clearpage
\newpage

\begin{figure}
    \centering
    \includegraphics[width=\linewidth]{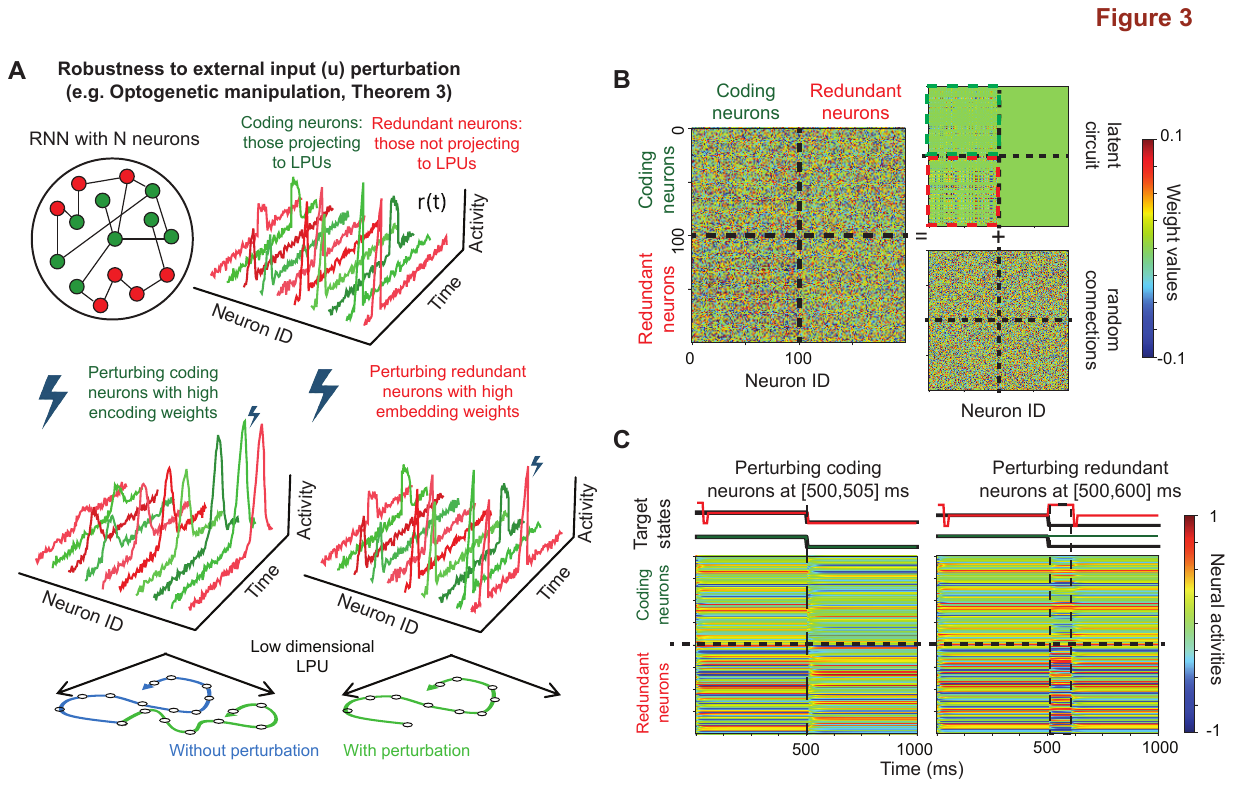}
\end{figure}
\clearpage
\newpage

\captionof{figure}{\textbf{A simple illustration of effective and ineffective perturbation directions in neural manifolds.} } \label{fig3}

\textbf{(A)} We define coding neurons for a given task as those that have non-zero latent projections for at least one of the latent variables of the circuit (\textbf{Methods}). Thanks to the distinction between (linear) encoding and (nonlinear) embedding maps, neurons that do not encode any of the latent variables can still hold redundant representations (\textbf{Theorem \ref{thm4}}). \\

\textbf{(B)} To illustrate how elusive the identities of coding and redundant neurons can be in practice, we conducted a simulation using the example RNN from Fig. \ref{fig1}\textbf{D-F} with 100 neurons, performing a 2-bit flip flop task. We added 100 additional neurons that did not code for the LPU but exhibited representations of the flip-flop states through non-zero embedding weights (\textbf{Methods}). To mask the stereotypical structure and bands of zero connectivities in the LPU, we introduced random projections with twice the standard deviation of the structured connections. The green rectangle marks connections between coding neurons, while the red rectangle highlights connections from coding to redundant neurons. \\

\textbf{(C)} While the encoding weights are not observable, the embedding weights are, indirectly, through tuning curves. We designed a realistic perturbation experiment based on observable tuning properties of individual neurons (\textbf{Methods}). The black lines denote the desired flip-flop states, the green lines represent states readout from coding neurons, and the red lines represent those readout from redundant neurons. \emph{Left:} Perturbing coding neurons even slightly led to the desired state change, readable from both coding and redundant neurons. \emph{Right:} Perturbing redundant neurons did not alter the flip-flop state, though the readout (incorrectly) changed for redundant neurons only. This experiment suggests that tuning properties alone do not necessarily indicate a neuron's contribution to the LPU.

\clearpage
\newpage

\begin{figure}
   \centering
\includegraphics[width=\linewidth]{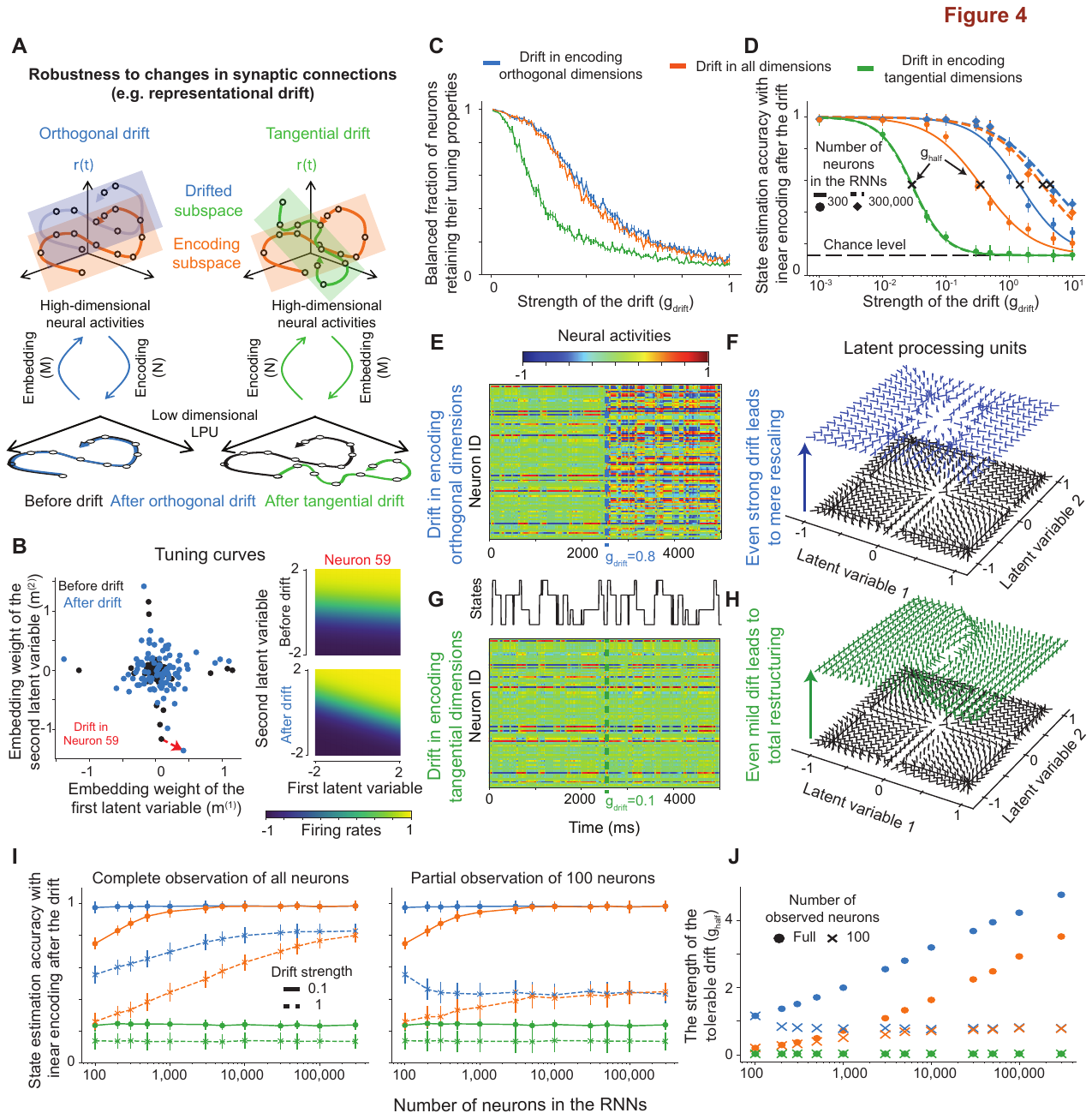}
\end{figure}

\clearpage
\newpage

\captionof{figure}{\textbf{Drift-robust encoding of neural activities into LPUs.} }  \label{fig4} 

\textbf{(A, B)} Drift-robust encoding of LPUs enables RNNs to preserve stable neural representations even as perturbations to the embedding weights alter neural tuning curves.  \\

\textbf{(A)} We simulated the representational drift as changes to the embedding weights ($M$), which define the maps from coding latent variables to the neural activities. Our goal is to study how different types of drifts in these weights might affect the latent code. \\

\textbf{(B)} Tuning properties of individual neurons change as representational drift occurs. \emph{Left:} Representational drift is simulated by altering the embedding weights from the latent variables to the neural activities. \emph{Right:} The change in the embedding weights of neuron 59 shifts its tuning properties. Before the drift, the neuron selectively codes for the second latent variable (aligned with the second output), whereas after the drift, the tuning curve becomes mixed selective, representing both outputs jointly. \\

\textbf{(C, D)} In these panels, we conducted drift experiments on low-rank RNNs performing 3-bit flip-flop tasks. After allowing the RNNs to evolve, we introduced instantaneous changes to the embedding weights in three distinct manners: For fully random drift, changes were sampled from a zero-mean Gaussian distribution with a standard deviation of $g_{\rm drift}.$ For targeted drifts, the random weights were projected onto directions either parallel or orthogonal to the encoding subspace, with rescaling to maintain constant power. \\

\textbf{(C)} The drift experiment, in which an abrupt change is applied to the embedding weights, changes the coding properties of the neurons depending on the strength of the drift. \\

\textbf{(D)} Average state estimation accuracies after the three types of drift. Linearly encoded LPUs remained robust to drifts in dimensions orthogonal to the encoding subspace with as few as a few hundred neurons. In contrast, robustness to fully random drifts emerged with large-scale networks. In contrast, drifts along the encoding subspace consistently resulted in significant decreases in state estimation accuracy from latent variables, regardless of the number of neurons. Here, $g_{\rm half}$ refers to the drift strength at which the sigmoid fit halves (see \textbf{Methods}), defined as the ``tolerable drift" strength before state estimation accuracies rapidly decline. \\

\textbf{(E-H)} To investigate the mechanism behind the distinct drops in readout accuracies following tangential versus orthogonal drifts, we reanalyzed the example RNN with 100 neurons from Fig. \ref{fig1}, trained to perform the 2-bit flip flop task. To highlight the contrast, we applied a strong encoding orthogonal drift with $g_{\rm drift} = 0.8$ and a milder encoding tangential drift with $g_{\rm drift} = 0.1$. \\

\textbf{(E)} Neural activity overview before and after the encoding orthogonal drift. The drift caused substantial changes in the tuning properties of individual neurons, significantly altering their selectivity. \\

\textbf{(F)} The LPU before (black, the same as in Fig. \ref{fig1}\textbf{(E)}) and after (blue) the drift. Despite the changes in the tuning curves, the LPU structure remained largely intact, preserving the four attractive fixed points. Drift simply rescaled the latent flow map. \\

\textbf{(G)} Overview of the neural activities before and after the encoding tangential drift. The mild drift caused changes in the tuning properties of neurons already involved in coding but, unlike in panel \textbf{(E)}, did not introduce new coding neurons.  \\

\textbf{(H)} The LPU before (black, same as in panel \textbf{(F)}) and after (green) the drift. Despite mild changes in the tuning curves, the LPU underwent significant restructuring. After the drift, the LPU no longer has the four attractive fixed points necessary to store the flip-flop states. \\

\textbf{(I-J)} Quantification of the robustness to drift as a function of the number of neurons under full and partial observation. \\

\textbf{(I)} Same results as in panel \textbf{(D)}, but shown as a function of neuron count across different drift strengths. While networks become increasingly resistant to higher levels of drift as neuron numbers grow, this robustness may be masked when only a small fraction of neurons are observed. \\

\textbf{(J)} Tolerable drift strength as a function of the number of neurons in the RNN for the three types of drift. Robustness to fully random drift emerged only in networks with more than a few thousand neurons, while encoding orthogonal drifts could be mitigated with just a handful of neurons. Encoding tangential drift consistently impaired performance, regardless of the number of neurons. Moreover, when the number of observed neurons was limited to 100, qualitative robustness was still observed in the first two types of drift, though the quantitative values were significantly reduced. \\

\clearpage
\newpage

\renewcommand\thefigure{S\arabic{figure}}
\renewcommand\thetable{S\arabic{table}}
\renewcommand{\theHfigure}{S\arabic{figure}}
\setcounter{table}{0}
\setcounter{figure}{0}

\clearpage
\newpage

\begin{figure}
    \centering
    \includegraphics[width=\linewidth]{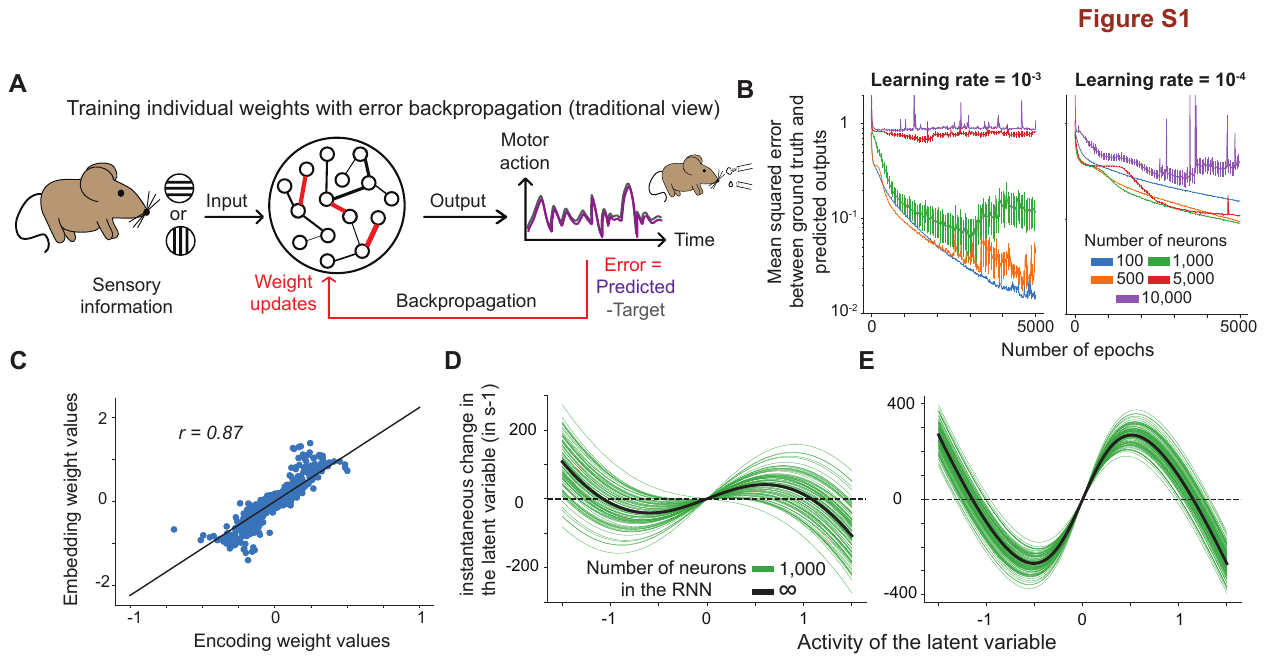}
\end{figure}

\clearpage
\newpage

\captionof{figure}{\textbf{Analyzing LPUs in large-scale RNNs capable of storing flip-flop states.}} \label{figs1}

\textbf{(A)} The traditional view of training artificial networks focuses on minimizing a loss function by propagating error back in time and space and updating the weights accordingly. In this view, the particular synaptic connections between individual neurons play significant roles. The neural activities are only relevant to the extent of how well the predictions match the target motor actions. \\

\textbf{(B)} We trained RNNs (with varying numbers of neurons) using backpropagation through time to perform the 3-bit flip-flop tasks. To train larger RNNs, lower learning rates were required, which inevitably increased the number of training epochs required to achieve low training errors. Trained RNNs were full-rank and no biologically motivated constraints were enforced to rule out alternative explanations of the learning deficiencies. Solid lines: means. Error bars: s.e.m. over 20 networks. \\

\textbf{(C)} To assess the synaptic structure of the LPUs in RNNs trained to store flip-flop states, we studied $50$ networks with $100$ neurons trained to perform the 3-bit flip-flop tasks. These RNNs were trained using the two-step paradigm described in \textbf{Methods}. The plot shows aggregated encoding, vs embedding weights across all networks and LPUs. Each dot corresponds to a single neuron. Two quantities showed significant correlations with each other ($p<0.001$).  \\

\textbf{(D-E)} As in Fig. \ref{fig1}\textbf{H}, but for LPUs with faster time-scales. RNNs with fewer neurons were able to faithfully model computations that had fast time scales on par with neuronal decay times (of 1 ms). \\ 

\clearpage
\newpage

\begin{figure}
    \centering
    \includegraphics[width=\linewidth]{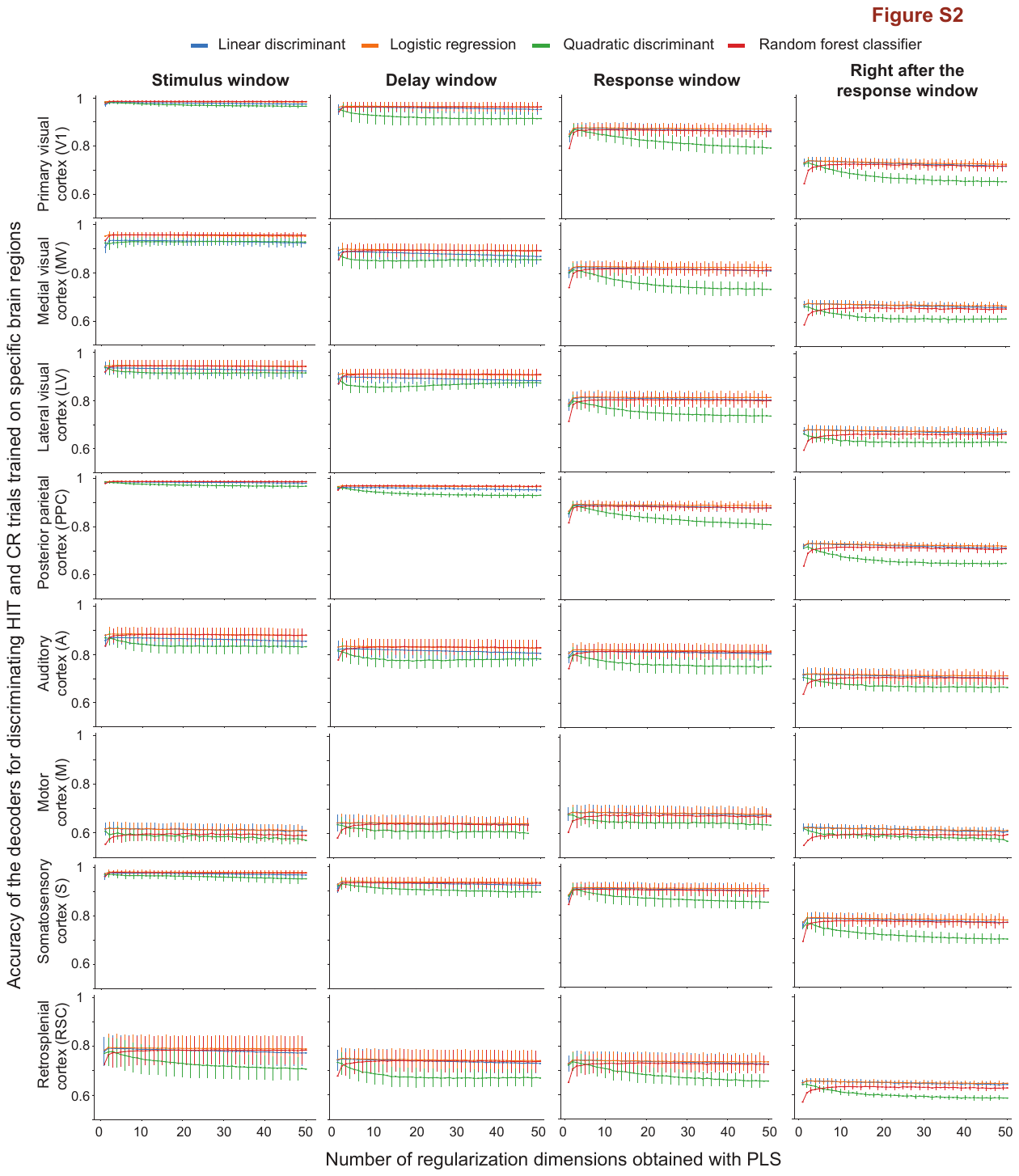}
\end{figure}

\clearpage
\newpage

\captionof{figure}{\textbf{Decoding of trial identity in eight neocortical regions through linear and nonlinear methods.}}
    \label{figs2}

Accuracy results for decoding HIT versus CR trials using various classifiers trained on neural data from different cortical regions, including the primary visual cortex (V1), medial visual cortex (MV), lateral visual cortex (LV), posterior parietal cortex (PPC), auditory cortex (A), motor cortex (M), somatosensory cortex (S), and retrosplenial cortex (RSC). Each classifier—linear discriminant, logistic regression, quadratic discriminant, and random forest—was trained with a varying number of regularization dimensions obtained using partial least squares (PLS). Results are shown for multiple time windows: stimulus, delay, response, and immediately after the response. Each plot indicates mean accuracy across 100 training and testing splits over six mice, with s.e.m. shown across the six mice.

\clearpage
\newpage

\begin{figure}
    \centering
    \includegraphics[width=\linewidth]{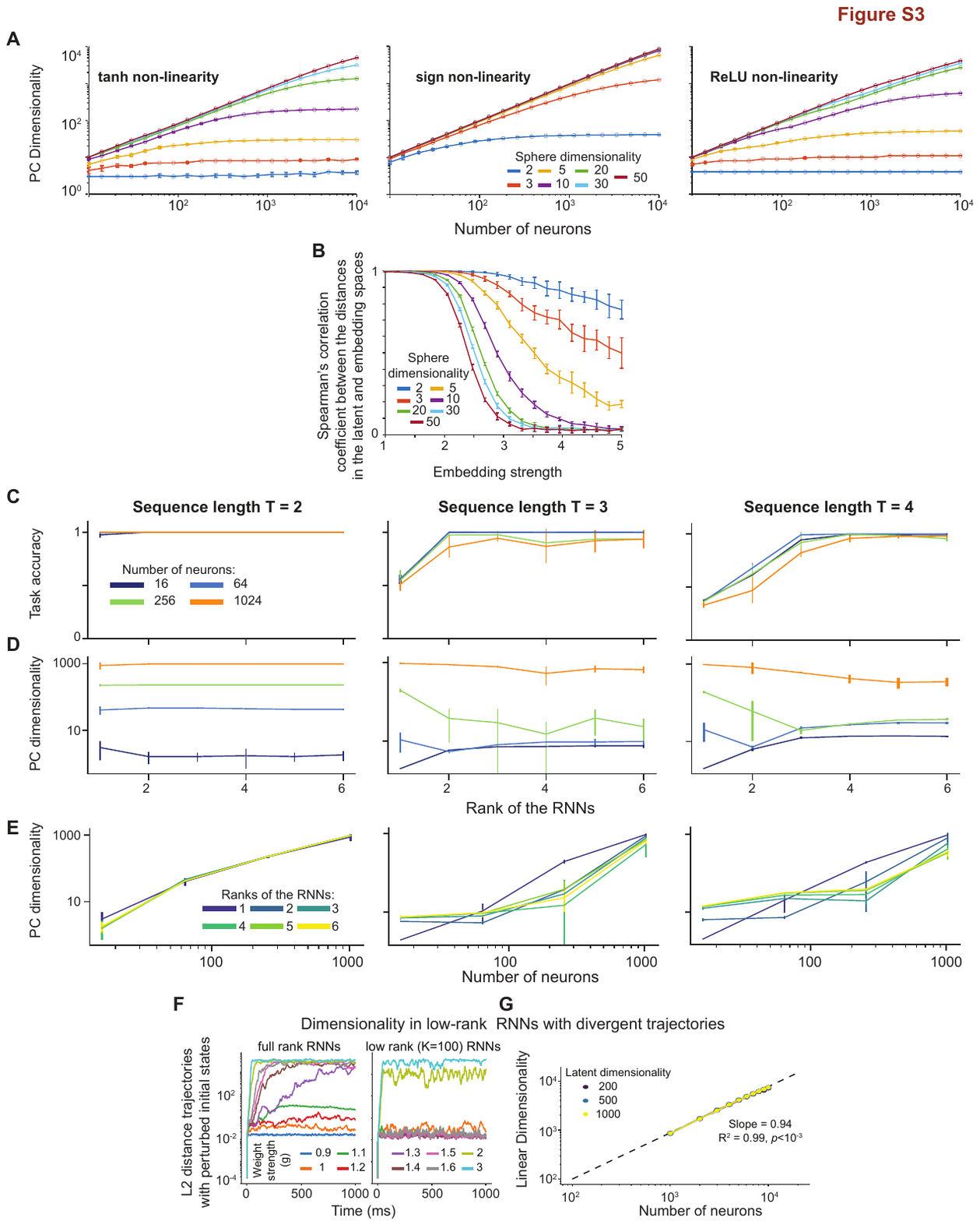}
\end{figure}

\clearpage
\newpage

\captionof{figure}{\textbf{Inflated linear dimensionality of neural manifolds in simulated networks.}} \label{figs3}

\textbf{(A)} We performed the analysis in Fig. \ref{fig2}\textbf{(F)} for three distinct nonlinearities: tanh, sign, and ReLU. \\ 

\textbf{(B)} The plot shows the Spearman’s correlation between the distances in the latent and embedding spaces. Increasing the embedding strength beyond a certain threshold resulted in embedding manifolds that failed to preserve the distances between points in the latent space. \\

\textbf{(C-E)} We varied the sequence length for RNNs performing a sequence sorting task and measured the resulting linear dimensionality. Points and bars represent the mean and standard deviation across 10 experimental repetitions with different random seeds. \\

\textbf{(C)} Task accuracy plateaued once the RNN rank matched the task dimensionality, $T - 1$, where $T$ is the sequence length of the sorting task. This plateau occurred independently of neuron count. \\

\textbf{(D)} The extrinsic dimensionality of the neural manifold remained constant across increasing RNN ranks, unaffected by the task dimensionality $T - 1$. \\

\textbf{(E)} The extrinsic dimensionality of the neural manifold increased linearly with the number of neurons, independent of both RNN rank and task dimensionality. \\

\textbf{(F)} When connections between neurons are randomly drawn from a Gaussian distribution with zero mean and s.d. $>\frac{1}{\sqrt{N_{\rm rec}}}$, the networks can sustain spontaneous activity and exhibit chaotic behavior. We designed a similar random sampling process for the weight matrices of low-rank RNNs. Even when only $K=100$ out of 1000 ranks of the weight matrix were occupied, the RNNs were able to sustain spontaneous, diverging neural trajectories. Each line represents a single network. \\

\textbf{(G)} Low-rank RNNs ($g=15$) show a nearly linear increase in their principal dimensionality with added neurons, displaying little sign of saturation. Line: mean. Error bars: s.d. across three runs.

\clearpage
\newpage

\begin{table}[!t]
\begin{center}
\caption{\label{tabs1} \textbf{An overview of existing research and relevance to this work.} Below, we provide a list of relevant work in order of their relevance and summarize how these approaches handle the questions covered by our framework. For theoretical calculations, please refer to the section ``\emph{Existing models in encoding-embedding framework}'' in \textbf{Methods}. \\}
\scriptsize
\begin{tabular}{ |p{1.5cm}|p{2.5cm}| p{2.5cm}| p{2.5cm}| p{3cm}|p{3cm} |} 
\hline
Description & Latent variables & Identifiability and behavioral readouts & Coding neural activity manifolds & Tuning curves and representational drift & Biological relevance   \\ \hline
Latent processing units & LPUs provide a causal model of latent dynamics. Theorem \ref{thm1} guarantees universal approximation capabilities, whereas Proposition \ref{prop1} introduces practical limits, \textit{i.e.}, large-scale populations are needed for computations with separation of timescales. & Theorem \ref{thm2} guarantees the optimality of linear readouts, and that they do not require any form of latent variable identification, thereby not suffering from identifiability issues. Latent variables themselves are linearly identifiable. & Proposition \ref{prop2} provides a positive example for how low-dimensional LPUs can lead to high-dimensional neural activities in non-trivial manifold geometries. Deviations to these coding manifolds decay exponentially fast (Theorem \ref{thm4}). &  LPUs can model non-trivial tuning curves thanks to the non-linearity of the embedding, which enables novel simulations of representational drift in RNNs. Theorem \ref{thm5} guarantees that not only behavioral readouts, but also the latent dynamics could remain robust to the drift. & LPUs are designed with biologically relevant motivations, and can explain several experimentally observed phenomena with no known contradictions.  \\ \hline
 Low-rank leaky current RNNs \cite{mastrogiuseppe2018linking,valente2022extracting,dubreuil2022role,beiran2021shaping}; also see \cite{Langdon2025} for a recent new direction & The first line of work to establish the connection between low-rank connectivity matrices and low-dimensional coding subspaces. In a particular RNN architecture, these works also introduce a causal model of latent dynamics. An equivalent of Theorem \ref{thm1} exists. & Linear readouts can circumvent identifiability issues, though no known results on their efficiency. Our Theorem \ref{thm2} could extend to this architecture if readouts were taken from the state variables $x$, but neural activities are nonlinear functions of stat variables, $\phi(x)$, in this architecture. Latent variables are linearly identifiable. & The coding subspaces lie on hyperplanes, \textit{i.e.}, flat manifolds. Similar to our Theorem \ref{thm4}, deviations to these hyperplanes decay exponentially. & Tuning curves are defined as linear maps from latent variables ($r = M \kappa$), where $M$ contains the left-singular vectors of the connectivity matrix (\textbf{Methods}). No known result on representational drift, and unclear whether our theorems cover these since tuning relationships are overly constrained. & The linearity constraint using left-singular vectors couples latent variable definition with tuning curves, which leads to several unrealistic restrictions: coding dimensions are hyperplanes, tuning curves are linear, implications for representational drift currently unknown. Assumptions contradict known results on non-linear tuning curves, and high-dimensional curved neural activity manifolds.  \\  \hline
 rSLDS framework  \cite{linderman2017bayesian,vinograd2024causal,hu2025modeling} & First line of work to prove the existence of line attractors in the brain with intervention experiments. rSLDS also provide a causal model of latent dynamics. No known result on universal approximation, though in the limit of infinitely many internal states, it may be possible to achieve universal computation. & Linear readouts can circumvent identifiability issues, though no known results on their efficiency. Latent variables are linearly identifiable. & The coding subspaces lie on hyperplanes, \textit{i.e.}, flat manifolds. & Linear tuning relationships, with no known discussion of representational drift. It is unclear whether our theorems cover these architectures.  & rSLDS framework often considers the learned networks as effective models of brain activities, rather than making claims on the biological processes. Similar to above, the linear map from latent variables significantly constraint the manifold geometry and tuning curves, though see \cite{karniolmodeling} and \cite{gao2016linear} for distinct directions with nonlinear readouts and/or dynamics. \\ \hline
Autoencoder based models \cite{schneider2023learnable,o2022direct,abbaspourazad2024dynamical} & Most models do not constraint latent dynamics, and some enforce simple dynamics such as linear maps. For the latter, it is not clear if the models have universal approximation properties. & In some cases, linear identifiability of latent \emph{variables} is achievable. Since encoding in these models is often nonlinear, identifiability or efficiency of linear readouts we have proven in this work do not extend to these models. & NA/Not a causal framework of neural coding & NA & NA \\ \hline
\end{tabular}
\end{center}
\end{table}

\clearpage
\newpage

\begin{table}[!t]
\begin{center}
\caption{\label{tabs2} \textbf{A summary of encoding and embedding types for dynamical systems.} For the theoretical derivations of the results in this table, See \textbf{Methods}. $^*$ANN stands for artificial neural networks. $^\dag$The leaky current RNNs fall within linear encoding and nonlinear embedding with additional constraints. }

\begin{tabular}{ |l|c|c| }
\hline
 & Nonlinear encoding $\phi$ & Linear encoding $\phi$ \\ \hline
\multirow{3}{*}{\parbox{2cm}{Nonlinear\\ embedding $\varphi$}} 
 & May require additional constraints  & Implemented by leaky current$^\dag$/firing rate RNNs  \\
 & Embedding and encoding function = ANNs$^*$  & $ \tau \dot  r(t)= - r(t) + \tanh\left(W^{\rm rec} r(t)  \right)$ \\
  & Can support rich computations  & Can support rich computations \\

 \hline
\multirow{3}{*}{\parbox{2cm}{Linear\\ embedding $\varphi$}}
 & May require constraints on & Implemented by linear systems \\
 & the nonlinearity  &  $\dot r(t) = A r(t)$ and $\dot \kappa(t) = B \kappa(t)$ \\
  &Can support rich computations & Cannot support rich computations \normalsize \\
 \hline
 
\end{tabular}
\end{center}
\end{table}

\clearpage
\newpage
\setlength{\parskip}{2ex} 
\setlength{\parindent}{1.5em}

\setcounter{equation}{0}
\setcounter{theorem}{0}
\setcounter{lemma}{0}
\setcounter{definition}{0}
\setcounter{corollary}{0}
\setcounter{section}{0}

\renewcommand\thesection{A\arabic{section}}
\renewcommand\theequation{A\arabic{equation}}
\renewcommand\thefigure{A\arabic{figure}}
\renewcommand\thetable{A\arabic{table}}
\renewcommand{\theHfigure}{A\arabic{figure}}
\setcounter{table}{0}
\setcounter{figure}{0}

\part*{\Large Supplementary Note: Constructing latent processing units in finite networks}

In this supplementary note, we illustrate some common problems associated with latent variable  \cite{schneider2023learnable,zimnik2024identifying}, especially in the case of partial observation of the networks. First, we discuss how to design one-dimensional LPUs that have bistable dynamics and derive constraints on the synaptic connections between neurons such that resulting dynamics can take place in target time-scales. Then. inspired by previous work that aims to model functional connections between neurons using data-constrained RNN models \cite{perich2020rethinking,perich2021inferring}, we design a simple, yet insightful, procedure for reconstructing LPUs with partial observation of neurons. Both analyses lead to analytical scaling relationships for the errors resulting from LPU (re-)construction, with two distinct empirically relevant predictions: i) a quantitative relationship for how many neurons are needed to accurately design the bistable LPU with desired time-scales, and ii) how well existing LPUs can be reconstructed, with particular focus on the contrast in the sparse ($N_{\rm obs} \ll N_{\rm rec}$) and rich ($N_{\rm obs} \approx N_{\rm rec}$) observation limits. Here, $N_{\rm obs}$ refers to the number of recorded neurons, which can become $\sim O(N_{\rm rec})$ with the recent large-scale recording technologies \cite{manley2024simultaneous}.

\section{Theoretical analysis of latent time-scales in a bistable circuit}

We start by writing the flow-map of a one-dimensional LPU without any inputs or biases:
\begin{equation} \label{supeq1}
   G_{N_{\rm rec}}(\kappa) := -\kappa + \frac{1}{N_{\rm rec}} \sum_{i=1}^{N_{\rm rec}} n_i \tanh(m_i \kappa) \xrightarrow{N_{\rm rec}\to \infty} -\kappa + E[n\tanh(m\kappa)]:=G(\kappa),
\end{equation}
where, for a given $\kappa$, $G_{N_{\rm rec}}(\kappa)$ approximates a Gaussian random variable with mean $G(\kappa)$ and some variance $\Sigma/N_{\rm rec}$ by the central limit theorem. Our goal in this section is to constrain the distribution of the pair $\{n,m\}$ such that the latent dynamical system approximates a bistable toy model in the mean-field limit, \textit{i.e.}, has a repelling fixed-point at the origin and two attractive fixed-points at $\kappa^* \pm 1$. Then, we will study the quantitative behavior and stability of this system as a function of finite $N_{\rm rec}$ neurons in the network.

\subsection{One-dimensional bistable toy model}

Since our interest is to design a bistable LPU (\textit{e.g.}, those in Fig. \ref{fig1}\textbf{G-I}), we start by discussing the behavior of the network around the origin, \textit{i.e.}, small $\kappa$ such that:
\begin{equation} \label{supeq2}
   G_{N_{\rm rec}}(\kappa) = -\kappa + \frac{1}{N_{\rm rec}} \sum_{i=1}^{N_{\rm rec}} n_i \tanh(m_i \kappa) \approx -  \left(1-\underbrace{\frac{1}{N_{\rm rec}} \sum_{i=1}^{N_{\rm rec}} n_i m_i }_{C(n,m)}\right) \kappa
\end{equation}
where $C(n,m)$ refers to the empirical covariance of the random variables $n$ and $m$ (which we assume to have zero mean for simplicity). By inspection, $E[C(n,m)] = E[nm]$, \textit{i.e.}, it is an unbiased estimator, and as $N_{\rm rec \to \infty}$, $C(n,m) \to E[nm]$ by the central limit theorem. Here, we focus on this limit for theoretical analysis, and later consider the asymptotical behavior for finite $N_{\rm rec}$. 

By inspection of the Eq. \eqref{supeq1}, the origin is  a fixed-point regardless of the distribution of the pair $\{n,m\}$. Then, inspecting the linearized dynamics in Eq. \eqref{supeq2} reveals that, for $E[nm] > 1$, the origin is a repeller fixed-point, whereas for $E[nm]<1$, it is an attractive fixed-point. For a bistable system, the origin is ought to be a repeller, hence we enforce that our distribution should satisfy the property $E[nm] > 1$. Then, with no additional constraint on the distribution, we can argue that at least one attractive fixed-point should be present for $\kappa > 0$ (and since the flow map in Eq. \eqref{supeq1} is anti-symmetric, similarly for $\kappa<0$). Specifically, for large $\kappa$, the second term in $G(\kappa)$ saturates such that $G(\kappa \to \infty) \to - \kappa + O(1)$. Yet, for $E[nm]>1$, $G(\kappa)>0$ for small but non-zero $\kappa >0$. Hence, due to the intermediate value theorem, there has to be at least one $\kappa^*$ such that $G(\kappa^*) = 0$. Moreover, we can also assert that at least for one such $\kappa^*$,  $G'(\kappa^*)<0$, since $G'(\kappa \to \infty) \to -1$ and $G'(\kappa = 0) > 0$. This statement can be proven by contradiction, \textit{i.e.}, if no such $\kappa^*$ exists, by the continuity of $G'(\kappa)$, $G(\kappa)$ cannot have a negative derivative for $\kappa>0$, which is a contradiction. Thus, we argued, with mild assumptions, that the origin becomes a repeller and there has to be an attractive fixed-point $\kappa^*$, which can be set $\kappa^* = 1$ without loss of generality by re-scaling the variables $n$ and $m$, respectively. This choice ensures that arbitrary transformations of $\kappa$, using the $GL_K(\mathbb R)$ symmetry, cannot trivially extend or shrink the desired time-scales of the LPUs.

Yet, this analysis does not constrain the number of $\kappa^*$ values for which $G(\kappa^*)=0$, which can be more than one for a general case. Fortunately, by noting the equivalency of Eq. \eqref{supeq1} to the latent dynamical system derived for the leaky current low-rank RNNs \cite{beiran2021shaping}, we can set a constraint ensuring that there is only one $\kappa^* > 0$ with the property $G(\kappa^*)=0$. Specifically, previous work has shown that if we constrain $\{n,m\}$ to be sampled from a Gaussian distribution, then this one-dimensional dynamical system can have two additional attractive fixed points \citep[Eq. (4.3)]{beiran2021shaping}. For this supplementary note, we constraint our design to Gaussian random variables, which leads to two options for the LPU: i) for $E[nm]<1$, there is only one fixed-point at the origin, or ii) for $E[nm]>1$, there is a repeller at the origin and attractive fixed points at $\pm \kappa^*$, \textit{i.e.}, the bistable system. Practically, after setting the scale on $\kappa^*=1$ (\textit{e.g.}, by rescaling $n$ and $m$ appropriately), there are two free parameters of the Gaussian distribution: i) the joint scale of the parameters $\sigma = \sigma_m \sigma_n$, and ii) the correlation, $\rho$, between the variables. 

\subsection{Designing the latent processing unit}

To design the LPU with a target time-scale, we now return to the linearization around $\kappa \approx 0$, which defines a local time-scale following:
\begin{equation}
    \tau \dot \kappa(t) = \kappa(t)  \left(1-E[nm] \right) \implies \frac{\tau}{\left(E[nm]-1 \right)} \dot \kappa(t) = \kappa(t).
\end{equation}
Here, we can define $\tau_{\rm latent} = \left|  \frac{\tau}{E[nm]-1}  \right|$ as the latent time-scale. Then, when $E[nm]<1$, the latent dynamics around the origin ($\kappa(0) = \kappa_0)$ decay following an exponential function $\kappa(t) = \kappa_0 \exp(-t/\tau_{\rm latent})$, whereas for $E[nm]>1$, they grow via $\kappa(t) = \kappa_0 \exp(t/\tau_{\rm latent})$. Hence, we can use the covariance, $E[nm] = \sigma_{n} \sigma_m \rho = \sigma \rho$, to design a bistable network with the following latent time-scale:
\begin{equation}
    \tau_{\rm latent} =  \frac{\tau}{\sigma \rho - 1},
\end{equation}
which satisfies the condition $\sigma \rho > 1$. By choosing $\sigma \rho \to 1^+$ or $\sigma \rho \to \infty$, it is theoretically possible to design a bistable circuit with a time-scale $\tau_{\rm latent} \in (0,\infty)$. However, how robust are these time-scales when there are finite number of neurons?

To study this question, we return to the finite-neuron limit. First, we note that in Fig. \ref{fig1}\textbf{G-I}, we observed a rather large variance in $\tau_{\rm latent}$ values, sometimes fluctuation between very large positive and negative values. Thus, for analytical convenience, we focus on the grow rates instead: 
\begin{equation} \label{eqsup_decay}
    \hat \gamma_{\rm latent} = \hat \tau_{\rm latent}^{-1} = \frac{C(n,m)-1}{\tau},
\end{equation}
which have continuous behavior between negative and positive values. Then, we compute the mean and the variance of the distribution:
\begin{equation}
    E[C(n,m)] = \sigma \rho  \quad \text{and} \quad  \text{Var}\left[C(n,m)\right] = \frac{1}{N_{\rm rec}} \text{Var}[n_i m_i ] = \frac{\sigma^2 (1+\rho^2)}{N_{\rm rec}}.
\end{equation}
With this, the normalized error in the growth rates of the designed circuits follow the relationship:
\begin{equation}
  \Delta_{\gamma} = \frac{\sqrt{\text{Var}[\hat \gamma_{\rm latent}]}}{E[\hat \gamma_{\rm latent}]} = \frac{\sigma \sqrt{1+\rho^2}}{\sigma \rho - 1} N^{-0.5}_{\rm rec} = \frac{\sigma \rho}{\sigma \rho - 1} \frac{\sqrt{1+\rho^2}}{\rho} N_{\rm rec}^{-0.5}
\end{equation}
After a change of variables, defining $\tau_{\rm latent} = \tau / (\sigma \rho - 1)$, we find:
\begin{equation}
    \frac{\sigma \rho}{\sigma \rho - 1} = \left(\frac{\sigma \rho-1}{\tau} + \frac{1}{\tau} \right) \tau_{\rm latent}  = \left( \frac{1}{\tau_{\rm latent}} + \frac{1}{\tau} \right) \tau_{\rm latent} = \frac{\tau_{\rm latent} + \tau}{\tau}.
\end{equation}
Plugging this back into the scaling equations leads to:
\begin{equation} \label{eqsup_sc1}
    \Delta_\gamma = \frac{\tau_{\rm latent} + \tau}{\tau}  \frac{\sqrt{1+\rho^2}}{\rho} N_{\rm rec}^{-0.5}.
\end{equation}
The scaling relationship already provides a clear prediction. The error is minimized for $\rho = 1$, \textit{i.e.}, when the variables are perfectly correlated. In reality, it is likely not possible to achieve this, since it requires precise control. Regardless, this term contributes as a pre-factor. Moreover, we are particularly interested in the limit $\tau_{\rm latent}\gg\tau$. Incorporating both observations, we achieve a scaling relationship of the form:
\begin{equation}
    \Delta_\gamma \propto \frac{\tau_{\rm latent}}{\tau } \frac{1}{\sqrt{N_{\rm rec}}}
\end{equation}
Intuitively, this means that to simulate dynamics in double time-scales with equivalent errors, the network should quadruple in size. Yet, biological neurons have $O(ms)$ time-scales, as opposed to second-long nature of many relevant behaviors. For example, decision-making processes often require extended periods of evidence accumulation before committing to one of two stable states (represented by attractive fixed points at $\kappa^* = \pm 1$). In such cases, the variable $\kappa$ must maintain an intermediate position near zero until sufficient evidence accumulates, rather than prematurely settling into either attractor state. Therefore, this scaling relationship quantifies the need for scale when designing LPUs with extended time-scales. 

\subsection{Latent processing unit reconstruction from partial observation}

The analysis above can be trivially extended to an empirically relevant consideration: LPU reconstructed from the partial observations. Specifically, the scaling relationship in Eq. \eqref{eqsup_sc1} applies to the scenario $N_{\rm obs} \ll N_{\rm rec} \to\infty$, where $N_{\rm obs}$ refers to the number of observed neurons. In this limit, resampling from the pairs of encoding and embedding weights, $\{n_i,m_i\}$ for $i=1,\ldots, N_{\rm rec}$, would be equivalent to designing a network with smaller number of neurons using the ground truth probability distribution. Below, we perform an extension of this analysis for a general $N_{\rm obs}$, \textit{i.e.}, even for when $N_{\rm obs} \approx N_{\rm rec}$. 

First, we introduce the reconstruction procedure for the LPUs from partial observation, but for an idealized limit that allows access to the $\{n,m\}$ pairs of the observed neurons. Specifically, we enforce that the functional connections, \textit{i.e.}, the pairs $\{n,m\}$, are perfectly estimated through some means (though often one focuses on reconstructing the dynamics, enforcing that the functional connections are properly reconstructed is one important aspect of how RNNs enable interpretable predictions when trained to reproduce brain recordings \cite{perich2021inferring}). Then, for a group of observed neurons and a one-dimensional LPU ($K=1$ in Eq. \eqref{eq_lowrank}), we reconstruct the recurrent weights by simply using their ground truth encoding and embedding values and replace the division by $N_{\rm rec}$ with a division by $N_{\rm observed}$ such that
\begin{equation}
    W_{\rm obs} = \frac{m n^T}{N_{\rm obs}},
\end{equation}
where $m$ and $n$ refer to the embedding and encoding weights of the one-dimensional LPU, respectively. If $N_{\rm obs}= N_{\rm rec}$, this estimation procedure reconstructs the ground truth weight matrix in Eq. \eqref{eq_lowrank}. More importantly, since $m$ and $n$ are sampled from the ground truth values, the reconstructed latent dynamical system:
\begin{equation}
    \tau \dot \kappa(t) = - \kappa(t) + \frac{1}{N_{\rm obs}} \sum_{i=1}^{N_{\rm obs}} n_i \tanh(m_i \kappa(t))
\end{equation}
on average reconstructs the time dynamics of the original LPU following:
\begin{equation}
    \frac{1}{N_{\rm obs}} \sum_{i=1}^{N_{\rm obs}} n_i \tanh(m_i \kappa(t)) \approx E_{n,m \sim P}[n \tanh(m \kappa(t))] \approx \frac{1}{N_{\rm rec}} \sum_{i=1}^{N_{\rm rec}} n_i \tanh(m_i \kappa(t)).
\end{equation}
In other words, this estimation process simply computes a sample mean, which has several desirable properties such as being unbiased and optimal under a Gaussian distribution of the underlying samples $\xi_i(t) =  n_i \tanh(m_i \kappa(t))$ (which is not true for a general scenario, yet sample mean remains a practical estimator). 

Clearly, such an estimation procedure is not feasible in an experimental scenario, since if we knew the ground truth weights and dynamical systems, why would we need data with partial observations? We collect the latter to gain insights to the former in the first place. In a more realistic scenario, the encoding-embedding weight pair for each observed neuron would need to be estimated from neural activities, which are themselves functions of the ground truth $\{n_i,m_i\}$ values. Hence, additional estimation procedures would make it unlikely to perfectly estimate the original pair. Yet, this idealized estimation procedure can illustrate a fundamental problem with reconstructing latent dynamics using subsampled neuronal populations even when the ground truth dynamical system equations are known: The reconstructed LPU at a particular reconstruction instance can have high variations from the ground truth if the neurons are extremely sparsely sampled. In other words, large-scale neural recordings are necessary to achieve minimal errors on reconstructed latent dynamical systems.

To demonstrate this, we now derive the errors between the estimated time-scales:
\begin{equation}
    \hat \gamma_{\rm rec}-\hat \gamma_{\rm obs} = \frac{\Delta C(n,m)}{\tau},
\end{equation}
where we define the variable:
\begin{equation}
    \Delta C(n,m) = \frac{1}{N_{\rm rec}} \sum_{i=1}^{N_{\rm rec}} n_i m_i - \frac{1}{N_{\rm obs}} \sum_{i=1}^{N_{\rm obs}} n_i m_i.
\end{equation}
We see that $E[\Delta C(n,m)] = 0$, \textit{i.e.}, $\hat \gamma_{\rm obs}$ is an unbiased estimator of $\hat \gamma_{\rm rec}$. Then, we can compute the variance as:
\begin{equation}
\begin{split}
    \text{Var}[\Delta C(n,m)]&= N_{\rm obs} \text{Var}\left[ n_i m_i \left(\frac{1}{N_{\rm rec}} - \frac{1}{N_{\rm obs}} \right) \right] + \frac{(N_{\rm rec}-N_{\rm obs})}{N_{\rm rec}^2} \text{Var}[ n_i m_i], \\
    &=\frac{ (N_{\rm rec} - N_{\rm obs})^2 + N_{\rm obs} (N_{\rm rec}-N_{\rm obs}) }{N_{\rm rec}^2 N_{\rm obs}} \sigma^2 (1+\rho^2), \\
    &= \frac{(N_{\rm rec} - N_{\rm obs}) }{N_{\rm rec} N_{\rm obs}} \sigma^2 (1+\rho^2)
\end{split}
\end{equation}

Similar to before, we are interested in the normalized error on the estimation:
\begin{equation}
    \hat \Delta_{\gamma}  = \frac{\sqrt{\text{Var}[ \hat \gamma_{\rm rec}-\hat \gamma_{\rm obs}]}}{E[\hat \gamma_{\rm rec}]} = \sqrt{\frac{N_{\rm rec}-N_{\rm obs}}{N_{\rm rec} N_{\rm obs}}}  \frac{\sqrt{1+\rho^2}}{\rho} \frac{\sigma\rho}{\rho \sigma - 1} = \sqrt{\frac{N_{\rm rec}-N_{\rm obs}}{N_{\rm rec} N_{\rm obs}}}  \frac{\sqrt{1+\rho^2}}{\rho} \frac{\tau_{\rm latent} + \tau}{\tau}
\end{equation}
After rearranging the terms, we obtain:
\begin{equation} \label{eqsup_sc2}
    \hat \Delta_{\gamma} = \frac{\tau_{\rm latent} + \tau}{\tau} \frac{\sqrt{1+\rho^2}}{\rho}  \sqrt{\frac{N_{\rm rec}-N_{\rm obs}}{N_{\rm rec} N_{\rm obs}}} \xrightarrow{N_{\rm rec} \to \infty} \frac{\tau_{\rm latent} + \tau}{\tau} \frac{\sqrt{1+\rho^2}}{\rho}  N_{\rm obs}^{-0.5},
\end{equation}
which is the same as in our calculation above in the limit $N_{\rm rec} \to \infty$, as expected. In other words, the reconstruction of the latent growth rates scales as $N_{\rm obs}^{-0.5}$ when only a sparse subset of neurons are observed. In the other limit, where $\Delta N = N_{\rm rec} - N_{\rm obs} \ll N_{\rm rec}$, \textit{i.e.}, a large population of neurons are observed, we obtain a different scaling:
\begin{equation}
   \hat \Delta_\gamma (N_{\rm obs} \approx N_{\rm rec}) \propto  \frac{\tau_{\rm latent}}{\tau} \frac{\sqrt{\Delta N}}{N_{\rm obs}},
\end{equation}
which means that the added neurons now rapidly decrease the reconstruction errors due to $\Delta N \to 0$ and $N_{\rm obs}^{-1}$ scaling.

\section{Empirical results with an idealized reconstruction procedure}

\begin{figure}[!t]
    \centering
    \includegraphics[width=\linewidth]{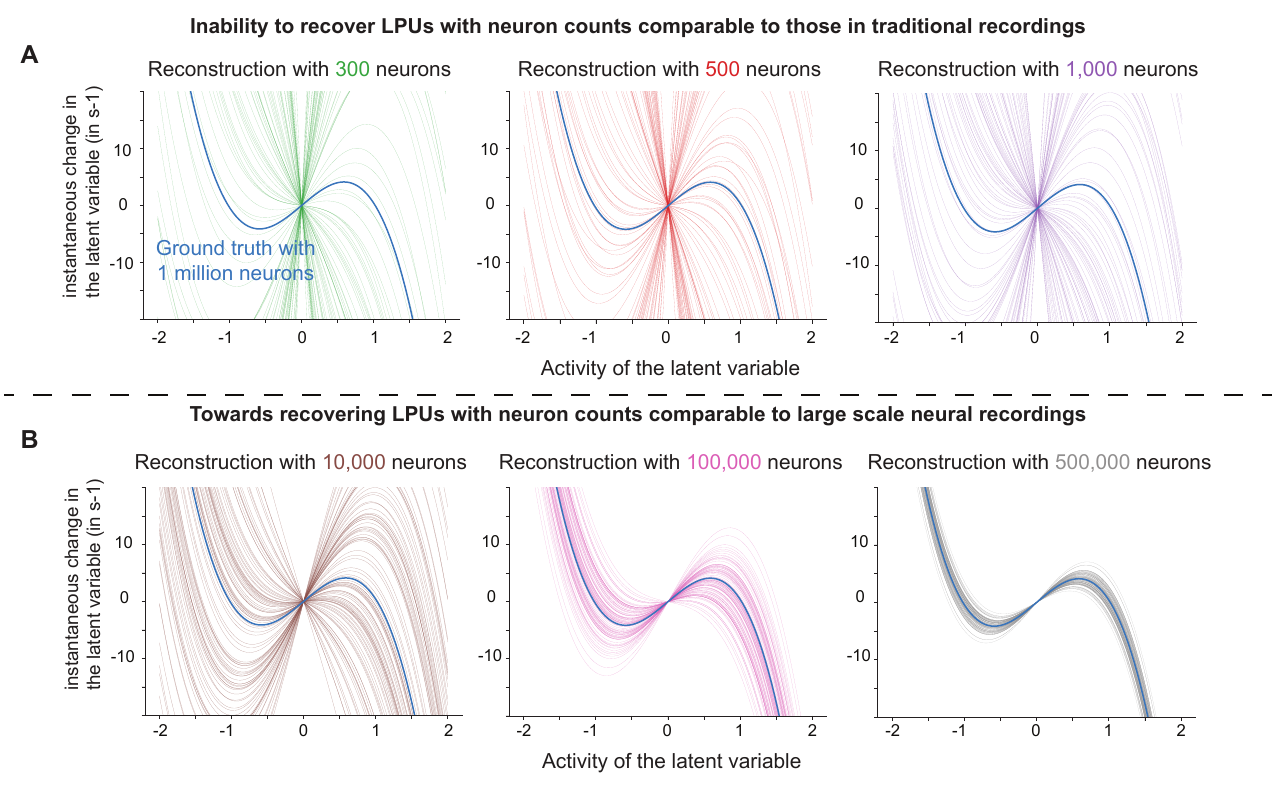}
    \caption{\textbf{Recovering LPUs with partial observations requires large-scale neural recordings.} We designed an RNN with a million neurons and a single LPU by sampling a million entries for encoding and embedding weights from a zero-mean and highly correlated 2D Gaussian distribution. We then reconstructed the LPUs with partial observation from the million neurons and appropriate rescaling of encoding weights (\textbf{Methods}). \textbf{(A)} Sparse observation of a few hundreds to thousands of neurons did not accurately reconstruct the true LPU. \textbf{(B)} Recording large populations, but still a fraction, of neurons led to accurate reconstruction of the true LPU. In panels \textbf{(A-B)}, each color-coded solid line corresponds to the LPU reconstructed with a randomly subsampled population. There are 100 such reconstructions per panel. Blue lines correspond to the ground truth LPU of the RNN with one million neurons.}
    \label{figa1}
\end{figure}

First, we designed an experiment to illustrate the errors inferred during LPU reconstruction from partial observation. Initially, we designed a LPU supported by one million neurons (using the procedure outlined in the previous section). Then, we attempted to reproduce this circuit by observing varying subsets of the original neuronal population (Fig. \ref{figa1}). The results qualitatively demonstrated that observing a small fraction of neurons led to substantial variations in the reconstructed LPUs (Fig. \ref{figa1}\textbf{A}), while accurate and low-variance reconstruction required observation of several percent of the original neurons (Fig. \ref{figa1}\textbf{B}). Importantly, while all reconstructed LPUs maintained their characteristic bistable form when approximately 10\% or more of the neurons were observed (Fig. \ref{figa1}\textbf{B}), reconstructions based on sparser observations could fail to capture even the qualitative structure of the LPU (Fig. \ref{figa1}\textbf{A}). Hence, we qualitatively confirmed our theoretical discussions above, \textit{i.e.}, reconstructing LPUs with second-long timescales requires observing large populations of neurons.

\begin{figure}[!t]
    \centering
    \includegraphics[width=0.8\linewidth]{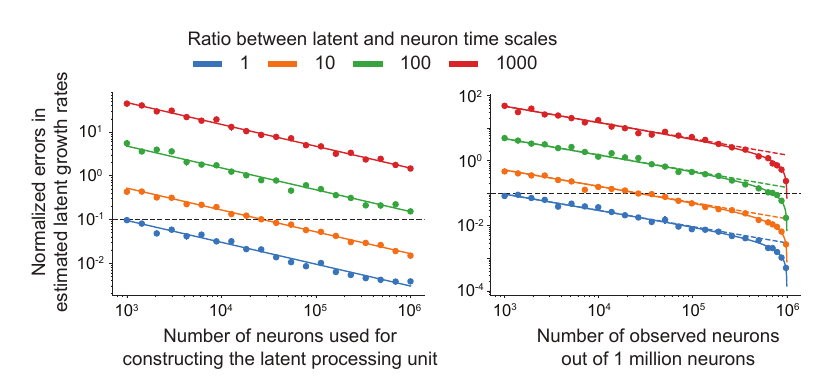}
    \caption{\textbf{Empirical verification of the scaling relationships with time-scales estimated from finite number of neurons.} To verify the scaling relationships in Eqs. \eqref{eqsup_sc1} and \eqref{eqsup_sc2}, we ran simulations by sampling $\{n,m\}$ pairs from a 2D Gaussian distributions. We computed the empirical growth rates following Eq. \eqref{eqsup_decay}, and computed the normalized errors predicted in Eqs. \eqref{eqsup_sc1} and \eqref{eqsup_sc2}. \emph{Left.} We computed the average deviations from the target latent growth rates as a function of number of neurons used to construct the LPU. Dots correspond to simulated values, solid lines the analytical scaling relationship. \emph{Right.} To test how well latent growth rates of a given LPU can be estimated, we first designed networks with $1,000,000$ neurons. Then, using subsampled neurons from the original group, we reconstructed the LPUs and computed their growth rates. Dots correspond to simulated values, solid lines correspond to the scaling relationships, and dotted lines correspond to the asymptotic $N_{\rm obs}^{-0.5}$ scaling described in the text. Parameters: For both cases, we set $\rho = 0.9$ and assumed $\tau = 1ms$. }
    \label{figa2}
\end{figure}

As the second experiment, we then set out to validate the quantitative scaling relationship we found in our analysis above. Specifically, we focused on the errors incurred regarding the latent growth rates, illustrated in Fig. \ref{figa2}. First, we designed an analytical circuit by first setting $\rho = 0.9$ and a given $\tau_{\rm latent}/\tau$ ratio, which constrained the remaining free parameter $\rho$ following:
\begin{equation}
    \sigma := \frac{1+\tau/\tau_{\rm talent} }{\rho}.
\end{equation}
With these parameters, we systematically varied the number of neurons used to construct the LPU and computed the deviation between the empirical (finite $N_{\rm rec}$) and target ($N_{\rm rec} \to \infty$) growth rates (Fig. \ref{figa2}, \emph{left}). The error scaled as $N_{\rm rec}^{-0.5}$, consistent with our theoretical predictions in Eq. \eqref{eqsup_sc1} (Note the agreement between simulation and analytical predictions). Importantly, we found that this relationship remained robust across different realizations of encoding and embedding weight distributions, confirming that the observed scaling behavior is not an artifact of a particular sampling configuration. 

Next, we performed a complementary experiment: instead of constructing the LPU in the $N_{\rm rec} \to \infty$ limit, we first designed a network (analogous to the brain) with finite number of neurons (1 million) and then reconstructed this particular instantiation of the LPU using only the ``observed'' neurons (Fig. \ref{figa2}, \emph{right}). Once again, we found that the simulations aligned with the empirical scaling relationship in Eq.~\eqref{eqsup_sc2}. As the number of observed neurons ($N_{\rm obs}$) increased, the deviation from the $N_{\rm obs}^{-0.5}$ scaling became more pronounced in a positive direction. Thus, increasing the number of observed neurons provided progressively greater improvements in the rate of error reduction, \textit{i.e.}, reconstructing latent dynamics benefits uniquely from large-scale neural recordings.

Overall, our findings demonstrate that large-scale neural recordings are essential for accurately estimating bistable dynamics, which represent one of the most fundamental phenomena in dynamical systems theory \cite{strogatz2018nonlinear}. Although recordings from a small subset of neurons may appear to capture the underlying dynamics, this example reveals how limited neuronal sampling introduces substantial variance that obscures the link between LPUs and the true functional connections (here, the $\{m,n\}$ pairs). This insight extends previous work suggesting that RNNs trained on neural recordings could capture functional connections \cite{perich2021inferring}, which may be studied to uncover inter-area communication patterns. Our analysis indicates that training such dynamical models requires, first and foremost, comprehensive large-scale recordings to achieve reliable results.

\end{document}